\documentclass[prodmode,acmtoplas]{acmsmall} 

\usepackage[ruled]{algorithm2e}

\SetAlFnt{\small}
\SetAlCapFnt{\small}
\SetAlCapNameFnt{\small}
\SetAlCapHSkip{0pt}
\IncMargin{-\parindent}

\usepackage{capt-of,setspace}
\usepackage{relsize,tipx}
\usepackage{fancyvrb}
\usepackage[table,dvipsnames]{xcolor}

\definecolor{pink}{RGB}{255, 192, 203}
\definecolor{yellow}{RGB}{255, 255, 224}
\definecolor{gray}{RGB}{224, 224, 224}

\usepackage{subfigure}
\usepackage{capt-of}
\usepackage{multirow}
\usepackage{graphicx}

\usepackage{array,xspace,listings}
\definecolor{keywordColor}{rgb}{0.5,0,0.35}
\definecolor{commentColor}{rgb}{0.25,0.5,0.35}
\definecolor{stringColor}{rgb}{0.06, 0.10, 0.98}
\definecolor{fieldColor}{rgb}{0.06, 0.10, 0.98}
\usepackage{amsmath, amssymb}
\usepackage{mathtools}
\usepackage{xcolor,xspace}
\usepackage[normalem]{ulem}
\usepackage{epstopdf}
\usepackage{color}
\usepackage{comment}

\newcommand{\stextmath}[1]{{\small $\mathcal{#1}$\normalsize}}
\newcommand{\cbox}[2]{{\setlength{\fboxsep}{0.3pt}\colorbox{#1}{#2}}}

\newcommand{\phd}{{\,\_}}
\newcommand{\fast}{{\sc Probe}\xspace}
\newcommand{\probe}{\fast}

\newcommand{\cs}{\emph{checkstyle}\xspace}

\newcommand{\todo}[1]{\textcolor{red}{#1}}
\newcommand{\ptr}{pt}
\newcommand{\lhm}{{\it lhm}}

\newcommand{\Soot}{{\sc Soot}\xspace}
\newcommand{\soot}{{\Soot}}
\newcommand{\Wala}{{\sc Wala}\xspace}
\newcommand{\Bdd}{{\sc Bddbddb}\xspace}
\newcommand{\bdd}{{\sc Bddbddb}}
\newcommand{\doop}{{\sc Doop}\xspace}
\newcommand{\Doop}{{\sc Doop}\xspace}
\newcommand{\doopr}{{\sc Doop$^r$}\xspace}
\newcommand{\Doopr}{{\sc Doop$^r$}\xspace}
\newcommand{\wala}{{\sc Wala}}
\newcommand{\chord}{{\sc Chord}\xspace}
\newcommand{\RefJava}{{\sc RefJava}\xspace}

\newcommand{\refjava}{\RefJava}
\newcommand{\refjavacname}{{\sc RefJava}$_{c}$\xspace}

\definecolor{mygreen}{rgb}{0.0, 0.5, 0.0}
\definecolor{myred}{rgb}{0.9, 0.17, 0.31}
\definecolor{myblue}{rgb}{0.0, 0.0, 0.8}

\usepackage{changepage}
\usepackage{wrapfig}
\usepackage{csquotes}
\usepackage{framed}

\usepackage{enumitem}
\usepackage{makecell}
\newcolumntype{I}{!{\vrule width 1pt}}
\newtheorem{assumption}{Assumption}

\usepackage{tikz}
\usetikzlibrary{automata}
\usetikzlibrary{calc}
\usetikzlibrary{positioning}
\usetikzlibrary{shapes}
\usetikzlibrary{arrows}
\usetikzlibrary{calendar}
\usetikzlibrary{decorations.pathmorphing}
\usetikzlibrary{backgrounds}
\usetikzlibrary{fit}

\usepackage{mathabx}
\usepackage{rotating} 
\usepackage{bigstrut} 

\newcommand{\ruledef}[2]{$\dfrac{\begin{array}[c]{c}#1\end{array}}{\begin{array}[c]{c}#2\end{array}}$}
	\newcommand{\rulename}[1]{\relsize{-0} {\color{red}[\textsc{#1}]}}
	\newcommand{\rulenameT}[1]{{\relsize{-1}{\color{red}[\textsc{#1}]}}}

\newcommand{\ibinding}[2]{#1<:#2}

\newcommand{\ifamily}[2]{#1\ll:#2}

\newcommand{\n}[1]{\textnormal{#1}}

\newcommand*\circled[1]{\tikz[baseline=(char.base)]{
            \node[shape=circle,draw,inner sep=0.5pt] (char) {\small #1};}}

\newcommand{\elf}{{\sc Elf}\xspace}
\newcommand{\solar}{{\sc Solar}\xspace}

\newcommand{\dam}{{\sc LHM}\xspace}
\newcommand{\TamiFlex}{{\sc TamiFlex}\xspace}

\newcommand{\calM}{{\cal M}}
\newcommand{\calF}{{\cal F}}

\begin{document}

\markboth{Y. Li et al.}{Understanding and Analyzing Java Reflection}

\title{Understanding and Analyzing Java Reflection}
\author{
YUE LI
\affil{UNSW, Australia}
TIAN TAN
\affil{UNSW, Australia}
JINGLING XUE
\affil{UNSW, Australia}
}

\begin{abstract}
Java reflection has been increasingly used in a wide
range of software. It allows a software system to 
inspect and/or modify the behaviour of its classes,
interfaces, methods and fields at runtime, enabling
the software to adapt to dynamically changing
runtime environments.  
However, this dynamic language feature imposes significant challenges to static analysis, because
the behaviour of reflection-rich software
is  logically complex and statically hard to predict, especially when manipulated frequently by
statically unknown string values. As a result, 
existing static analysis tools either ignore 
reflection or handle it partially, resulting in missed, important behaviours, i.e., unsound results. Therefore,
improving or even achieving
soundness in (static) reflection analysis---an 
analysis that infers statically the behaviour of
reflective code---will provide significant benefits to many analysis clients, such as bug detectors, security analyzers and program verifiers.

This paper makes two contributions: we provide
a comprehensive understanding of Java reflection through examining its underlying concept, API and real-world usage, and,
building on this, we introduce a new static approach to resolving Java reflection effectively in practice. We have implemented our reflection analysis
in an open-source tool, called \solar, and evaluated its effectiveness extensively with large Java programs 
and libraries. Our
experimental results demonstrate that \solar is able to (1) resolve reflection more soundly than the state-of-the-art reflection analysis; (2) automatically and accurately identify the parts of the program where reflection is resolved unsoundly or imprecisely; and (3) guide users to iteratively refine the analysis results by 
using lightweight annotations until their specific requirements are satisfied.

\end{abstract}

\category{F.3.2}{Semantics of Programming Languages}{Program Analysis}
\terms{Object-Oriented Programming Languages, Program Analysis}

\keywords{Static Analysis, Java Reflection, Pointer Analysis, Call Graph}

\maketitle

\section{Introduction}
\label{sec:introduction}

Java reflection allows a software system to
inspect and/or modify the behaviour of its classes,
interfaces, methods and fields at runtime, enabling
the software to adapt to dynamically changing
runtime environments. This dynamic language feature 
eases the development and maintenance of Java programs in many programming tasks by, for example, facilitating 
their flexible integration with the third-party code 
and their main behaviours to be configured 
according to a deployed runtime environment in a decoupled way. Due to such advantages, reflection has been widely 
used in a variety of Java applications and frameworks~\cite{Yue14,Zhauniarovich15}.

Static analysis is widely recognized as a fundamental
tool for bug detection~\cite{Dawson01,Mayur06}, security vulnerability analysis~\cite{Livshits05security,Bodden14}, compiler optimization~\cite{Dean95,Sui13}, program verification~\cite{Das02,Cousot03}, and program debugging and understanding~\cite{Manu07,Yue16}.
However, when applying static analysis to
Java programs, reflection poses a major 
obstacle~\cite{Livshits05,Yue14,Yue15,Yannis15}.
If the behavior of reflective code is not resolved
well, much of the codebase will be
rendered invisible for static analysis, resulting in 
missed, important behaviours, i.e., unsound analysis results~\cite{Livshits15}.
Therefore, improving or even achieving soundness in 
(static) reflection analysis---an analysis that
infers statically the behavior of reflective code---will provide significant benefits to all the client analyses
as just mentioned above.

\subsection{Challenges}
\label{intro:cha}
Developing effective reflection analysis for real-world
programs remains a hard problem, widely
acknowledged by the static analysis community~\cite{Livshits15}: 
\begin{quote}
``\emph{Reflection usage and the size of libraries/frameworks make it very difficult to scale points-to analysis to modern Java programs.}''~\cite{wala};
\end{quote}

\begin{quote}
``\emph{Reflection makes it difficult to analyze statically.}''~\cite{Rastogi13};
\end{quote}

\begin{quote}
``\emph{In our experience~\cite{Ernst13}, the largest challenge to analyzing Android apps is their use of reflection ...}''~\cite{Ernst15}
\end{quote}

\begin{quote}
``\emph{Static analysis of object-oriented code is an exciting, ongoing
and challenging research area, made especially challenging
by dynamic language features, a.k.a. reflection.}''~\cite{Barros17}
\end{quote}

There are three reasons on why it is hard to
untangle this knotty problem:
\begin{itemize}
\item The Java reflection API is large and its 
	common uses in Java programs are complex.
	It remains unclear how an analysis
	should focus on its effort on analyzing
	which of its reflection methods in order
	to achieve some analysis results as desired.
\item The dynamic behaviours of reflective calls are mainly specified by their string arguments, which are usually unknown statically (e.g., with some
	string values being encrypted, read from configuration files, or retrieved from the Internet).  
\item The reflective code in a Java program
	cannot be analyzed alone in isolation. To resolve
reflective calls adequately, a reflection analysis often works
	inter-dependently with a pointer analysis~\cite{Livshits05,Yue14,Yue15,Yannis15pta,Yannis15}, with each being both the producer and consumer of the other.
When some reflective calls are not yet resolved, 
the pointer information that is currently available can be over- or under-approximate. Care must be taken to ensure that the reflection analysis helps increase soundness (coverage) while still maintaining sufficient precision for the pointer analysis.
Otherwise, the combined analysis would be unscalable
for large programs. 
\end{itemize}

As a result, most of the papers on static analysis for object-oriented languages, like Java, treat
reflection orthogonally (often without even
mentioning its existence). Existing static
analysis tools either ignore reflection or handle it partially and ineffectively.

\subsection{Previous Approaches}
\label{intro:pre}

Initially,
reflection analysis mainly relies on string analysis, especially when the string arguments to reflective calls are string constants, to resolve reflective targets, i.e., methods or fields reflectively accessed.
Currently, this mainstream approach is still adopted by many static
analysis tools for Java, such as \soot, \wala, \chord and \doop. However, as described in Section~\ref{intro:cha}, 
string analysis will fail in many situations where string arguments are unknown, resulting in limited soundness and precision.
As a static analysis, a (more) sound reflection 
analysis is one that allows (more) true reflective 
targets (i.e., targets that are
actually accessed at runtime) to be
resolved statically.
In practice, any reflection analysis must inevitably 
make a trade-off among soundness, precision, scalability, and (sometimes) automation. 

In addition, existing reflection analyses~\cite{Livshits05,Yannis09,Yannis15,Ernst15,Lili16,Zhang17} cannot answer two critical questions that are raised naturally, in practice: Q(1) how sound is a given reflection analysis and Q(2) which reflective calls are resolved unsoundly or imprecisely? 
We argue for their importance as follows:
\begin{itemize}
\item
If Q(1) is unanswered, users would be unsure (or lose confidence) about the effectiveness of the analysis results
produced. For example, 
a bug detector  that reports no bugs may actually
miss many bugs if
some reflective calls are resolved unsoundly. 
\item 
If Q(2) is unanswered, users would not have an opportunity 
to contribute in improving the precision and soundness of the analysis results, e.g.,
by providing some user annotations. For some
client analyses (e.g., verification),
soundness is required.
\end{itemize}

\subsection{Contributions}
\label{intro:contri}
In this paper, we attempt to uncover the mysterious veil of Java reflection and change the informed opinion 
in the program analysis community about static
reflection analysis: ``\emph{Java reflection is a dynamic feature which is nearly impossible to handle effectively in static analysis}''. We 
make the following contributions:
\begin{itemize}
\item
We provide a comprehensive understanding of Java reflection through examining its underlying \emph{concept} (what it is), \emph{interface} (how its API is designed),
and \emph{real-world usage} (how it is used in practice). As a result, we will provide
the answers to several critical questions, which
are somewhat related, including:
\vspace{0.5ex}
	\begin{itemize}
	\item What is reflection, why is it introduced in programming languages, and how is Java reflection derived from the basic reflection concept?
\vspace{0.5ex}
\item Which methods of the Java reflection API
should be analyzed carefully and how are they
related, as the API is large and complex
	(with about 200 methods)?
\vspace{0.5ex}
	\item How is reflection used in real-world Java programs and what can we learn from its common uses? 
		We have conducted a comprehensive study about reflection usage in a set of 16 representative Java programs by examining their 1,423 reflective call sites. We report
		7 useful findings to enable the 
		development of improved practical reflection analysis techniques and tools
in future research.
	\end{itemize}
\vspace{1ex}

\begin{figure}
\centering
\includegraphics[width=0.93\textwidth]{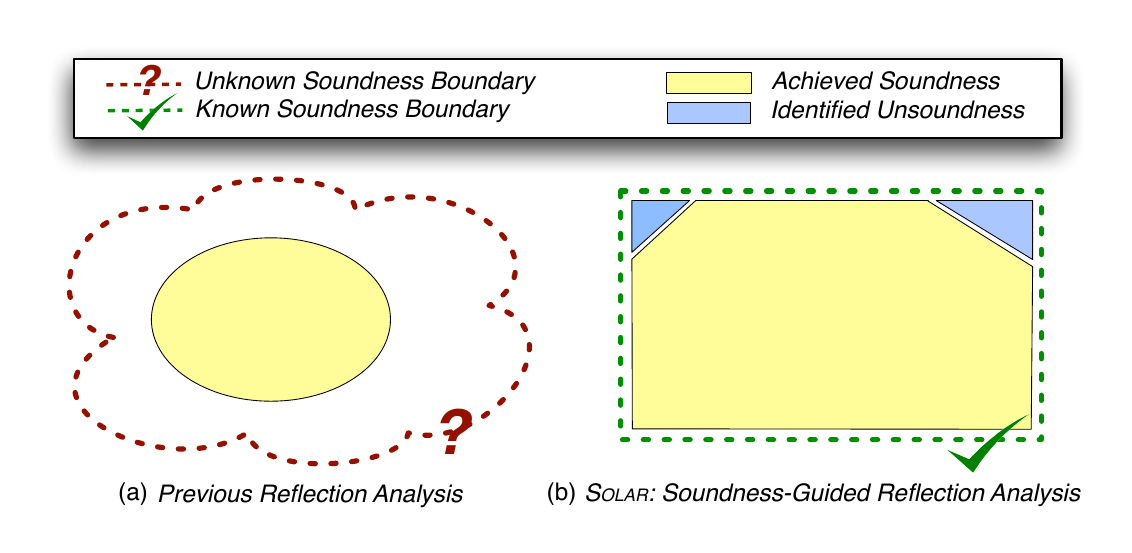}
\vspace{-3ex}
\caption{Reflection analysis: prior work vs. \solar.}
\label{fig:position}
\end{figure}

\item
We introduce a new static analysis approach, called \solar (\emph{soundness-guided reflection analysis}), to resolve Java reflection effectively in practice. As shown in Figure~\ref{fig:position}, 
\solar has three unique advantages 
compared with previous work: 
	\begin{itemize}
\vspace{0.5ex}
	\item \solar is able to yield significantly more sound results than the state-of-the-art reflection analysis. In addition, \solar allows its soundness to be reasoned about when some reasonable assumptions are met.
\vspace{0.5ex}
	\item \solar is able to accurately identify the parts of the program where reflection is analyzed unsoundly or imprecisely, making it possible for users to be
		aware of the effectiveness of their analysis results (as discussed in Section~\ref{intro:pre}).
\vspace{0.5ex}
	\item \solar provides a mechanism to guide users to iteratively refine the analysis results by adding
		\emph{lightweight} annotations until their specific requirements are satisfied, enabling
		reflection to be analyzed in
		a controlled manner.
	\end{itemize}

\vspace{1ex}
\item
We have implemented \solar in \Doop~\cite{Yannis09} (a state-of-the-art pointer analysis tool for Java) 
and released it as an open-source tool. In particular,
\solar can output its reflection analysis results with the format that is supported by \Soot (a popular framework for analyzing Java and Android applications), 
allowing \soot's clients to use \solar's results directly.

\vspace{1ex}
\item
We conduct extensive experiments on evaluating \solar's effectiveness with large Java applications and libraries. Our experimental results provide convincing evidence  on the ability of \solar in analyzing Java reflection effectively, in practice.

\end{itemize}

\subsection{Organization}
\label{intro:org}
The rest of this paper is organized as follows.
We will start
by providing a comprehensive understanding of Java reflection in Section~\ref{sec:understand}. Building on
this understanding, we give an overview of \solar in Section~\ref{sec:solar} and introduce its underlying methodology
in Section~\ref{sec:meth}. Then, we formalize \solar in Section~\ref{sec:form}, describe its implementation in Section~\ref{sec:impl}, and evaluate its effectiveness in Section~\ref{sec:eval}. Finally, we discuss the related work in Section~\ref{sec:related} and conclude in Section~\ref{sec:conclude}.

\section{Understanding Java Reflection}
\label{sec:understand}
Java reflection is a useful but complex language 
feature. 
To gain a deep understanding about Java reflection,
we examine it in three steps.
First, we describe what Java reflection is, why we need it, and how it is proposed (Section~\ref{sec:under:concept}). 
Second, we explain how Java reflection
is designed to be used, i.e., its API 
(Section~\ref{sec:under:inter}).
Finally, we investigate comprehensively how it has been used in real-world Java applications  
(Section~\ref{sec:under:usage}).
After reading this section, the readers are expected
to develop a whole picture about the basic
mechanism behind Java reflection,
understand its core API design, 
and capture the key insights needed for
developing practical reflection analysis tools.

\subsection{Concept}
\label{sec:under:concept}
Reflection, which has long been studied in philosophy, 
represents one kind of human abilities for
introspecting and learning their nature. 
Accordingly, a (non-human)
object can also be endowed with the capability of such self-awareness. This arises naturally in artificial intelligence: ``\emph{Here I am walking into a dark room. Since I cannot see anything, I should turn on the light}". As explained in~\cite{Sobel96}, such thought fragment reveals a self-awareness of behaviour and state, one that leads to a change in that selfsame behaviour and state, which allows an object to examine itself and make use of the meta-level information to decide what to do next.

Similarly, when we enable programs to avail themselves of such reflective capabilities, reflective programs 
will also allow the programs to observe and modify properties of their own behaviour. Thus, let a program be self-aware --- this is the basic motivation of the so-called \emph{computational reflection}, which is also considered as the \emph{reflection} used in the area of programming languages~\cite{Demers95}.

In the rest of this section, we will introduce what computational reflection is (Section~\ref{sec:compRef}), what reflective abilities it supports (Section~\ref{sec:intro-inter}), and how Java reflection is derived from it (Section~\ref{sec:javaRef}).

\subsubsection{Computational Reflection}
\label{sec:compRef}

Reflection, as a concept for computational systems, dates from Brian Smith's doctoral dissertation~\cite{Smith82}. 
Generally, as shown in Figure~\ref{fig:compRef}(a), a \emph{computational system} is related to a domain
and it answers questions about and/or support actions in the domain~\cite{Maes87}.
Internally, a computational system incorporates both
the data that represents entities and relations in the domain and a program that describes how these data may be manipulated.

A computational system $\mathcal{S}$ is said to be
also a \emph{reflective system}, as shown in
Figure~\ref{fig:compRef}(b), if the following two conditions are satisfied:
\begin{itemize}
\item First, the system $\mathcal{S}$ has its own representation, 
	known as its \emph{self-representation} or
	\emph{metasystem}, 
	in its domain as a kind of data to be examined and manipulated. 

\item Second, the system $\mathcal{S}$ and its representation are causally connected: a change to the representation implies a change to the system, and vice versa.
\end{itemize}

\begin{figure}
\includegraphics[width=1\textwidth]{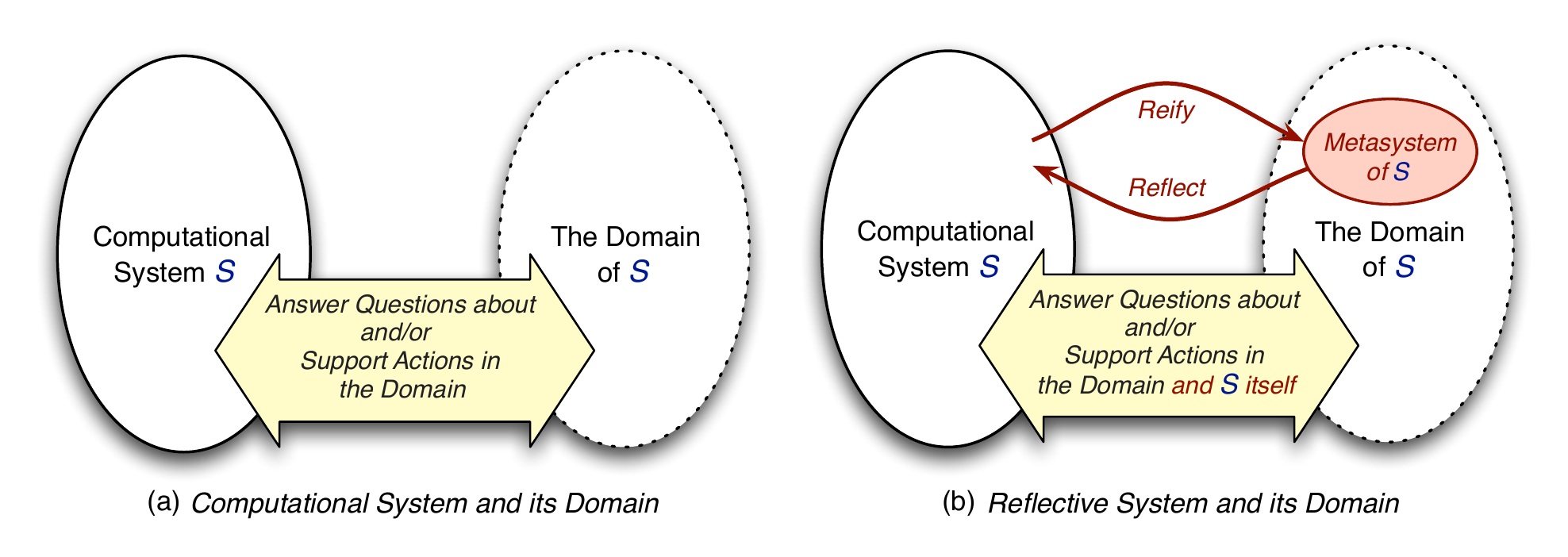}
\caption{Computational vs. reflective computational systems.}
\label{fig:compRef}
\end{figure}

The base system $\mathcal{S}$ should be \emph{reified} into its representation before its metasystem can operate. Then the metasystem examines and manipulates its
behaviour using the reified representation. If any changes are made by the metasystem, then the effects will also be \emph{reflected} in the behavior of the corresponding base system.

\subsubsection{Reflective Abilities}
\label{sec:intro-inter}
Generally, (computational) reflection is the ability of a program to \emph{examine} and \emph{modify} the structure and behavior of a program at runtime~\cite{Herzeel08,reftut96}. 
Thus, it endows the program the capabilities of \emph{self-awareness} and \emph{self-adapting}. 
These two reflective abilities are known 
as \emph{introspection} and \emph{intercession}, respectively, and both 
require a reification mechanism to encode a program's execution state as data first~\cite{Demers95}.

\begin{itemize}
\item \emph{Introspection}: the ability of a program to \emph{observe}, and consequently, reason about its own
	execution state.
\item \emph{Intercession}: the ability of a program to \emph{modify} its own execution state or alter its own interpretation or meaning.
\end{itemize}

Providing full reflective abilities as shown above is hardly acceptable in practice,
as this will introduce both implementation complexities
and performance problems~\cite{Chiba00}. 
Thus, in modern programming languages like Java, 
reflective abilities are only partially supported~\cite{Forman04,Mirrors04}.

\subsubsection{Java Reflection}
\label{sec:javaRef}
Java reflection supports introspection and very limited intercession. In particular, an introspection step is usually followed by behaviour changes such as object creation, method invocation and attribute manipulation\footnote{Some other researchers 
	hold a different view 
that Java reflection does not support intercession
at all~\cite{Mirrors04,fmco2007}, as they adopt
a more strict definition of intercession, which
implies the ability to modify the self-representation of a program.}~\cite{Cazzola04,Forman04}.

Despite its limited reflective abilities, Java reflection is able to allow programmers to break the constraints of staticity and encapsulation, making 
the program adapt to dynamically changing runtime
environments. As a result, Java reflection 
has been widely used in real-world Java applications to facilitate flexibly different programming tasks,
such as reasoning about control (i.e., about which computations to pursue next)~\cite{Forman04}, interfacing (e.g., interaction with GUIs or database systems)~\cite{guiRef02,databaseRef03}, and self-activation (e.g., through monitors)~\cite{monitorRef08}.

Java reflection does not have a reify operation as described in Section~\ref{sec:compRef} (Figure~\ref{fig:compRef}(b))  to turn the basic (running) system (including stack frames) into a representation (data
structure) that is passed to a metasystem. Instead, a kind of metarepresentation, based on 
\emph{metaobjects}, exists when the system begins running and persists throughout the execution of the system~\cite{Forman04}. 

A metaobject is like the reflection in a mirror:  one
can adjust one's smile (behaviour changes) by looking
at oneself in a mirror (introspection).
In Section~\ref{sec:under:inter}, we will look at how Java reflection uses metaobjects and its API to facilitate reflective programming.

\begin{figure}[hpt]
\centering
\begin{tabular}{c}
\begin{lstlisting}
A a = new A();
String cName, mName, fName = ...;
Class clz = Class.forName(cName);
Object obj = clz.newInstance();
Method mtd = clz.getDeclaredMethod(mName,{A.class});
Object l = mtd.invoke(obj, {a}); 
Field fld = clz.getField(fName);
X r = (X)fld.get(a); 
fld.set(null, a);
\end{lstlisting}
\end{tabular}
\caption{An example of reflection usage in Java.}
\label{study:fig:example}
\label{study:fig:mot}
\end{figure}

\subsection{Interface}
\label{sec:under:inter}
We first use a toy example to illustrate some common
uses of the Java reflection API (Section~\ref{sec:under:inter:example}).
We then delve into the details of its core methods, 
which are relevant to (and thus should be handled by)
any reflection analysis (Section~\ref{sec:under:inter:over}).

\subsubsection{An Example}
\label{sec:under:inter:example}

There are two kinds of metaobjects: \texttt{Class} 
objects and member objects. In 
Java reflection, one always starts with a 
\texttt{Class} object and then obtain its member 
objects (e.g., 
\texttt{Method} and \texttt{Field} objects) 
from the \texttt{Class} object by calling 
its corresponding accessor methods (e.g., 
\texttt{getMethod()} and \texttt{getField()}).

In Figure~\ref{study:fig:mot}, the metaobjects \texttt{clz},
\texttt{mtd} and \texttt{fld} are instances of 
the metaobject classes
\texttt{Class}, \texttt{Method} and \texttt{Field},
respectively. \texttt{Constructor} 
can be seen as \texttt{Method} except that the method 
name ``\verb"<init>"'' is implicit. 
\texttt{Class} allows an object to be created reflectively by 
calling \texttt{newInstance()}. As shown in line 4, the dynamic type of \texttt{obj}
is the class (type) represented by \texttt{clz} (specified by \texttt{cName}).
In addition, \texttt{Class} provides
accessor methods such as 
\texttt{getDeclaredMethod()} in line 5 and
\texttt{getField()} in line 7 to allow the member 
metaobjects (e.g., of \texttt{Method} and \texttt{Field})
related to a \texttt{Class} object to be
introspected.  With dynamic invocation,
a \texttt{Method} object can be commanded to invoke the method that it represents (line 6). Similarly, 
a \texttt{Field} object can be commanded to access or modify
the field that it represents (lines 8 and 9). 

\subsubsection{Core Java Reflection API}
\label{sec:under:inter:over}

In reflection analysis,
we are concerned with reasoning about
how reflection affects the control and data flow information in the program.
For example, if a target method (say $m$) that is
reflectively invoked in line 6 in Figure~\ref{study:fig:mot} cannot be resolved statically, the call graph edge 
from this call site to method $m$
(control flow)
and the values passed interprocedurally
 from \texttt{obj} and \texttt{a} to \texttt{this} and the parameter of $m$ (data flow), respectively, will
 be missing.
Therefore, we should focus on the part of the 
Java reflection API that affects a pointer analysis,
a fundamental analysis that statically resolves the control and data flow information in a program~\cite{Livshits05,Yue14,Yannis15,spark,Milanova05,Yannis11,Tian16,Tian17}.

\begin{figure}
\hspace{-2ex}
\includegraphics[width=1.03\textwidth]{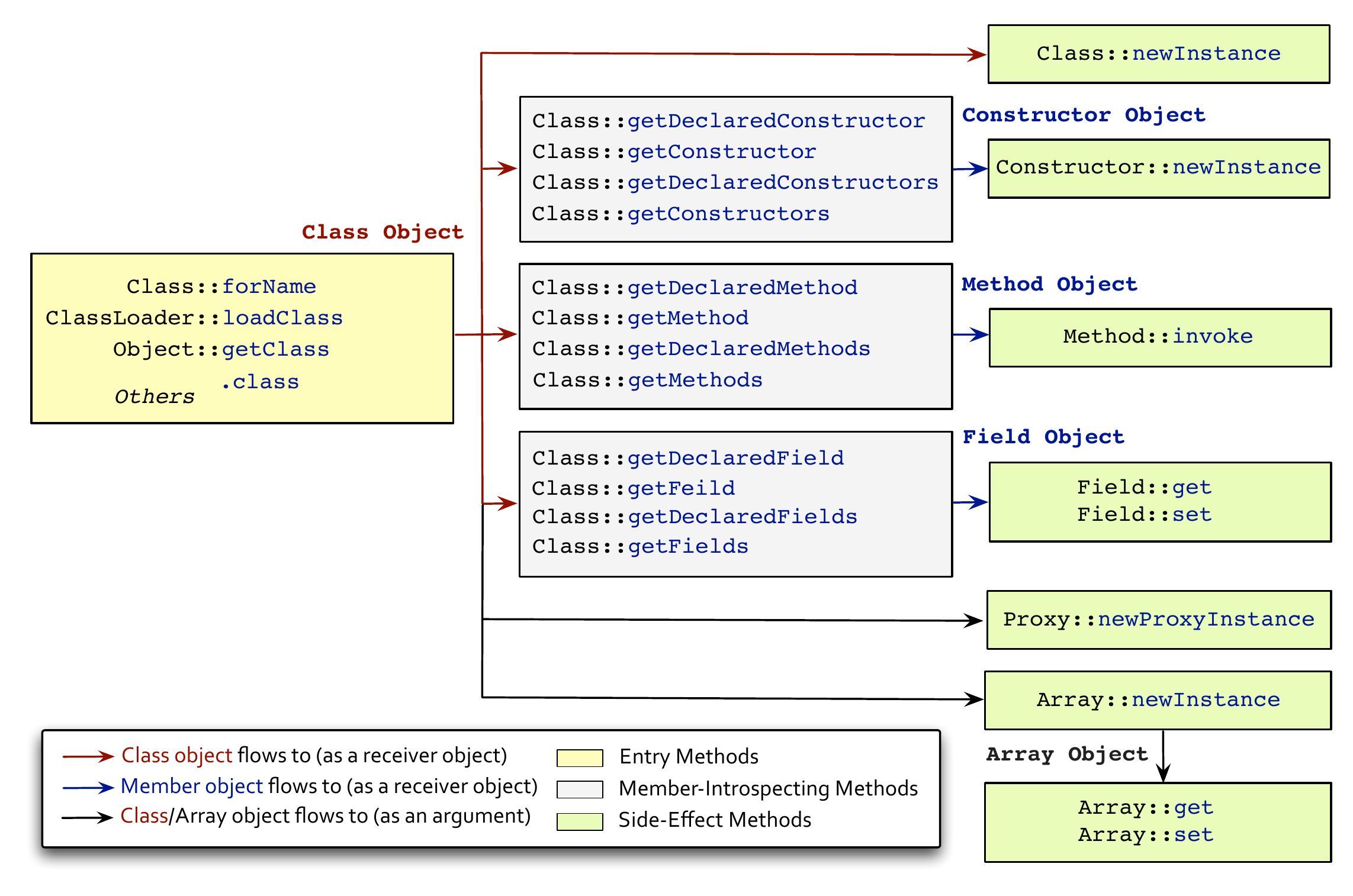}
\caption{Overview of core Java reflection API.\footnotemark{}}
\label{study:fig:coreAPI}
\end{figure}
\footnotetext{We summarize and explain the core reflection API (25 methods) that is critical to static analysis. A more complete reflection API list (181 methods) is given in~\cite{Barros17} without explanations though.}

It is thus sufficient to consider only the 
pointer-affecting methods in the 
Java reflection API.
We can divide such reflective methods
into three categories (Figure~\ref{study:fig:coreAPI}): 

\begin{itemize}
\item
\emph{Entry methods}, 
which create \texttt{Class} objects, e.g., \texttt{forName()} in line 3 in Figure~\ref{study:fig:mot}.
\item
\emph{Member-introspecting methods}, which introspect and retrieve member metaobjects, i.e.,
\texttt{Method} (\texttt{Constructor}) 
and \texttt{Field} objects from a \texttt{Class} object,
e.g., 
\texttt{getDeclaredMethod()} in line 5  and
\texttt{getField()} in line 7 in Figure~\ref{study:fig:mot}.
\item
\emph{Side-effect methods}, which affect 
the pointer information in the program reflectively, 
e.g., 
\texttt{newInstance()}, 
\texttt{invoke()}, 
\texttt{get()}
and \texttt{set()} in lines 4, 6, 8 and 9 in Figure~\ref{study:fig:mot} for 
creating an object, invoking a method,
accessing and modifying a field, respectively.
\end{itemize}

\paragraph{Entry Methods}
\texttt{Class} objects are returned by \emph{entry} 
methods, as everything in Java reflection begins with
\texttt{Class} objects.
There are many entry methods in the Java reflection API.
In Figure~\ref{study:fig:coreAPI}, only the four
most widely used ones are listed explicitly.

Note that
\texttt{forName()} (\texttt{loadClass()}) 
returns a \texttt{Class} object representing a class 
that is specified by the value of its string argument.
The \texttt{Class} object returned by $o$.\texttt{getClass()} and $A$\texttt{.class}
represents the dynamic type (class) of $o$ and $A$,
respectively.

\paragraph{Member-Introspecting Methods}
\texttt{Class} provides a number 
of accessor methods for retrieving 
its member metaobjects, i.e., the \texttt{Method}
(\texttt{Constructor}) and \texttt{Field} objects. 
In addition, these member metaobjects can be used to introspect the 
methods, constructors and fields in their target class.
Formally, these accessor methods are referred to here
as the \emph{member-introspecting} methods.

As shown in Figure~\ref{study:fig:coreAPI}, for each kind of member metaobjects,
there are four member-introspecting methods. 
We take a \texttt{Method} object as an example to illustrate these methods, whose receiver objects are the \texttt{Class} objects returned by the entry methods.

\begin{itemize}
\item
\verb"getDeclaredMethod(String, Class[])" 
returns a \texttt{Method} object
that represents a declared method of the target 
\texttt{Class} object with the name 
(formal parameter types) specified by the first
(second) parameter (line 5
in Figure~\ref{study:fig:mot}).
\item
\verb"getMethod(String, Class[])" 
is similar to \texttt{getDeclaredMethod(String, Class[])} except that
the returned \texttt{Method} object is public
 (either declared or inherited).
 If the target \texttt{Class} does not have a
 matching method, then its superclasses are searched
 first recursively (bottom-up) before its interfaces (implemented).
 \item
\verb"getDeclaredMethods()" returns an array of
 \texttt{Method} objects representing all the methods
 declared in the target \texttt{Class} object. 
 \item
\verb"getMethods()" is similar to \texttt{getDeclaredMethods()} except that
 all the public methods (either declared or inherited)
 in the target \texttt{Class} object are returned.
 \end{itemize}

\paragraph{Side-Effect Methods}

\begin{table}[htbp]
\tbl{Nine side-effect methods and their side effects
	on the pointer analysis, assuming that the target class of \emph{clz} and \emph{ctor} is \emph{A}, the target method of \emph{mtd} is \emph{m} and the target field of \emph{fld} is \emph{f}.
}{
\renewcommand{\arraystretch}{1.6}
\begin{tabular}{ p{4cm} < {\centering} | p{5cm} < {\centering} | p{3.2cm} < {\centering} }
\hline
\textbf{Simplified Method}  & \textbf{Calling Scenario}   & \textbf{Side Effect }\\
\hline \hline
Class::newInstance & o = \emph{clz}.newInstance() & o = new \emph{A}() \\
Constructor::newInstance & o = \emph{ctor}.newInstance(\{arg$_1$, ...\}) & o = new \emph{A}(arg$_1$, ...)\\
Method::invoke & a = \emph{mtd}.invoke(o, \{arg$_1$, ...\}) & a = o.\emph{m}(arg$_1$, ...)\\
Field::get & a = \emph{fld}.get(o) & a = o.\emph{f} \\
Field::set & \emph{fld}.set(o, a) & o.\emph{f} = a \\
Proxy::newProxyInstance & o = \emph{Proxy}.newProxyInstance(...) & o = new Proxy\$*(...)\\
Array::newInstance & o = \emph{Array}.newInstance(\emph{clz}, size) & o = new \emph{A}[size] \\
Array::get & a = \emph{Array}.get(o, i) & a = o[i] \\
Array::set & \emph{Array}.set(o, i, a) & o[i] = a \\
 \hline
\end{tabular}}
\label{study:table:sideAPI}
\end{table}

As shown in
Figure~\ref{study:fig:coreAPI}, a total of nine 
side-effect methods that can possibly modify or use 
(as their side effects) the pointer information in 
a program are listed.  
Accordingly, Table~\ref{study:table:sideAPI} explains how these methods affect the pointer information by
giving their side effects on the pointer analysis.

In Figure~\ref{study:fig:coreAPI}, the first five side-effect methods use four kinds of metaobjects as their receiver objects while the last four methods use 
\texttt{Class} or \texttt{Array} objects
as their arguments.
Below we briefly examine them in the order given
in Table~\ref{study:table:sideAPI}.
\begin{itemize}
\item
The side effect of \texttt{newInstance()} is allocating an object with the type specified by its metaobject \emph{clz} or \emph{ctor} (say $A$) and initializing it
via a constructor of $A$, which is the default
constructor in the case of \texttt{Class::newInstance()} 
and the constructor specified explicitly in the case
of \texttt{Constructor::newInstance()}.
\item
The side effect of \texttt{invoke()} is a virtual call when the first argument of \texttt{invoke()}, say $o$, is not \texttt{null}. The receiver object is 
$o$ as shown in the ``Side Effect'' column in Table~\ref{study:table:sideAPI}. When $o$ is \texttt{null}, \texttt{invoke()} should be a static call.
\item 
The side effects of \texttt{get()} and \texttt{set()} are retrieving (loading) and modifying (storing) the value of a instance field, respectively, when their first argument, say $o$, is not \texttt{null}; 
otherwise, they are operating on a static field.  

\item
The side effect of \verb"newProxyInstance()" is creating an object of a proxy class \verb"Proxy$*", and this proxy class is generated dynamically according to its arguments (containing a \texttt{Class} object). 
\verb"Proxy.newProxyInstance()" can be analyzed 
according to its
semantics. A call to this method returns a \texttt{Proxy} object, which has an associated
invocation handler object that implements the \texttt{InvocationHandler} interface. A
method invocation on a \texttt{Proxy} object through one of its \texttt{Proxy} interfaces will be dispatched to the \texttt{invoke()} method of the object's invocation handler.
\item 
The side effect of \texttt{Array.newInstance()} is creating an array (object) with the component type represented by the \texttt{Class} object (e.g., \emph{clz} in Table~\ref{study:table:sideAPI}) used as its first argument.
\texttt{Array.get()} and \texttt{Array.set()} are 
retrieving and modifying an index element
in the array object specified 
as their first argument, respectively.
\end{itemize}

\subsection{Reflection Usage}
\label{sec:under:usage}
The Java reflection API is rich and complex.
We have conducted an empirical study  
to understand reflection usage in practice
in order to guide the design and 
implementation of a sophisticated reflection analysis
described in this paper.
In this section, we first list the focus
questions in Section~\ref{sec:study:ques},
then describe the experimental setup in 
Section~\ref{sec:study:setup},
and finally, present the study results in Section~\ref{sec:study:results}.

\subsubsection{Focus Questions}
\label{sec:study:ques}

We consider to address the following seven focus 
questions in
order to understand how Java reflection 
is used in the real world:
\begin{itemize}
\item 
Q1. The core part of reflection analysis is to resolve all the nine side-effect methods
(Table~\ref{study:table:sideAPI}) effectively. What are the side-effect methods that are
most widely used and how are the remaining ones
used in terms of their relative frequencies?
\item
Q2. The Java reflection API contains many entry
methods for returning \texttt{Class} 
objects. Which ones should be focused on by
an effective reflection analysis?
\item
Q3. Existing reflection analyses resolve reflection 
by analyzing statically the string arguments of entry and member-introspecting method calls.
How often are these strings constants and how often 
can non-constant strings be resolved by a
simple string analysis that models 
string operations such as ``\texttt{+}'' and \texttt{append()}?
\item
Q4. Existing reflection analyses ignore
the member-introspecting methods that return an array of member metaobjects. Is it necessary to handle such
methods?
\item
Q5. 
Existing reflection analyses usually treat
reflective method calls and field accesses as
being non-static. Does this treatment work well in
real-world programs? Specifically, how often are static reflective targets used in reflective code?
\item
Q6. In~\cite{Livshits05}, intraprocedural post-dominating cast operations are leveraged to resolve \texttt{newInstance()} when its class type is unknown. This approach is still adopted by many reflection analysis tools. 
Does it generally work in practice?
\item
Q7. What are new insights on handling Java reflection
(from this paper)?
\end{itemize}

\subsubsection{Experimental Setup}
\label{sec:study:setup}
We  have selected a set of  16 representative Java programs,
including three popular desktop applications, 
 \texttt{javac-1.7.0}, \texttt{jEdit-5.1.0} and 
 \texttt{Eclipse-4.2.2} (denoted \texttt{Eclipse4}),
two popular server applications,
\texttt{Jetty-9.0.5} and 
\texttt{Tomcat-7.0.42}, and all eleven DaCapo benchmarks (2006-10-MR2)~\cite{DaCapo}.
Note that the DaCapo benchmark suite includes
an older version of \texttt{Eclipse}
(version 3.1.2). We exclude its \texttt{bloat} benchmark since
its application code is reflection-free.
We consider \texttt{lucene} instead of
\texttt{luindex} and \texttt{lusearch} separately since
these two benchmarks are derived from \texttt{lucene}
with the same reflection usage.

We consider a total of 191 methods in the Java
reflection API (version 1.6), including the ones 
mainly from package \verb"java.lang.reflect" and class
\verb"java.lang.Class".

We use \textsc{Soot}~\cite{soot} to pinpoint 
the calls to reflection methods in the bytecode
of a program. To understand the common reflection usage, 
we consider only the 
reflective calls found in the application classes and
their dependent libraries but exclude the standard 
Java libraries. To increase the code coverage
for the five applications considered,
we include the \texttt{jar} files whose names
contain the names of these applications 
(e.g., \texttt{*jetty*.jar} for \texttt{Jetty}) 
and make them available under the \emph{process-dir} 
option supported by \textsc{Soot}. For
\texttt{Eclipse4}, we use
\verb"org.eclipse.core." \verb"runtime.adaptor.EclipseStarter" 
to let \textsc{Soot} locate all the other jar files used.

We manually inspect the reflection usage
in a program in a demand-driven manner, starting
from its side-effect methods, assisted by  
\emph{Open Call Hierarchy}
in \texttt{Eclipse}, by following their
backward slices. For a total of 
609 side-effect call sites examined,
510 call sites for calling entry methods
and 304 call sites for calling member-introspecting methods
are tracked and analyzed.
As a result, a total of 1,423 reflective call sites,
together with some nearby statements,
are examined in our study.

\begin{figure}[t]
\centering
\includegraphics[width=0.9\linewidth]{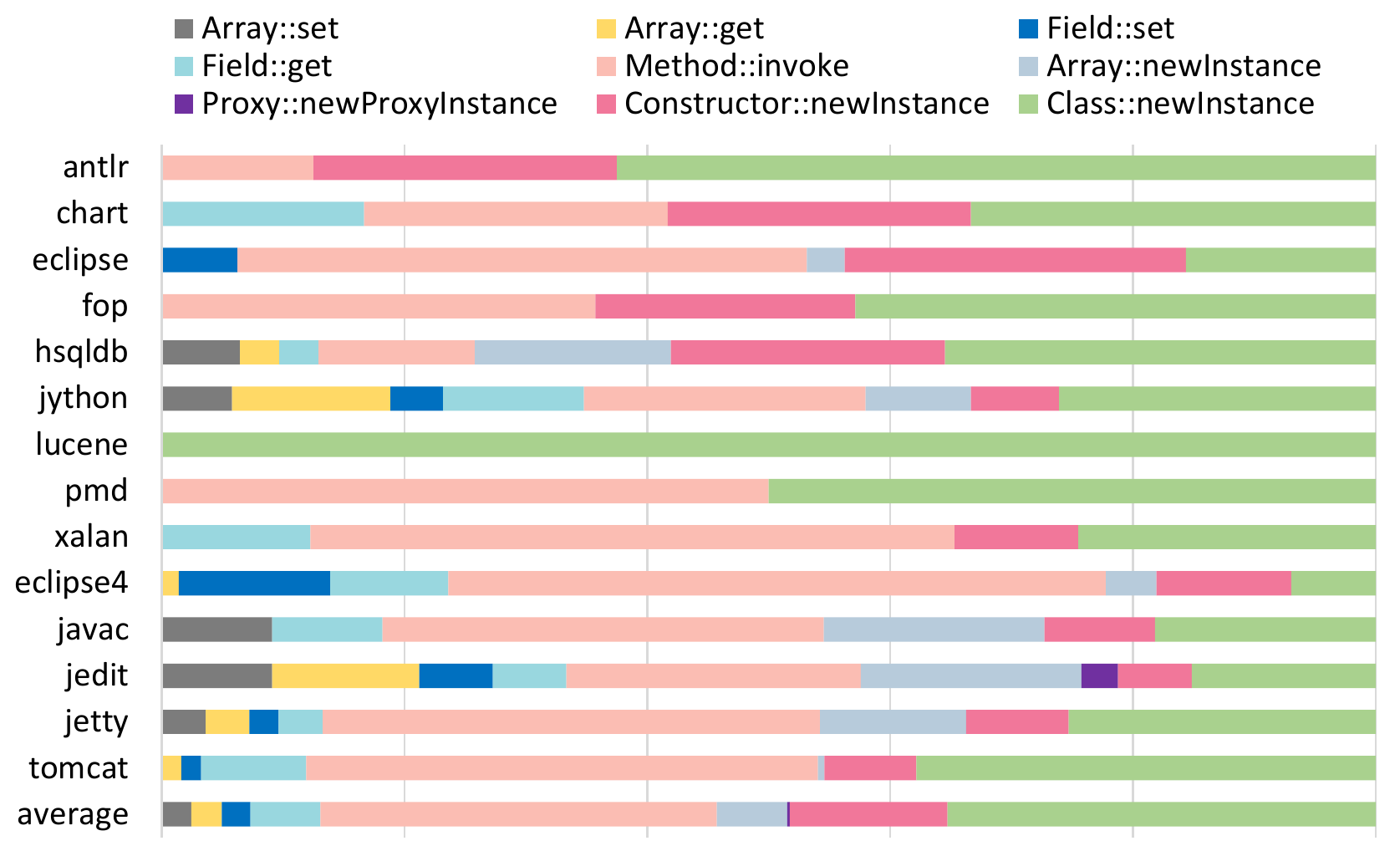}
\caption{Side-effect methods.}
\label{study:fig:side}
\end{figure}

\subsubsection{Results}
\label{sec:study:results}

Below we describe our seven findings on reflection 
usage as our answers to the seven focus questions listed
in Section~\ref{sec:study:ques}, respectively. 
We summarize our findings as individual remarks, which are expected 
to be helpful in guiding the development of practical
reflection analysis techniques and tools in future
research.

\paragraph{Q1. Side-Effect Methods}
\label{study:sideeffect}

Figure~\ref{study:fig:side} depicts the
percentage frequency distribution of all the
nine side-effect methods 
in all the programs studied. 
We can see that \verb"newInstance()" and \verb"invoke()" 
are the ones that are
most frequently used (46.3\% and 32.7\%, respectively,
on average). Both of them are
handled by existing static analysis tools such as \doop, \soot, \Wala and \Bdd. 
However, \texttt{Field}- and \texttt{Array}-related side-effect
methods, which are also used in many programs, are 
ignored by most of these tools. To the best of our
knowledge, they are handled only by  \elf~\cite{Yue14},
 \solar~\cite{Yue15} and \doop~\cite{Yannis15}. 
Note that \verb"newProxyInstance()"  is used 
in \texttt{jEdit} only in our study and a recent 
survey on reflection analysis~\cite{Barros17}
reports more its usages in the real world.

\begin{framed}
\noindent\textbf{Remark 1.} \emph{Reflection analysis should at least handle \texttt{newInstance()} and \texttt{invoke()} as they are the most frequently used side-effect methods (79\% on average), which will significantly affect a program's behavior, in general; otherwise, much of the codebase may be invisible for analysis. Effective reflection analysis should also consider \texttt{Field}- and \texttt{Array}-related side-effect methods, as they are also commonly used.}
\end{framed}

\paragraph{Q2. Entry Methods}
\label{study:entry}

Figure~\ref{study:fig:entry} shows the percentage frequency
distribution of eight entry methods.  
``\emph{Unknown}'' is included since we failed to find
the entry methods for some side-effect calls (e.g.,
\verb"invoke()") even by using \texttt{Eclipse}'s
\emph{Open Call Hierarchy} tool. 
For the first 12 programs,
the six entry methods as shown (excluding 
``\emph{Unknown}'' and
``\emph{Others}'') are the only ones leading to
side-effect calls.
For the last two, \texttt{Jetty} and
\texttt{Tomcat},
``\emph{Others}'' stands for 
\texttt{defineClass()} in \texttt{ClassLoader} and
\texttt{getParameterTypes()} in \texttt{Method}.
Finally,
\verb"getComponentType()" is usually used in the form
of \verb"getClass().getComponentType()" for creating a
\texttt{Class} object argument for 
\verb"Array.newInstance()". 

On average, 
\texttt{Class.forName()}, \texttt{.class}, \texttt{getClass()} and \texttt{loadClass()}
are the top four most frequently used 
(48.1\%, 18.0\%, 17.0\% and 9.7\%, respectively).
A class loading strategy can
be configured in \texttt{forName()}
and \texttt{loadClass()}. In practice, \texttt{forName()}
is often used by the system class loader and \texttt{loadClass()} is usually overwritten in customer class loaders, especially in framework applications such as \texttt{Tomcat} and \texttt{Jetty}.

\begin{figure}[thp]
\centering
\includegraphics[width=0.9\linewidth]{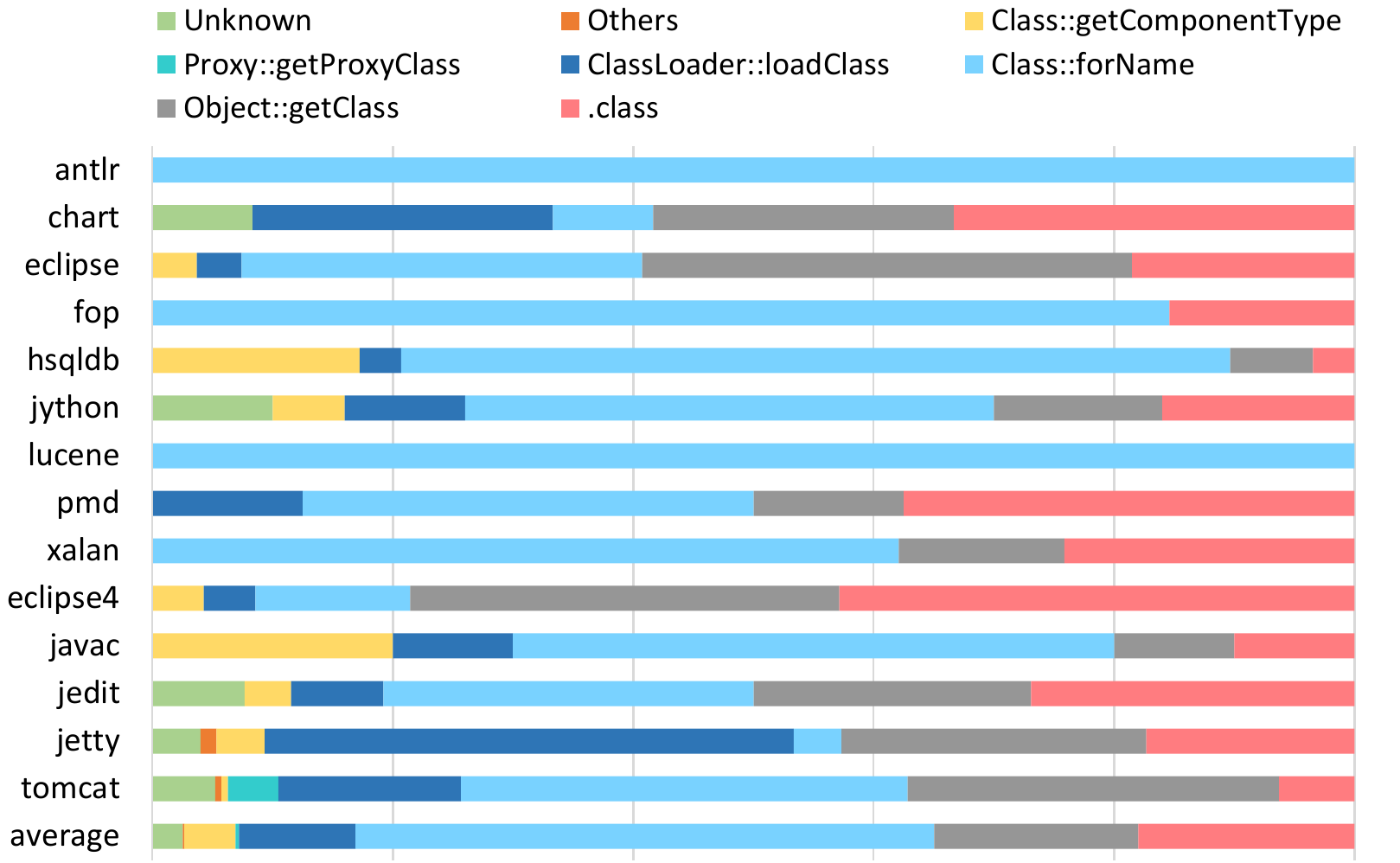}
\caption{Entry methods.}
\label{study:fig:entry}
\end{figure}

\begin{framed}
\noindent\textbf{Remark 2.} \emph{Reflection analysis should handle \texttt{Class.forName()}, \texttt{getClass()}, \texttt{.class}, and \texttt{loadClass()},
which are the four major entry methods for creating
\texttt{Class} objects. In addition,
\texttt{getComponentType()} should also be modeled if \texttt{Array}-related side-effect methods are analyzed, as they are usually used together.}
\end{framed}

\paragraph{Q3. String Constants and String Manipulations}
\label{study:sec:string}

In entry methods,
\verb"Class.forName()" and \verb"loadClass()" each have
a \verb"String" parameter to specify the target class. In member-introspecting methods,
\verb"getDeclaredMethod(String,...)" and
\verb"getMethod(String,...)" 
each return a \texttt{Method} object
named by its first parameter;
\verb"getDeclaredField(String)" and
\verb"getField(String)" 
each return a \texttt{Field} object named by its single parameter.

As shown in Figure~\ref{study:fig:str}, string constants
are commonly used when calling the two entry methods
(34.7\% on
average) and the four member-introspecting methods
(63.1\% on average). In the presence of
string manipulations, 
many class/method/field names are unknown exactly. 
This is mainly because their static resolution 
requires \emph{precise} handling of many different operations 
e.g., \verb"subString()" and \verb"append()". 
In fact, many cases are rather complex and thus 
cannot be handled well by simply
modeling the \texttt{java.lang.String}-related API.
Thus, \solar does not currently handle 
string manipulations.  However,
the incomplete information about class/method/field names (i.e., partial string information)
can be exploited beneficially~\cite{Yannis15}.

We also found that
many string arguments are \emph{Unknown} 
(55.3\% for calling entry methods
and 25.1\% for calling member-introspecting methods,
on average). These are the strings that may be 
read from, say, configuration files, command lines, or even Internet URLs.
Finally, string constants are found to be
more frequently used for calling 
the four member-introspecting methods
than the two entry methods: 146 calls to
\verb"getDeclaredMethod()" and \verb"getMethod()",
27 calls to 
\verb"getDeclaredField()" and \verb"getField()" in contrast
with 98 calls to \verb"forName()" and \verb"loadClass()". 
This suggests that the analyses
that ignore string constants
flowing into some
member-introspecting methods may fail to exploit 
such valuable
information and thus become imprecise.
\begin{figure}[t]
\begin{tabular}{cc}
\multicolumn{2}{c}{
\hspace{-0.3cm}
\includegraphics[width=1.0 \textwidth]{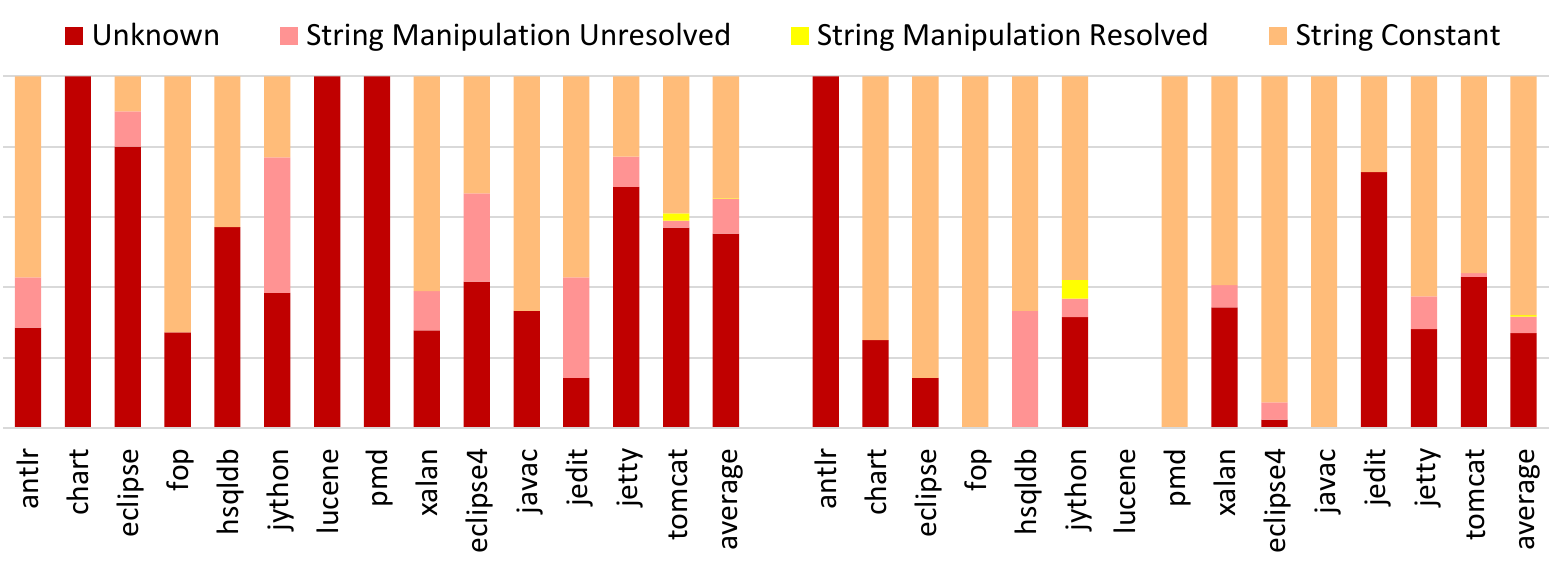}
}\\
\hspace{1.2cm}{\small (a) Calls to entry methods}
&
\hspace{1.2cm}{\small (b) Calls to member-introspecting methods}
\end{tabular}
\vspace*{-1pt}
\caption{Classification of the \texttt{String} arguments 
of two entry methods, \texttt{forName()} and
\texttt{loadClass()}, and four member-introspecting
methods,
\texttt{getMethod()}, \texttt{getDeclaredMethod()},
\texttt{getField()} and \texttt{getDeclaredField()}.
\label{study:fig:str}
}
\end{figure}

\begin{framed}
\noindent\textbf{Remark 3.} \emph{Resolving reflective targets by string constants does not always work.
On average, only 49\% reflective call sites (where string arguments are used to specify reflective targets) use string constants. In addition, fully resolving non-constant string arguments by string manipulation, although mentioned elsewhere~\cite{Livshits05,Bodden11}, may be hard to achieve, in practice.}
\end{framed}

\paragraph{Q4. Retrieving an Array of Member Objects}
\label{study:array}
As introduced in Section~\ref{sec:under:inter:over},
half of member-introspecting methods 
(e.g., \texttt{getDeclaredMethods()})
return an array of member metaobjects. 
Although not as frequently used as the ones returning single member
metaobject (e.g., \texttt{getDeclaredMethod()}), 
they play
an important role in introducing new program behaviours
in some applications.
For example, in the two \texttt{Eclipse} programs
studied, there are four
\verb"invoke()" call sites called on an array of 
\texttt{Method}
objects returned from \verb"getMethods()" and
15 \verb"fld.get()" and \verb"fld.set()" 
call sites called on an array of \texttt{Field} objects
returned by \verb"getDeclaredFields()". Through these calls, dozens of methods are invoked and hundreds of fields are modified reflectively.
Ignoring such methods as in prior work
\cite{Livshits05} and tools (\bdd, \wala, \soot) may lead to significantly missed program behaviours by the analysis.

\begin{framed}
\noindent\textbf{Remark 4.} \emph{In member-introspecting methods, get(Declared)Methods/Fields/\linebreak Constructors(), which return an array of member metaobjects,
are usually ignored by most of existing reflection analysis tools. However, they play an important role in certain applications for both method invocations and field manipulations.}
\end{framed}

\begin{figure}[t]
\begin{minipage}{0.49\textwidth}
\includegraphics[width=\textwidth]{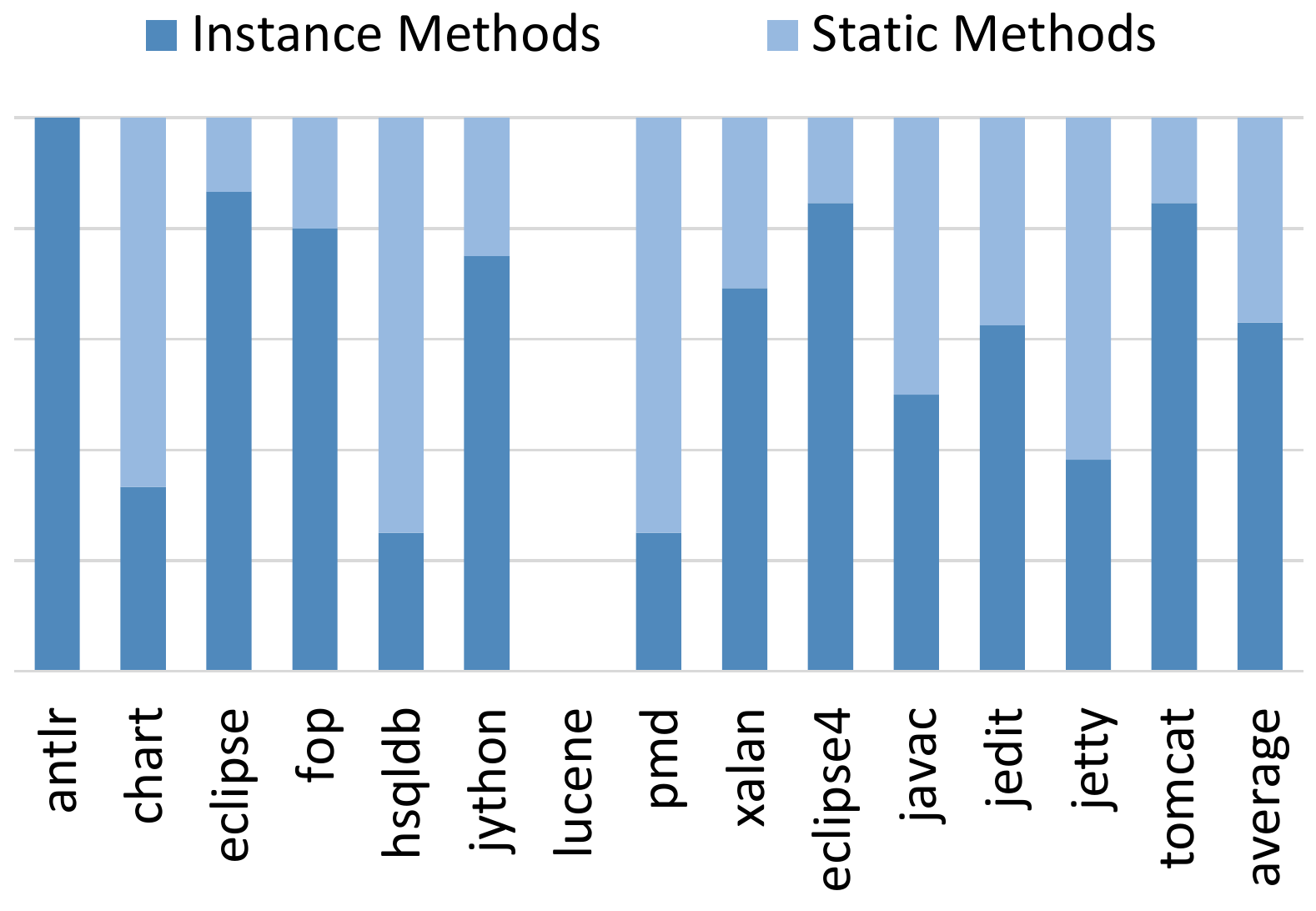}
\centering
{\small (a) Method::invoke() call sites  }
\end{minipage}
\hspace{1ex}
\begin{minipage}{0.5\textwidth}
\includegraphics[width=\textwidth]{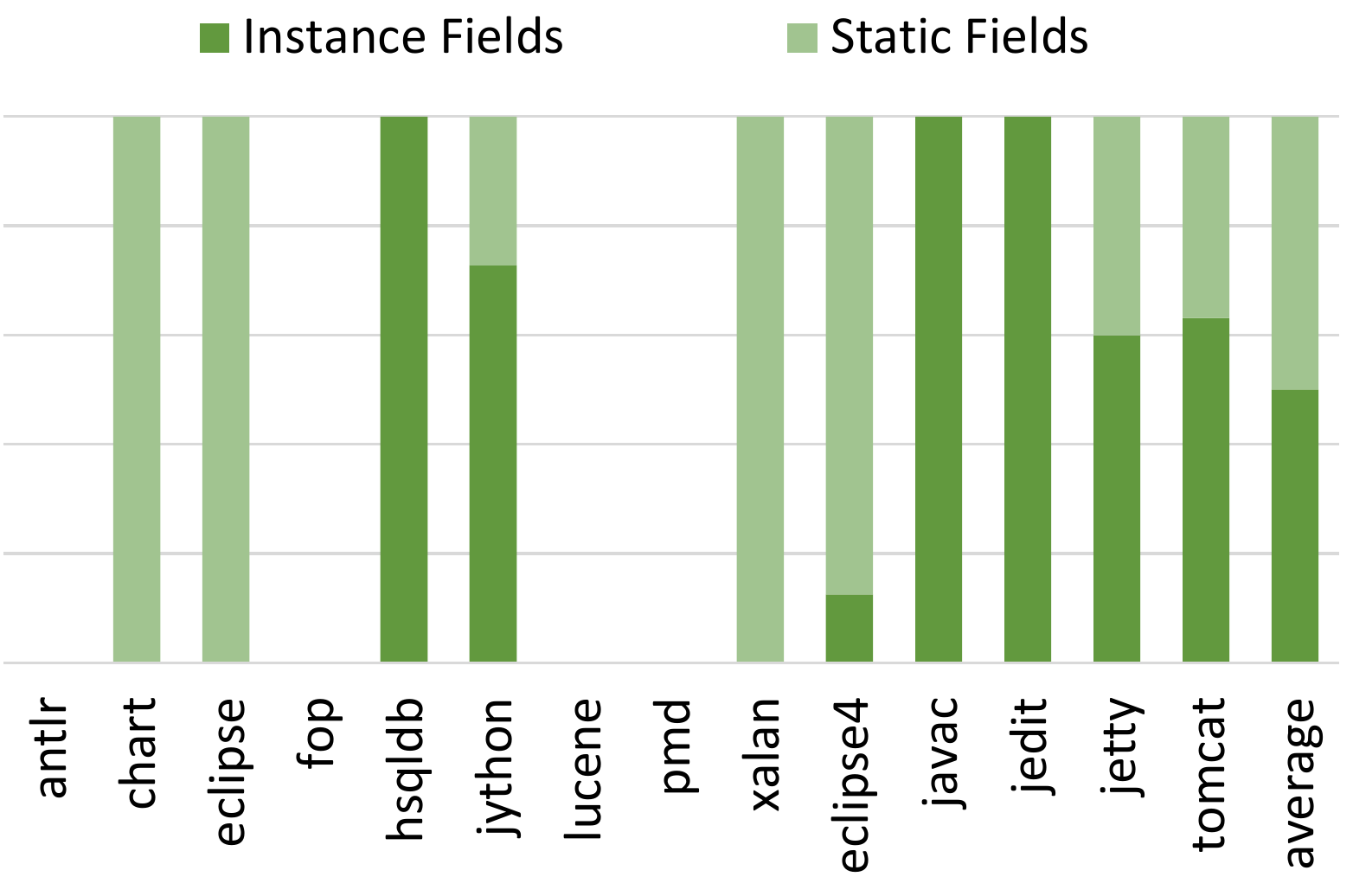}
\centering
{\small(b) Field::get()/set() call sites}
\end{minipage}
\caption{The percentage frequency distribution of side-effect call sites on instance and static members.
\label{study:fig:static}
}
\end{figure}

\paragraph{Q5. Static or Instance Members}
\label{study:static}

In the literature on reflection analysis~\cite{Livshits05,Yue14,Yannis15}, reflective targets are mostly
assumed to be instance members. Accordingly, calls to
the side-effect methods such as \texttt{invoke()}, \texttt{get()} and \texttt{set()}, are usually 
considered as virtual calls, instance field accesses,
and instance field modifications, respectively (see Table~\ref{study:table:sideAPI} for details). 
However, in real programs, as shown in Figure~\ref{study:fig:static}, on average, 37\% of the \texttt{invoke()} call sites are found to invoke static methods and 50\% of the \texttt{get()}/\texttt{set()} call sites are found to access/modify static fields. Thus in practice, reflection analysis should distinguish both cases and 
also be aware of whether a reflective target is
a static or instance member, since the approaches for
resolving both cases are usually different.

\begin{framed}
\noindent\textbf{Remark 5.} \emph{Static methods/fields are invoked/accessed as frequently as instance methods/fields in Java reflection, even though the latter 
has received more attention in the literature. In practice, reflection analysis should distinguish
the two cases and adopt appropriate approaches for
handling them.}
\end{framed}

\paragraph{Q6. Resolving \texttt{newInstance()} by Casts}
\label{study:sec:cast}

In Figure~\ref{study:fig:example}, when \texttt{cName} is a not string constant, the (dynamic) type of \texttt{obj} created by \texttt{newInstance()} in line 4 is unknown. For this case, Livshits et al.~\cite{Livshits05} propose to infer the type of \texttt{obj} by leveraging the cast operation that post-dominates intra-procedurally the \texttt{newInstance()} call site. If the cast type is \texttt{A}, the type of \texttt{obj} must be \texttt{A} or one of its subtypes assuming that the cast operation does not throw any exceptions.
This approach has been implemented in many analysis tools such as \wala, \Bdd and \elf. 

However, as shown in Figure~\ref{study:fig:cast}, 
exploiting casts this way
does not always work. On average, 28\% of \texttt{newInstance()} call sites have no such intra-procedural post-dominating casts.
As \texttt{newInstance()} is the most widely used side-effect method, its unresolved call sites 
may significantly affect the soundness of the analysis,
as discussed in Section~\ref{dynamic}. Hence, we need a better solution to handle \texttt{newInstance()}.

\begin{figure}[thp]
\centering
\includegraphics[width=0.7\linewidth]{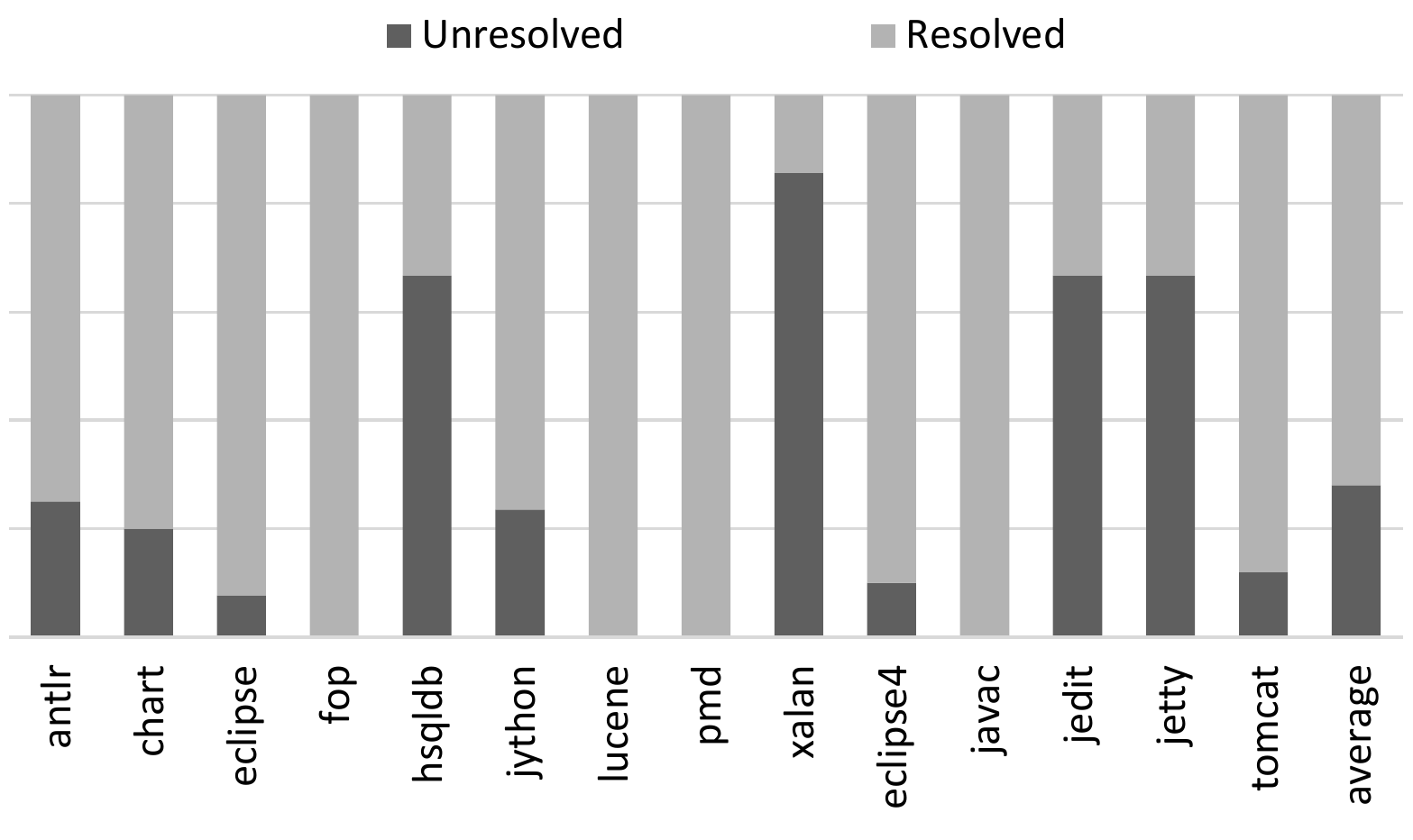}
\caption{\texttt{newInstance()} resolution by leveraging intra-procedural post-dominating casts.}
\label{study:fig:cast}
\end{figure}

\begin{framed}
\noindent\textbf{Remark 6.} \emph{Resolving \texttt{newInstance()} calls by leveraging their intra-procedural post-dominating cast operations fails to work for 28\% of the \texttt{newInstance()} call sites found. As \texttt{newInstance()} affects critically the soundness of reflection analysis (Remark 1), a more effective approach for its 
resolution is required.}
\end{framed}

\paragraph{Q7. Self-Inferencing Property}
\label{study:sec:self}

As illustrated by the program given
in Figure~\ref{study:fig:example}, the names of its reflective targets are specified by the \emph{string arguments} (e.g., \texttt{cName}, \texttt{mName} and \texttt{fName}) at the entry and member-introspecting reflective calls. Therefore, string analysis has been a
representative approach for static reflection analysis in the last decade. However, if the value of a string
is unknown statically (e.g., read from external files or command lines), then the related reflective calls,
including those to \texttt{newInstance()},
may have to be ignored, rendering the corresponding codebase or operations invisible to the analysis.
To improve precision, in this case,
the last resort is to exploit the existence of
some intra-procedurally post-dominating cast
operations on a call to \texttt{newInstance()} in order
to deduce the types of objects reflectively created
($Q6$).

However, in our study, we find that there are many
other rich hints about the behaviors of reflective
calls at their usage sites. Such hints can be and 
should be exploited to make reflection analysis more 
effective, even when some string values 
are partially or fully unknown. 
In the following, we first look at
three real example programs to examine 
what these hints are and expose a so-called
\emph{self-inferencing property} inherent in these
hints.
Finally, we explain why self-inferencing property is 
pervasive for Java reflection and discuss its 
potential in making reflection analysis more effective.

\lstset{numbers=none}
\begin{figure}[hptb]
\begin{lstlisting}[escapeinside={(*}{*)}]
(*\large{Application: \color{myred}{Eclipse (v4.2.2)}}*)
(*Class:\color{myblue}org.eclipse.osgi.framework.internal.core.FrameworkCommandInterpreter*)
(*\color{mygreen}{123}*)  public Object execute(String cmd) {...
(*\color{mygreen}{155}*)    Object[] parameters = new Object[] {this}; ...
(*\color{mygreen}{167}*)    for (int i = 0; i < size; i++) {
(*\color{mygreen}{174}*)      method = target.getClass().getMethod("_" + cmd, parameterTypes);
(*\color{mygreen}{175}*)      retval = method.invoke(target, parameters); ...}
(*\color{mygreen}{228}*)  }
\end{lstlisting}
\caption{Self-inferencing property for a
	reflective method invocation, deduced from
	the number and dynamic types of the components
	of the one-dimensional array argument,
	\texttt{parameters}, at 
	a \texttt{invoke()} call site.}
\label{study:fig:invoke}
\end{figure}

\begin{example}[Reflective Method Invocation 
		(Figure~\ref{study:fig:invoke})]
The method name (the first argument of \texttt{getMethod()} in line 174) is statically unknown as part of it is read from command line \texttt{cmd}. However, the target method (represented by \texttt{method}) can be deduced from the second argument (\texttt{parameters}) of the corresponding side-effect call \texttt{invoke()} in line 175. Here, \texttt{parameters} is an array of objects,
with only one element (line 155). By querying the pointer analysis and also leveraging the type information 
in the program, we know that the type of the object pointed to by \texttt{this} is \texttt{FrameworkCommandInterpreter}, which has no subtypes. As a result, we can infer that the descriptor of the target method in line 175
must have only one argument and its declared type must be \texttt{FrameworkCommandInterpreter} or one of its supertypes.
\end{example}

\begin{figure}[hptb]
\begin{lstlisting}[escapeinside={(*}{*)}]
(*\large{Application: \color{myred}{Eclipse (v4.2.2)}}*)
(*Class:\color{myblue}org.eclipse.osgi.framework.internal.core.Framework*)
(*{\color{mygreen}1652}*)  public static Field getField(Class clazz, ...) {
(*{\color{mygreen}1653}*)    Field[] fields = clazz.getDeclaredFields(); ...
(*{\color{mygreen}1654}*)    for (int i = 0; i < fields.|length|; i++) { ...
(*{\color{mygreen}1658}*)      return fields[i]; } ...
(*{\color{mygreen}1662}*)  }
(*{\color{mygreen}1682}*)  private static void forceContentHandlerFactory(...) {
(*{\color{mygreen}1683}*)    Field factoryField = getField(URLConnection.class, ...);     
(*{\color{mygreen}1687}*)    java.net.ContentHandlerFactory factory = 
          (java.net.ContentHandlerFactory) factoryField.get(null); ...
(*{\color{mygreen}1709}*)  }
\end{lstlisting}
\caption{Self-inferencing property for 
	a reflective field access, deduced from
	the cast operation and the \texttt{null} argument  used at a \texttt{get()} call site.
\label{study:fig:get}
}
\end{figure}

\begin{example}[Reflective Field Access (Figure~\ref{study:fig:get})]
In this program, \verb"factoryField" (line 1683) is obtained as a
\texttt{Field} object from an array of 
\texttt{Field} objects created in line 1653 
for all the fields in \verb"URLConnection".
In line 1687,
the object returned from \verb"get()" is cast to 
\verb"java.net.ContentHandlerFactory". Based on 
its cast operation and \texttt{null} argument, we know that the call to \texttt{get()}
may only access the static fields of \verb"URLConnection" 
with the type  \verb"java.net.ContentHandlerFactory", 
its supertypes or its subtypes. 
Otherwise, all the fields in
\verb"URLConnection" must be assumed to be accessed
conservatively.
\end{example}

\begin{figure}[hptb]
\begin{lstlisting}[escapeinside={(*}{*)}]
(*\large{Application: \color{myred}{Eclipse (v4.2.2)}}*)
(*Class:\color{myblue}org.eclipse.osgi.util.NLS*)
(*{\color{mygreen}300}*)  static void load(final String bundleName, Class<?> clazz) {
(*{\color{mygreen}302}*)    final Field[] fieldArray = clazz.getDeclaredFields();
(*{\color{mygreen}336}*)    computeMissingMessages(..., fieldArray, ...); ...
(*{\color{mygreen}339}*)  }
(*{\color{mygreen}267}*)  static void computeMissingMessages(..., Field[] fieldArray,...) {
(*{\color{mygreen}272}*)    for (int i = 0; i < numFields; i++) {
(*{\color{mygreen}273}*)      Field field = fieldArray[i];
(*{\color{mygreen}284}*)      String value = "NLS missing message: " + ...;
(*{\color{mygreen}290}*)      field.set(null, value); } ...
(*{\color{mygreen}295}*)  }
\end{lstlisting}
\caption{Self-inferencing property for a
reflective field modification, deduced from 
the \texttt{null} argument  and
the dynamic type of the \texttt{value} argument at
a \texttt{set()} call site.
\label{study:fig:set}
}
\end{figure}

\begin{example}[Reflective Field Modification (Figure~\ref{study:fig:set})]
Like the case in Figure~\ref{study:fig:get}, the field object in line 290 is also 
read from an array of field objects
created in line 302.
This code pattern appears one more time in line 432 in the same class, i.e., \texttt{org.eclipse.osgi.util.NLS}. 
According to the two arguments,
\texttt{null} and
\texttt{value},  provided
at \texttt{set()} (line 290), we can deduce that 
the target field (to be modified in line 290) is static (from \texttt{null}) and its declared type must be 
\texttt{java.lang.String} or one of its supertypes (from
the type of \texttt{value}).
\end{example}

\begin{definition}[Self-Inferencing Property]
\label{def:self}
For each side-effect call site, where reflection is used, all the information of its arguments (including the receiver object), i.e., the number of arguments, their types, and the possible downcasts on its returned values, together with the possible string values statically resolved at its corresponding entry and member-introspecting call sites, forms its self-inferencing property.
\end{definition}

We argue that the self-inferencing property is a pervasive fact about Java reflection, due to the 
characteristics of object-oriented programming and 
the Java reflection API. 
As an example, the declared type of the object (reflectively returned by  \texttt{get()} and \texttt{invoke()} or created by \texttt{newInstance()}) is always
\texttt{java.lang.Object}. Therefore,
the object returned must be first cast 
to a specific type before it is used as 
a regular object, except when its dynamic type is 
\texttt{java.lang.Object} or it will be used only
as an receiver for
the methods inherited from \texttt{java.lang.Object}; otherwise,
the compilation would fail.
As another example, the descriptor of a
target method reflectively called at \texttt{invoke()} 
must be consistent with what is specified by its
second argument (e.g., \texttt{parameters} in line 176
of Figure~\ref{study:fig:invoke});
otherwise, exceptions would be thrown at runtime. 
These constraints should be exploited to
enable resolving reflection in a disciplined way. 

The self-inferencing property not only helps resolve 
reflective calls more effectively when the values of string arguments are partially known (e.g., when either a
class name or a member name is known), but also 
provides an opportunity to resolve some reflective
calls even if the string values are fully unknown.
For example, in some Android apps, class and method names for reflective calls are encrypted for benign or malicious obfuscation, which ``\emph{makes it impossible for any static analysis to recover the reflective call}''~\cite{Rastogi13}. 
However, this appears to be too pessimistic 
in our setting, because, in addition to the string values, some other self-inferencing hints are possibly available to facilitate reflection resolution. For example, 
given \texttt{(A)invoke(o, \{...\})}, the class type of the target method can be inferred from the dynamic type 
of \texttt{o} (by pointer analysis). In addition, the 
declared return type and descriptor of the target 
method can also be deduced from
\texttt{A} and \texttt{\{...\}}, 
respectively, as discussed above.

\begin{framed}
\noindent\textbf{Remark 7.} 
\emph{
Self-inferencing property is an inherent and pervasive 
one in the reflective code of Java programs. 
However, this property has not been fully exploited 
in analyzing reflection before. 
We will show how this property can be leveraged in different ways (for analyzing different kinds of reflective methods as shown in Sections~\ref{sec:collect} and~\ref{sec:lhm}) in order
to make reflection analysis significantly more effective.
}
\end{framed}

\section{Overview of \textsc{Solar}}
\label{sec:solar}
We first introduce the design goal of, challenges faced
by, and insights behind \solar in Section~\ref{sec:solar:gci}. We then present an overview of the \solar framework including its basic working mechanism and the 
functionalities of its components in Section~\ref{sec:solar:frame}.

\subsection{Goals, Challenges and Insights}
\label{sec:solar:gci}

\paragraph{Design Goal}
As already discussed in Section~\ref{intro:contri},
\solar is designed to resolve reflection as soundly
as possible (i.e., more soundly or even soundly when 
some reasonable assumptions are met) and accurately identify the reflective calls resolved unsoundly.

\paragraph{Challenges}

In addition to the challenges described
in Section~\ref{intro:cha}, we must also
address another critical problem: it is hard to reason 
about the soundness of \solar and 
identify accurately which parts of the reflective code have been
resolved unsoundly.

If one target method at 
one reflective call is missed by the analysis, it
may be possible to identify the statements that are
unaffected and thus
still handled soundly. However, 
the situation will deteriorate sharply if many 
reflective calls are resolved unsoundly. In the worst
case, all the other statements in the program
may be handled unsoundly. To play safe,
the behaviors of all
statements must be assumed to be under-approximated in the analysis, as we do not know 
which value at which statement has been affected by the unsoundly resolved reflective calls.

\paragraph{Insights}
To achieve the design goals of \solar, we first need to ensure that as few reflective calls are resolved unsoundly as possible. This will reduce the propagation of
unsoundness to as few statements as possible in the
program.
As a result, if \solar reports that some analysis results are sound (unsound), then they are likely sound
(unsound) with high confidence. This is the key to enabling
\solar to achieve practical \emph{precision} in terms of both soundness reasoning and unsoundness identification.

To resolve most or even all reflective calls soundly, \solar needs to maximally leverage the available information (the string values at reflective calls are inadequate as they are often unknown statically) in the program to help resolve reflection. Meanwhile, \solar should resolve reflection precisely. Otherwise, \solar
may be unscalable due to too many false reflective 
targets introduced. In Sections~\ref{sec:collect} and \ref{sec:lhm}, we will describe how \solar leverages the \emph{self-inferencing property} in a program (Definition~\ref{def:self}) to analyze reflection with good soundness and precision. 

Finally, \solar should be aware of the conditions under which a reflective target cannot be resolved. In other words, we need to formulate a set of soundness criteria for different reflection methods based on different resolution strategies adopted. If the set of 
criteria is not satisfied, \solar can mark the corresponding reflective calls as the ones that are resolved unsoundly. Otherwise, \solar can determine the soundness of the reflection analysis under some reasonable assumptions (Section~\ref{sec:meth-ass}).

\subsection{The \textsc{Solar} Framework}
\label{sec:solar:frame}
Figure~\ref{fig:solar} gives an overview of \solar. 
\solar consists of four core components: an 
\emph{inference engine} for discovering 
reflective targets, an \emph{interpreter} for 
soundness and precision, a \emph{locater} for
unsound and imprecise calls, and a  
\probe (a lightweight version of \solar).
In the rest of this section, we first introduce the basic working mechanism of \solar and then briefly explain
the functionality of each of its components.

\subsubsection{Working Mechanism}
\label{sec:mech}

Given a Java program, the \emph{inference engine} resolves and infers the reflective targets that
are invoked or accessed at all side-effect method call sites in the program, as soundly as possible. 
There are two possible outcomes. If the reflection resolution is scalable (under a given time budget), the 
\emph{interpreter} will proceed to assess the quality of
the reflection resolution under soundness and precision criteria.
Otherwise, \probe, a lightweight version of \solar, would be called upon to analyze the same program again. As \probe resolves reflection less soundly but much more precisely than \solar, its scalability can be usually guaranteed. 
We envisage providing a range of \probe variants 
with different trade-offs among soundness, precision 
and scalability, so that the scalability
of \probe can be always guaranteed.

\begin{figure}[t]
\hspace{4ex}
\includegraphics[width=0.9\textwidth]{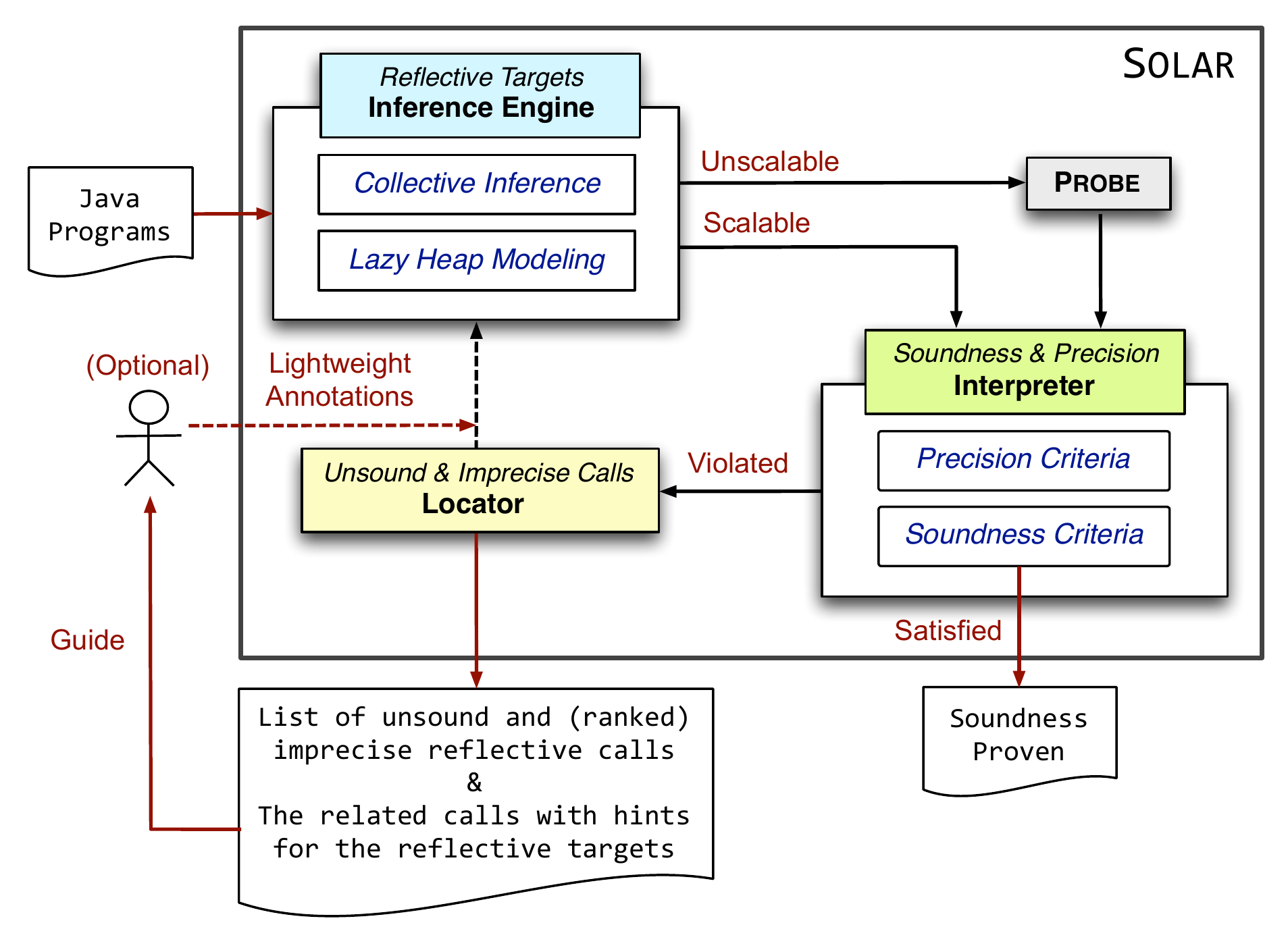}
\vspace{-2ex}
\caption{Overview of \solar.
\label{fig:solar}}
\end{figure}

If the \emph{interpreter} confirms that the soundness
criteria are satisfied, \solar reports that the reflection analysis is sound. Otherwise, the \emph{locater}
will be in action to identify which reflective calls
are resolved unsoundly.
In both cases,
the \emph{interpreter} will also report the reflective 
calls that are resolved imprecisely if the precision criteria are violated. This allows potential precision improvements to be made for the analysis.

The \emph{locater} not only outputs the list of reflective calls that are resolved unsoundly or imprecisely in
the program but also pinpoints the related entry and member-introspecting method calls of these ``problematic'' calls, which contain the hints to guide users to
add annotations, if possible. Figure~\ref{fig:output}
depicts an example output.

As will be demonstrated in Section~\ref{sec:eval}, for
many programs, \solar is able to resolve reflection 
soundly under some reasonable assumptions. However, for certain programs, like other existing
reflection analyses, \solar is unscalable. In this case, \probe (a lightweight version of \solar whose scalability can be guaranteed as explained above) is applied to analyze the same program.
Note that \probe is also able to identify the reflective calls that are resolved unsoundly or imprecisely in
the same way as \solar.
Thus, with some unsound or imprecise reflective 
calls identified by \probe and  annotated by users, \solar will 
re-analyze the program, scalably after one or more iterations of this ``probing'' process. As discussed
in Section~\ref{sec:eval}, the number of such iterations is usually small, e.g., only one is required for most of the programs evaluated.

For some programs, users may choose not to add 
annotations to facilitate reflection analysis. Even
in this case, users can still benefit from the \solar 
approach, for two reasons. First,
\probe is already capable of producing
good-quality reflection analysis results, more soundly
than string analysis. Second,
users can understand the quality of these 
results by inspecting the \emph{locater}'s output, 
as discussed in Section~\ref{intro:pre}.

\begin{figure}[t]
\includegraphics[width=1.02\textwidth]{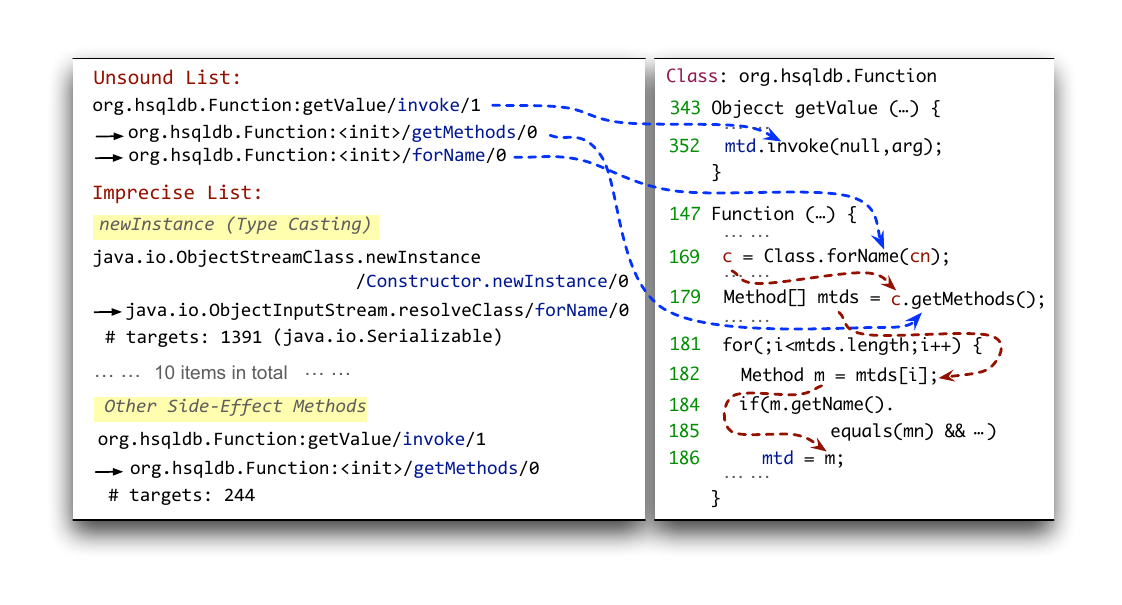}
\vspace{-6ex}
\caption{An example output from \solar when its
	soundness/precision criteria are violated.
\label{fig:output}}
\end{figure}

\subsubsection{Basic Components}
Their functionalities 
are briefly explained below.

\paragraph{Reflective Target Inference Engine}

We employ two techniques to discover reflective targets: \emph{collective inference} for resolving
reflective method invocations
(\texttt{invoke()}) and field accesses/modifications
(\texttt{get()/set()}) and \emph{lazy heap modeling} for handling reflective object creation (\texttt{newInstance()}).
Both techniques exploit the \emph{self-inferencing property} found in our reflection usage study (Section~\ref{sec:study:results}) to resolve reflection in a disciplined manner
with good soundness and precision.
We will explain their approaches in Sections~\ref{sec:collect} and~\ref{sec:lhm}, respectively, and further formalize them in Section~\ref{form:infer}.

\paragraph{Soundness and Precision Interpreter}
\solar currently adopts a simple but practical scheme to measure precision,  in terms of the number of targets resolved at a side-effect call site. \solar allows users to specify a threshold value in advance to define the imprecision that can be tolerated for each kind of side-effect calls, forming its precision criteria.
Its soundness criteria are formulated in terms of 
conditions under which various inference rules (adopted by
the inference engine) can be applied soundly.
We will formalize the soundness criteria required
in Section~\ref{sec:cond}.

\paragraph{Unsound and Imprecise Call Locater}

A side-effect reflective call is identified as being
imprecisely resolved if the number of resolved targets
is higher than permitted by its 
corresponding precision criterion. Similarly, a 
side-effect reflective call is marked
as being unsoundly resolved if its corresponding soundness 
criterion is violated.

To facilitate user annotations for an imprecisely or
unsoundly resolved side-effect
reflective call, the locater also pinpoints its
corresponding
entry and member-introspecting method call sites. 
It can be difficult to understand the semantics of a reflective call by reading just the code at its vicinity.
Often, more hints about its
semantics are available at or around its
entry and member-introspecting method call sites, which may reside in different methods or even classes in the program.

\begin{example}
\label{example:anno}
Figure~\ref{fig:output} illustrates \solar's output
for a real program. The \texttt{invoke()} (side-effect method) call site in method \texttt{getValue} is the 
unsoundly resolved call identified. Its entry method (\texttt{forName()}) and member-introspecting method (\texttt{getMethods()}) call sites, which are located in the constructor of class \texttt{org.hsqldb.Function}, are also highlighted. 
At the right-hand side of the figure, we can see that the hints for annotations are available around the entry and member-introspecting call sites (e.g., lines 169, 184 and 185) rather than the side-effect call site (line 352). This demonstrates the usefulness of \solar's annotation strategy.
\end{example}

We will further explain how \solar identifies unsoundly resolved reflective calls in Section~\ref{sec:unsound} and how users are guided to add annotations in Section~\ref{sec:anno}.

\paragraph{\probe}
\probe is a lightweight version of \solar by weakening the power of its inference engine. \probe changes its inference strategies in both collective inference and lazy heap modeling, by resolving reflection  more precisely but less soundly. Thus,
the scalability of \probe can be usually guaranteed as 
fewer false reflective targets are introduced. 
We will formalize \probe based on the formalism of \solar in Section~\ref{sec:probe}.

\section{The \solar Methodology}
\label{sec:meth}

We first define precisely a set of assumptions made
(Section~\ref{sec:meth-ass}). Then we examine the methodologies
of collective inference (Section~\ref{sec:collect}) and lazy heap modeling (Section~\ref{sec:lhm}) used in \solar's inference engine.
Finally, we explain how \solar identifies unsoundly resolved reflective calls (Section~\ref{sec:unsound}) and 
how doing so helps guide users to add 
lightweight annotations to facilitate a 
subsequent reflection analysis (Section~\ref{sec:anno}).

\subsection{Assumptions}
\label{sec:meth-ass}
There are four reasonable assumptions. The first one is commonly made on static analysis~\cite{closed14} and the next two are made previously on reflection analysis for Java~\cite{Livshits05}. 
\solar adds one more assumption to allow reflective allocation sites to be modeled lazily. Under the four
assumptions, it becomes possible to reason about the
soundness and imprecision of \solar.

\begin{assumption}[Closed-World]
\label{ass:world}
Only the classes reachable 
from the class path at analysis time can be used 
during program execution.
\end{assumption}

This assumption is reasonable 
since we cannot expect static analysis to 
handle all classes 
that a program may download from the Internet and load at runtime. In addition, Java native methods are
excluded as well. 

\begin{assumption}[Well-Behaved Class Loaders] 
\label{ass:loader}
The name of the class returned by a call to 
\texttt{Class.forName(cName)} equals \texttt{cName}.
\end{assumption}

This assumption says that the class to be loaded by \texttt{Class.forName(cName)} is the expected  one specified by the value of \texttt{cName}, thus avoiding handling the situation where a different class is loaded by, e.g., a malicious custom class loader.  How to handle custom class loader statically is still an open hard problem. Note that this assumption also applies to \texttt{loadClass()}, another entry method shown in
Figure~\ref{study:fig:coreAPI}.

\begin{assumption}[Correct Casts]
\label{ass:cast}
Type cast operations applied to the results of calls
to side-effect methods are correct, without throwing a ClassCastException.
\end{assumption}

This assumption has been recently demonstrated as practically valid through extensive experiments in~\cite{Barros17}.

\begin{assumption}[Object Reachability]
\label{ass:reach}
Every object \texttt{o} created reflectively in a call to 
\texttt{newInstance()} flows into (i.e., will be used in) either (1)  a
type cast operation \texttt{\ldots = (T) v}
or (2) a call to a side-effect method,
\texttt{get(v)},
\texttt{set(v,\ldots)} or
\texttt{invoke(v,\ldots)}, where \texttt{v} points to
\texttt{o}, along every execution path in the program.
\end{assumption}

Cases (1) and (2) represent two
kinds of usage points at which the class types of object 
$o$ will be inferred lazily. 
Specifically, case (1) indicates that $o$ is used as a regular object, and case (2) says that $o$ is used reflectively, i.e., flows to the first argument of different side-effect calls as a receiver object.
This assumption does not cover only one rare situation where $o$
is created but never used later.
As validated in Section~\ref{sec:eval-ass},
Assumption~\ref{ass:reach} is found to hold for
almost all reflective allocation sites in the real code.

\subsection{Collective Inference}
\label{sec:collect}

Figure~\ref{fig:collect} gives an overview of 
collective inference
for handling reflective method invocations and field 
accesses/modifications.
Essentially, we see how 
the side-effect method calls \texttt{invoke()}, \texttt{get()}
and \texttt{set()} are resolved.  A \texttt{Class} object
\stextmath{C} is first created for the target class named
\texttt{cName}. Then a
\texttt{Method} (\texttt{Field}) object 
\stextmath{M} 
(\stextmath{F}) 
representing the target method (field) named 
\texttt{mName} (\texttt{fName})
in the target class of \stextmath{C} is created. Finally, 
at some reflective call sites, e.g., 
\texttt{invoke()},
\texttt{get()} and
\texttt{set()}, the target method (field) is invoked
(accessed) on the target object \texttt{o}, with the
arguments, \texttt{\{...\}} or \texttt{a}.

\solar works as part of a pointer analysis, with each being
both the producer and consumer of the other.
By exploiting the self-inferencing property (Definition~\ref{def:self}) inherent in the reflective code, \solar 
employs the following two component analyses:
\begin{description}
\item[Target Propagation (Marked by Solid Arrows)]
\solar resolves the targets (methods or fields) of 
reflective calls, \texttt{invoke()}, \texttt{get()} and \texttt{set()}, 
by propagating the names of their
target classes and methods/fields (e.g., those pointed by
\texttt{cName}, \texttt{mName} and \texttt{fName}
if statically known)
along the solid lines into the points symbolized by circles. 


\begin{figure}
\includegraphics[width=1\textwidth]{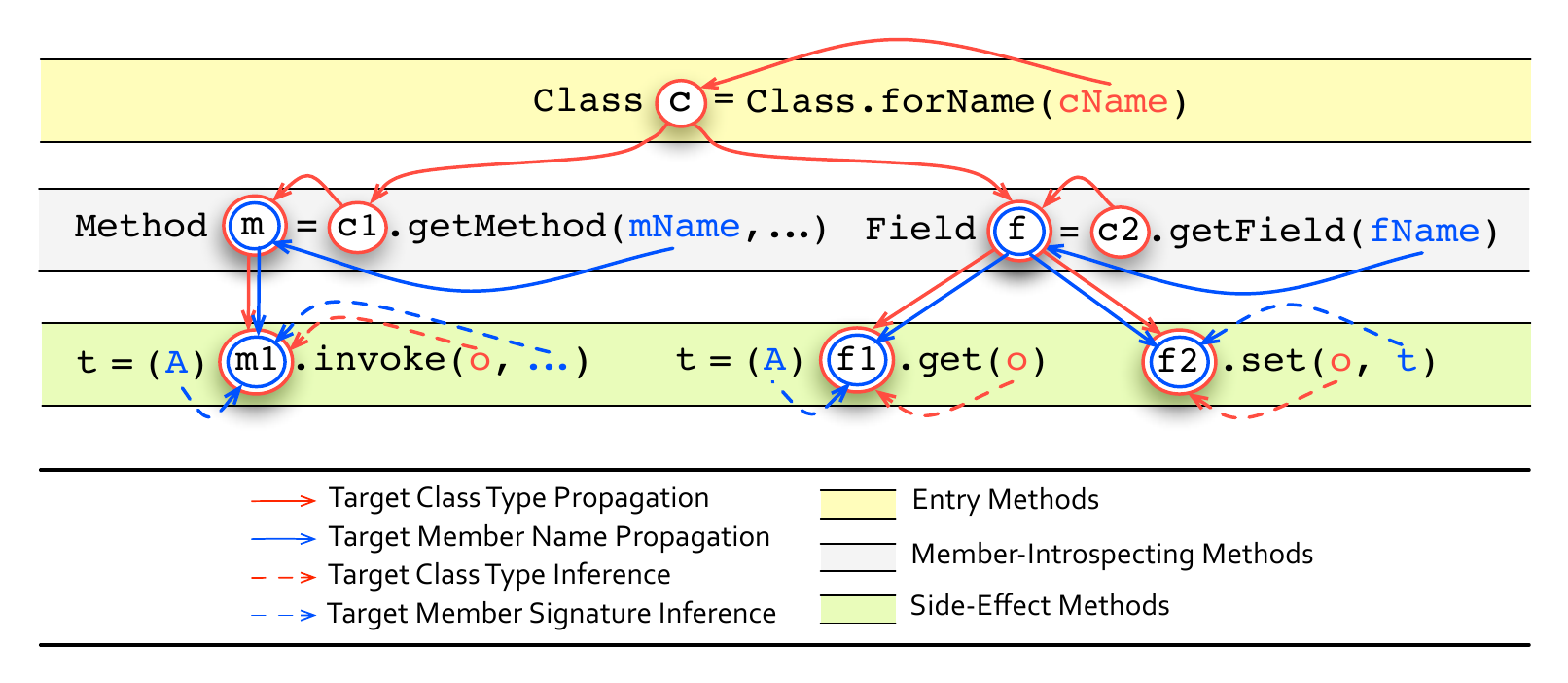}
\caption{Collective Inference in \solar.
	\label{fig:collect}}
\end{figure}

\vspace{1ex}
\item[Target Inference (Marked by Dashed Arrows)] By using
\emph{Target Propagation} alone, 
a target member name (blue circle) or its 
target class type (red circle) 
at a reflective call site
may be missing, i.e., unknown,
due to the presence of input-dependent
strings (Figure~\ref{study:fig:str}). 
If the target class type (red circle) is missing, \solar will infer it
from the dynamic type of the target object \texttt{o} 
(obtained by pointer analysis)
at \texttt{invoke()}, \texttt{get()} or \texttt{set()}
(when \texttt{o != null}). 
If the target 
member name (blue circle) is missing, \solar will 
infer it from (1) the dynamic types of the arguments 
of the target call, e.g., 
\texttt{\{...\}} of \texttt{invoke()} and
\texttt{a} of \texttt{set()}, and/or (2) 
the downcast on the result of 
the call, such as \texttt{(A)} at 
\texttt{invoke()} and \texttt{get()}. 

\end{description}

\begin{example}
Let us illustrate \emph{Target Inference} by considering
\texttt{r = (A) \stextmath{F}.get(o)}
in Figure~\ref{fig:collect}. If a target field name
is known but its target class type (i.e., red circle) 
is missing, we can infer it from the 
types of all pointed-to objects $o'$ by \texttt{o}.
If $B$ is one such a type, then a 
potential target class of
\texttt{o} is $B$ or any of its supertypes. If the target
class type of \stextmath{F} is $B$ but
a potential target
field name (i.e., blue circle) is missing, we
can deduce it from the downcast \texttt{(A)}
to resolve the call to 
\verb"r = o.f", where \verb"f" is a member field in $B$
whose type is \verb"A" or a supertype or subtype of 
\verb"A".  A supertype is possible because a field of
this supertype may initially point to an object of type
\texttt{A} or a subtype of \texttt{A}.
\end{example}

In Figure~\ref{fig:collect}, if 
\texttt{getMethods()} (\texttt{getFields()}) is called
as a member-introspecting method
instead, then an array of \texttt{Method} (\texttt{Field})
objects will be returned so that \emph{Target
Propagation} from it is implicitly performed by pointer analysis. All the other
methods in \texttt{Class} for introspecting
methods/fields/constructors are handled similarly.

\paragraph{Resolution Principles}
To balance soundness, precision and scalability in a 
disciplined manner, collective inference resolves the targets at a side-effect method call site (Figure~\ref{fig:collect}) if and only if one of the following three conditions is met:
\begin{itemize}
\item Both its target class type (red circle) and 
target member name (blue circle) are made available by target propagation (solid arrow) \emph{or} target inference (dashed arrow).
\item Only its target class type (red circle) is made
	available by target propagation (solid arrow) \emph{or} target inference (dashed arrow).
\item Only its target member name (blue circle) is 
	made available by \emph{both}
target propagation (solid arrow) \emph{and} target inference (dashed arrow).
\end{itemize}

In practice, the first condition is met by
many calls to \texttt{invoke()}, \texttt{get()} and \texttt{set()}. In this case,
the number of spurious targets introduced 
can be significantly reduced due to the simultaneous
enforcement of two constraints (the red and blue circles). 

To increase the inference power of \solar,
as explained in Section~\ref{sec:solar:gci}, we will
also resolve a side-effect call under the one of
the other two conditions (i.e., when only one circle is available).
The second condition requires only its target
class type to be inferable, as a 
class (type) name that is prefixed with its
package name is usually unique.  However, 
when only its target member name (blue circle)
is inferable, we insist both its name 
(solid arrow) and its descriptor (dashed arrow)
are available. In a large program,
many unrelated classes may happen to use the same 
method name. Just relying only the name of a 
method lone in the last condition may cause imprecision.

If a side-effect call does not satisfy any of 
the above three conditions, then
\solar will flag it as being
unsoundly resolved, as described in Section~\ref{sec:unsound}. 

\subsection{Lazy Heap Modeling}
\label{sec:lhm}

\begin{figure}
\includegraphics[width=1\textwidth]{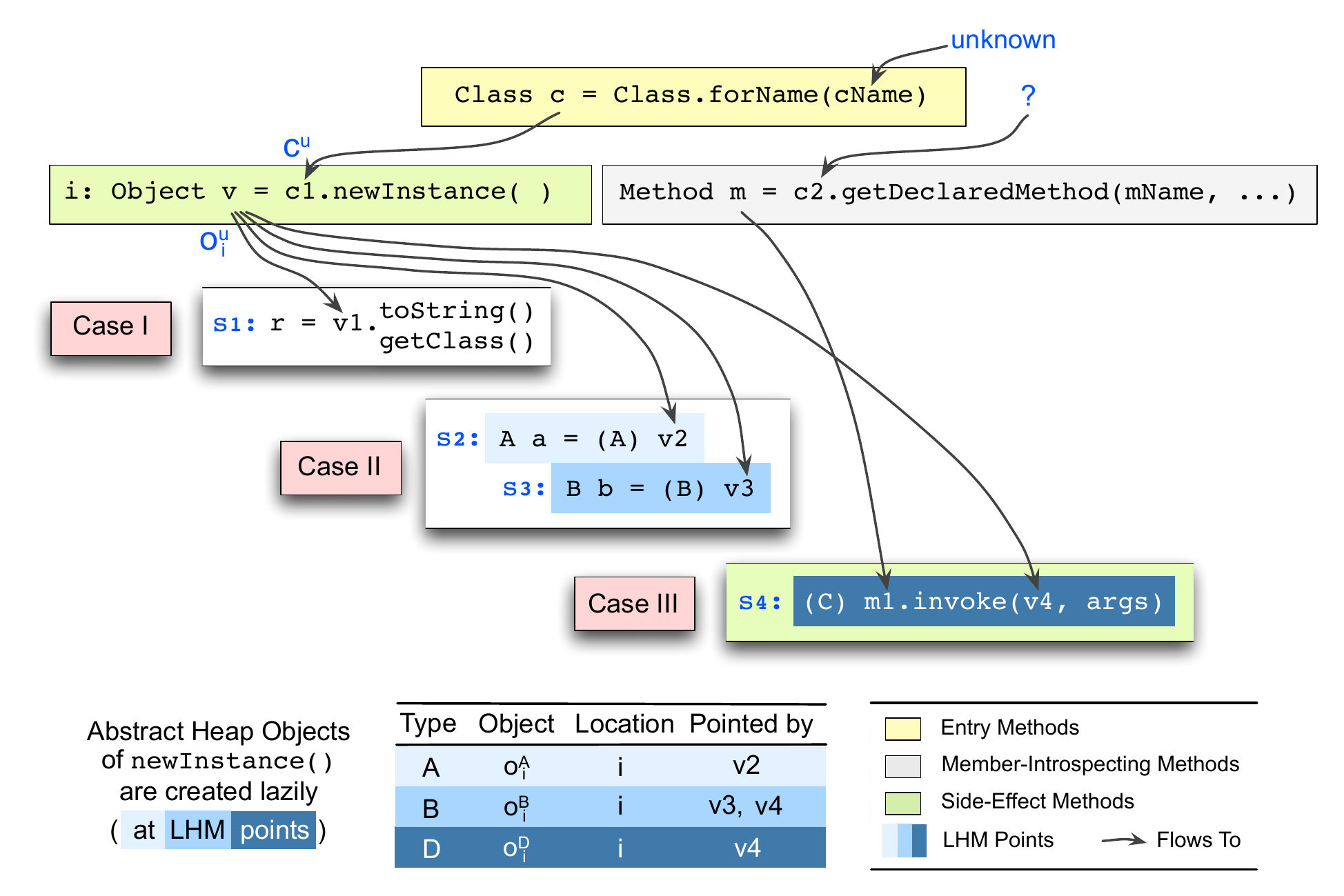}
\caption{Lazy heap modeling (LHM).
The abstract objects,
$o_{\texttt i}^{\texttt A}$,
$o_{\texttt i}^{\texttt B}$ and $o_{\texttt i}^{\texttt D}$, for 
\texttt{newInstance()} are created lazily at  
the two kinds of LHM (usage) points in Cases (II) and (III), where
\texttt{A} and \texttt{B} have
no subtypes and 
\texttt{m1} is declared in
\texttt{D} with one subtype
\texttt{B}, implying that the dynamic types of the
objects pointed by \texttt{v4} is \texttt{D} or 
\texttt{B}.
	\label{fig:lhm}}
\end{figure}

As shown in Section~\ref{sec:study:results}, reflective object creation, i.e., \texttt{newInstance()} is the most widely used side-effect method.
Lazy heap modeling (LHM), illustrated in
Figure~\ref{fig:lhm},
is developed to facilitate its target inference and the
soundness reasoning for \solar.

There are three cases. Let us consider Cases (II) and
(III) first.
Usually, an object, say \texttt{o}, created by \texttt{newInstance()} will be used later either regularly or reflectively as shown in Cases (II) and (III),
respectively. 
In Case (II), since the declared type of \texttt{o} is \texttt{java.lang.Object}, 
$o$ is first cast to a specific type before used for
calling methods or accessing fields as a 
regular object. 
Thus, \texttt{o} will flow to some cast operations.
In Case (III), \texttt{o} is used in a reflective way, i.e., as the first argument of a call to a side-effect method, \texttt{invoke()}, \texttt{get()} or \texttt{set()}, on which the target method (field) is called (accessed). This appears to be especially commonly used in Android apps. 

For these two cases, we can leverage the information at \texttt{o}'s usage sites to infer its type lazily and also make its corresponding effects (on static analysis) visible there. As for the (regular) side-effects that may be made by \texttt{o} along the paths from \texttt{newInstance()} call site to its usages sites, we use Case (I) to cover this situation. 

Now, we examine these Cases (I) -- (III), which are highlighted in Figure~\ref{fig:lhm}, one by one, in more
detail.
If \texttt{cName} at \texttt{c = Class.forName(cName)}
is unknown, \solar will 
create a \texttt{Class} object \texttt{c$^u$} that 
represents this unknown class and assign it to 
\texttt{c}. On discovering that \texttt{c1}
points to a \texttt{c$^u$} at an allocation site 
\texttt{i} (\texttt{v = c1.newInstance()}), \solar will 
create
an abstract object \texttt{o$^u_{\tt i}$} of an unknown type
for the site to mark it as being
unresolved yet. Subsequently,
\texttt{o$^u_{\tt i}$} will flow into Cases (I) -- (III).

In Case (I), the returned type of
\texttt{o$^u_{\tt i}$}
is declared as \texttt{java.lang.Object}. Before 
\texttt{o$^u_{\tt i}$} flows to a cast operation,
the only side-effect that can be made by this object is to call some methods declared in \texttt{java.lang.Object}. In terms of reflection analysis, only the two 
pointer-affecting methods shown 
in Figure~\ref{fig:lhm}
need to be considered. \solar handles both 
soundly, by
returning (1) an unknown string for \texttt{v1.toString()}
and (2) an unknown \texttt{Class} object for 
\texttt{v1.getClass()}. Note that \texttt{clone()} 
cannot be called on \texttt{v1} of type \texttt{java.lang.Object}
(without a downcast being performed on \texttt{v1} first).

Let us consider Cases (II) and (III),
where each 
statement, say $S_x$, is called an \emph{LHM point}, containing a variable $x$ into which
\texttt{o$^u_{\tt i}$} flows. In Figure~\ref{fig:lhm},
we have
$x\in\{\texttt{v2}, \texttt{v3}, \texttt{v4}\}$.
Let \lhm$(S_x)$ be the set of class types discovered
for the unknown class $u$ at $S_x$ by inferring from the cast operation 
at $S_x$ as in Case (II) or the information available
at a call to \texttt{(C) m1.invoke(v4, args)} (e.g., on
\texttt{C}, \texttt{m1} and \texttt{args}) as in Case
(III). For example, given
\texttt{\tt S2$_{\tt v2}$: A a = (A) v2},
\lhm(\texttt{\tt S2$_{\tt v2}$}) contains 
\texttt{A} and its subtypes. 
To account for 
the side-effects of \texttt{v = c1.newInstance()}
at $S_x$ lazily, we add (conceptually) a statement,
\texttt{x = new T()}, for every
$T\in\lhm(S_x)$, before $S_x$. Thus, 
\texttt{o$^u_{\tt i}$} is finally split into and thus aliased
with $n$ distinct abstract objects,
\texttt{o$_{\tt i}^{T_1}$},\dots,
\texttt{o$_{\tt i}^{T_n}$}, where $\lhm(S_x)=\{
T_1,\dots,T_n\}$, such
that $x$ will be made to point to all these new abstract objects.

Figure~\ref{fig:lhm} illustrates lazy heap modeling
for the case when neither \texttt{A}
nor \texttt{B} has subtypes and the declaring class for
\texttt{m1} is discovered to be \texttt{D}, which has 
one subtype \texttt{B}. Thus, \solar will deduce that
$\lhm(\tt S2_{\tt v2})=\{\texttt{A}\}$,
$\lhm(\tt S3_{\tt v3})=\{\texttt{B}\}$ and
$\lhm(\tt S4_{\tt v4})=\{\texttt{B,D}\}$. 
Note that in Case (II), 
\texttt{o$^u_{\tt i}$} will not flow to
\texttt{a} and \texttt{b} due to the cast operations.

As \texttt{java.lang.Object} contains no fields, all field
accesses to \texttt{o$^u_{\tt i}$} will only be made
on its lazily created objects. Therefore, if the same
concrete object represented by
\texttt{o$^u_{\tt i}$} flows to both 
$S_{x_1}$ and $S_{x_2}$, then $ \lhm(S_{x_1}) \cap
\lhm(S_{x_2})\neq \emptyset$.
This implies that $x_1$ and $x_2$ will point to a common
object lazily created. For example, in 
Figure~\ref{fig:lhm}, \texttt{v3} and \texttt{v4} points
to $o_{\tt i}^{\tt B}$
since $\lhm(S3_{\tt v3}) \cap
\lhm(S4_{\tt v4})=\{{\tt o}^{\tt B}_{\tt i}\}$. 
As a result, the alias relation
between $x_1.f$ and $x_2.f$ is correctly
maintained, where $f$ is a field of \texttt{o$^u_{\tt i}$}.

\begin{figure}[h]
\centering
\includegraphics[width=1\textwidth]{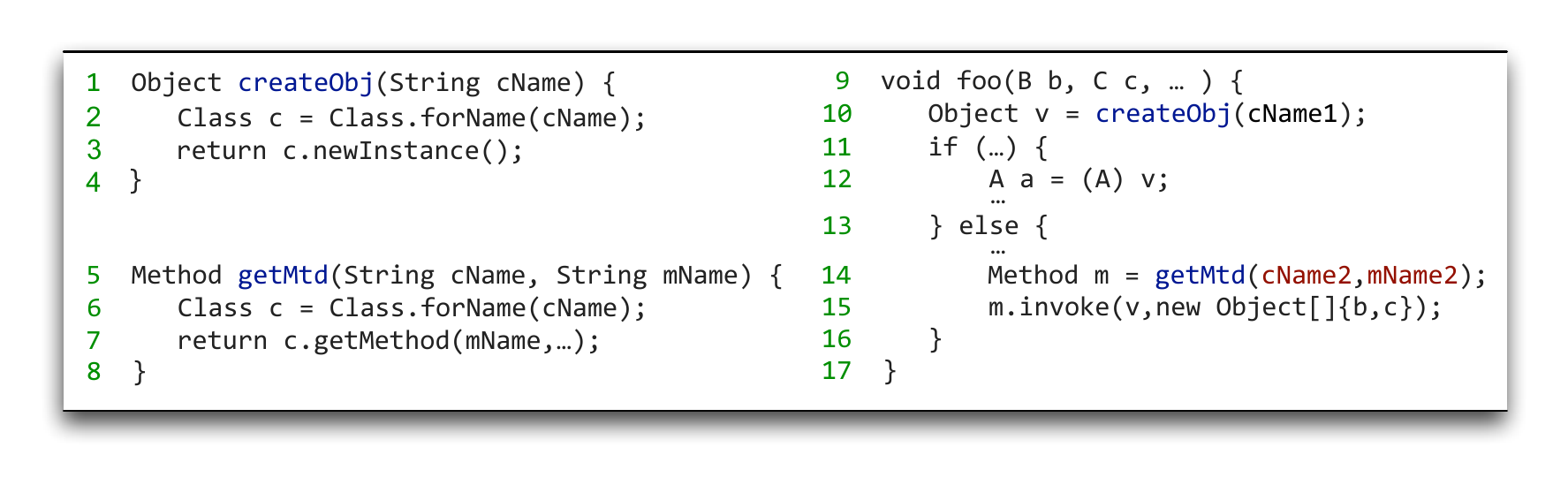}
\vspace{-4ex}
\caption{An example for illustrating LHM in \solar.
	\label{fig:mot}}
\end{figure}

\begin{example}
\label{example:mot}
In Figure~\ref{fig:mot},
\solar will model the
\texttt{newInstance()} call in line 3 lazily (as
\texttt{cName1} is statically unknown) by returning an
object $o_3^u$ of an unknown type $u$. Note that $o_3^u$
flows into two kinds of usage points: 
the cast operation in line 12 and the \texttt{invoke()}
call in line 15. In the former case, \solar will infer
$u$ to be \texttt{A} and its subtypes in line 12. In the latter
case, \solar will infer $u$ based on the information 
available in line 15 by distinguishing three cases. (1) If
\texttt{cName2} is known, then \solar deduces $u$ 
from the known class in \texttt{cName2}. (2) If
\texttt{cName2} is unknown but \texttt{mName2} is known,
then \solar deduces
$u$ from the known method name in
\texttt{mName2} and the second argument 
\texttt{new Object[] \{b,c\}} of the \texttt{invoke()}
call site. 
(3) If both \texttt{cName2}
and \texttt{mName2} are unknown (given that the 
types of $o_3^u$ are already unknown), 
then \solar will flag the \texttt{invoke()}
call in line 15 as being unsoundly resolved, 
detected automatically by verifying one of the soundness criteria, i.e., Condition (\ref{eq:sound-inv}) in Section~\ref{sec:cond}. 
\end{example}

\paragraph{Discussion} 
Under Assumption~\ref{ass:reach}, we need only to handle
the three cases in Figure~\ref{fig:lhm}
in order to establish whether a \texttt{newInstance()}
call has been modeled soundly or not. 
The rare exception (which breaks Assumption~\ref{ass:reach}) is that \texttt{o$^u_{\tt i}$} 
is created but never used later (where no hints are available).
To achieve soundness in this rare case, the corresponding
constructor (of the dynamic type of \texttt{o$^u_{\tt i}$}) must be annotated to be analyzed statically unless
ignoring it will not affect the points-to information to be
obtained. Again, as validated in Section~\ref{sec:eval-ass}, Assumption~\ref{ass:reach} is found 
to be very practical.

\subsection{Unsound Call Identification}
\label{sec:unsound}

Intuitively, we mark a side-effect reflective
call as being
unsoundly resolved when \solar has exhausted all
its inference strategies to resolve it, but to no
avail.
In addition to Case (3) in Example~\ref{example:mot},
let us consider another case in
Figure~\ref{fig:lhm}, except that
\texttt{c2} and \texttt{mName} are assumed to be
unknown. Then \texttt{m1} at \texttt{s4}:
\texttt{m1.invoke(v4, args)} will be unknown. 
\solar will mark it as unsoundly resolved, since
just
leveraging \texttt{args} alone to infer its target
methods may cause \solar to be too imprecise to scale
(Section~\ref{sec:collect}).


The formal soundness criteria that are 
used to identify unsoundly resolved reflective calls 
are defined and illustrated in Section~\ref{sec:cond}.

\subsection{Guided Lightweight Annotation}
\label{sec:anno}

As shown in Example~\ref{example:anno}, \solar can
guide users 
to the program points where hints for annotations are potentially available for unsoundly or imprecisely resolved reflective calls.
As these ``problematic'' call sites are the places in a program where side-effect methods are invoked, we can hardly extract the information there to know the  names of the reflective targets, as they are specified at the corresponding entry and member-introspecting call sites (also called the annotation sites), which 
may not appear in the same method or class (as the 
``problematic call sites'').
Thus, \solar is designed to automatically track the flows of metaobjects from the identified ``problematic'' call sites in a demand-driven way to locate all the related annotation sites.

In \solar, we propose to add annotations for unsoundly
resolved side-effect call sites, which are often
identified accurately. As a result,
the number of required annotations (for achieving
soundness) would be significantly less than that
required in~\cite{Livshits05}, which simply asks for annotations when the string argument of a reflective call is statically unknown. 
This is further validated in Section~\ref{sec:rq1}.

\section{Formalism}
\label{sec:form}

We formalize \solar, as illustrated in Figure~\ref{fig:solar},
for \RefJava, which is Java 
restricted to a core subset of its reflection API. 
\solar is flow-insensitive but context-sensitive. 
However, our formalization is context-insensitive for simplicity.
We first define \RefJava (Section~\ref{sec:lang}), give a 
road map for the formalism (Section~\ref{map}) and present 
some notations used (Section~\ref{def}). We then introduce
a set of 
rules for formulating collective 
inference and lazy heap modeling in \solar's inference engine (Section~\ref{form:infer}). Based on these rules,
we formulate a set of soundness criteria 
(Section~\ref{condition}) that enables reasoning about the
soundness of \solar (Section~\ref{properties}).
Finally, we 
describe how to instantiate \probe from \solar (Section~\ref{sec:probe}), and handle static
class members (Section~\ref{other}). 

\subsection{The \RefJava Language}
\label{sec:lang}
\label{sec:refjava}

\refjava consists of all Java programs (under
Assumptions~\ref{ass:world} -- \ref{ass:reach}) except that
the Java reflection API is restricted to the seven core reflection methods: 
one entry method \texttt{Class.forName()}, two member-introspecting methods \texttt{getMethod()} and \texttt{getField()}, and four side-effect methods for reflective object creation \texttt{newInstance()}, reflective method invocation \texttt{invoke()}, reflective fields access \texttt{get()} and modification \texttt{set()}.

Our formalism is designed to allow its straightforward
generalization to the entire Java reflection API.
As is standard, a Java
program is represented only by five kinds of statements in
the SSA form, as shown in Figure~\ref{fig:pta}.
For simplicity, we assume that all the members (fields or
methods) of a class accessed reflectively are its
instance members, i.e.,
$\texttt{o} \neq \texttt{null}$
in \texttt{get(o)}, \texttt{set(o,a)} and 
\texttt{invoke(o,\ldots)} in Figure~\ref{fig:collect}. We
will formalize how to handle static members in
Section~\ref{other}.

\subsection{Road Map}
\label{map}

\begin{figure}[hbtp]
\centering
\includegraphics[width=0.7\textwidth]{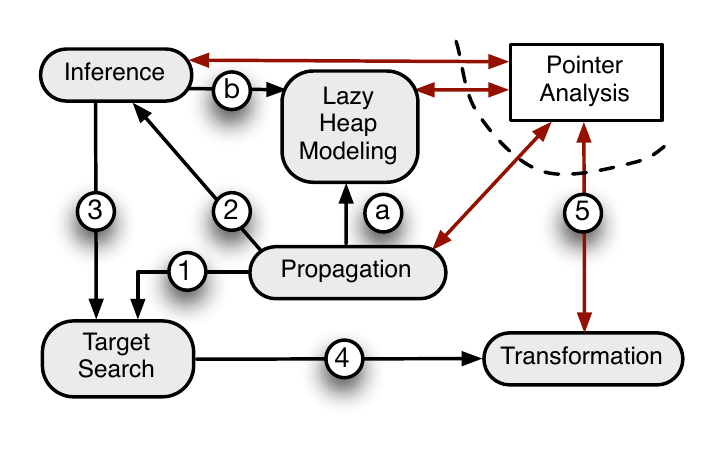}
\caption{\solar's inference engine: five components and
their inter-component dependences (depicted by black arrows). The dependences between \solar and pointer analysis are depicted in red arrows.
\label{fig:form}}
\end{figure}

As depicted in Figure~\ref{fig:form}, \solar's inference engine, 
which consists of five components,
works together with a pointer analysis. The arrow
\textcolor{BrickRed}{$\longleftrightarrow$} between a
component and the pointer analysis means that each
is both a producer and consumer of the other.

Let us take an example to see how this road map works.
Consider the side-effect call
\texttt{t = \texttt{f1}.get(o)} in Figure~\ref{fig:collect}. 
If \texttt{cName} and \texttt{fName} are 
string constants,
\emph{Propagation} will create a \texttt{Field} object
(pointed to by \texttt{f1}) carrying its known class 
and field information and pass it to 
\emph{Target Search} (\circled{1}). 
If \texttt{cName} or \texttt{fName} is not a 
constant, a \texttt{Field} object  marked as such
is created and passed to \emph{Inference} 
(\circled{2}), which will 
infer the missing information and pass
a freshly generated \texttt{Field} object enriched
with the missing information
to \emph{Target Search} 
(\circled{3}). 
Then \emph{Target Search} maps a \texttt{Field} object to
its reflective target $f$ in its declaring class
(\circled{4}).
Finally, \emph{Transformation} turns the reflective call 
\texttt{t = \texttt{f1}.get(o)} into a regular statement 
\texttt{t = o.$f$} and pass it to
the pointer analysis (\circled{5}). 
Note that \emph{Lazy Heap Modeling} handles \texttt{newInstance()} 
based on the information discovered
by \emph{Propagation} (\circled{a})
or  \emph{Inference} (\circled{b}).

\subsection{Notations}
\label{def}

In this paper, a field signature consists of the field
name and descriptor (i.e., field type), and a field is 
specified by its field signature and the class where it
is defined (declared or inherited).
Similarly, a method signature consists of the method name
and descriptor (i.e., return type and parameter types) and, 
a method is specified by its method signature and the class 
where it is defined.

\begin{figure}[th]
\centering
{
\begin{tabular}{p{5cm}p{5cm}}
class type & $t\in \mathbb{T}$ \\
\texttt{Field} object* & $\texttt{f}^t_s$, $\texttt{f}^t_u$,
$\texttt{f}^u_s$, $\texttt{f}^u_u$ $\in$
$\mathbb{FO}$ = $\mathbb{\widehat{T}}$ $\times$ $\mathbb{S}_f$\\
field/method name & $n_f, n_m\in \mathbb{N}$ \\
field signature* & $s\in \mathbb{S}_f$ = $\mathbb{\widehat{T}}$ $\times$ $\mathbb{\widehat{N}}$\\
field & $ f \in \mathbb{F} = \mathbb{T} \times \mathbb{T} \times \mathbb{N}$\\
field type* & $s.t_f$ $\in$ $\mathbb{\widehat{T}}$\\
parameter (types) & $p \in \mathbb{P} =
\mathbb{T}^0 \cup \mathbb{T}^1 \cup \mathbb{T}^2 \dots $\\
field name*  & $s.n_f$ $\in$ $\mathbb{\widehat{N}}$\\
method & $ m \in \mathbb{M} = \mathbb{T} \times \mathbb{T} \times \mathbb{N} \times \mathbb{P}$\\
\texttt{Method} object* & $\texttt{m}^t_s$,
$\texttt{m}^t_u$, 
$\texttt{m}^u_s$,
$\texttt{m}^u_u$ $\in$
$\mathbb{MO}$ = $\mathbb{\widehat{T}}$ $\times$ $\mathbb{S}_m$\\
local variable & $ {\rm c}, {\rm f },{\rm m} \in \mathbb{V}$\\
method signature* & $s\in \mathbb{S}_m$ = $\mathbb{\widehat{T}}$ $\times$ $\mathbb{\widehat{N}}$ $\times$ $\mathbb{\widehat{P}}$\\
Abstract heap object & $o_1^t, o_2^t,\dots, o_1^u, o_2^u,\dots \in \mathbb{H}$\\
return type* & $s.t_r$ $\in$ $\mathbb{\widehat{T}}$\\
unknown  & $u$\\
method name*  & $s.n_m$ $\in$ $\mathbb{\widehat{N}}$\\
\texttt{Class} object & $\texttt{c}^t$, $\texttt{c}^u$ $\in$ $\mathbb{CO}$\\
parameter* & $s.p$ $\in$ $\mathbb{\widehat{P}}$\\
\end{tabular}
}
\caption[Notations.]{Notations. Here 
$\widehat{X} = X \cup \{u\}$,
where $u$ is an unknown class type or an unknown
field/method signature. A superscript `*' marks a 
domain that contains $u$.
\label{fig:def}
}
\end{figure}

We will use the notations given in Figure~\ref{fig:def}.
$\mathbb{CO}$, $\mathbb{FO}$ and $\mathbb{MO}$ represent the 
set of \texttt{Class}, \texttt{Field} and \texttt{Method} 
objects, respectively. 
In particular, $\texttt{c}^t$ denotes a
\texttt{Class} object of a known class $t$ 
and $\texttt{c}^u$ denotes a \texttt{Class} object of an unknown class $u$. 
As illustrated in Figure~\ref{fig:lhm}, we write
$o_i^t$ to represent an abstract object created at
an allocation site $i$ if it is an instance of a known class
$t$ and $o_i^u$ of (an unknown class type) otherwise.
For a \texttt{Field} object, we write
$\texttt{f}^t_s$ if it is a field defined in
a known class $t$
and $\texttt{f}^u_s$ otherwise, 
with its signature being $s$. 
In particular, we write $\texttt{f}^\phd_u$ for
$\texttt{f}^\phd_s$ in the special case
when $s$ is unknown, i.e.,
$s.t_f=s.n_f=u$.
Similarly, 
$\texttt{m}^t_s$,
$\texttt{m}^t_u$,
$\texttt{m}^u_s$ and
$\texttt{m}^u_u$ are used to represent
\texttt{Method} objects. We write
$\texttt{m}^\phd_u$ for $\texttt{m}^\phd_s$ 
when $s$ is 
unknown (with the return type
$s.t_r$ being irrelevant, i.e., either known or
unknown), i.e., $s.n_m= s.p=u$.

\vspace{-1ex}
\subsection{The Inference Engine of \solar}
\label{form:infer}

We present the inference rules used by all the components
in Figure~\ref{fig:form}, starting with the pointer
analysis and moving to the five components of \solar.
Due to their cyclic dependencies, the reader is invited to
read ahead sometimes, particularly to
Section~\ref{sec:lhm-rules} on LHM, before
returning back to the current topic.

\subsubsection{Pointer Analysis}
\label{pta}

Figure~\ref{fig:pta} gives a standard formulation of 
a flow-insensitive Andersen's pointer 
analysis for \RefJava.
\emph{pt(x)} represents the \emph{points-to set} of 
a pointer \emph{x}. An array object is analyzed 
with its elements collapsed to a single field, denoted $arr$.
For example, \texttt{x[i] = y} can be seen as 
\texttt{x.$arr$ = y}.
In \rulenameT{A-New}, $o_i^t$ uniquely identifies the abstract
object created as an instance of $t$ at this allocation site, 
labeled by \emph{i}. In \rulenameT{A-Ld} and \rulenameT{A-St}, only the
fields of an abstract object $o_i^t$ of a known type $t$ 
can be accessed. In Java, as explained in
Section~\ref{sec:lhm}, the field accesses to $o_i^u$ 
(of an unknown type) can only
be made to the abstract objects of known types
created lazily
from $o_i^u$ at \dam points.

\begin{figure}[htbp]
\small
\setlength{\tabcolsep}{3ex}
{
\begin{tabular}{ll}
\ruledef{\color{blue}{i:\; \text{x} = new \; t()}}{\{o^t_i\} \in pt($x$)}
\rulename{A-New}
&
\ruledef{\color{blue}{\text{x} = \text{y}}}{pt($y$) \subseteq pt($x$)}
\rulename{A-Cpy}\\[6ex]
\ruledef{\color{blue}{\text{x} = \text{y.f}} \color{black}\quad\hspace*{-1.2ex} o^t_i \in pt\n{(y)}}{pt(o^t_i.\n{f}) \subseteq pt\n{(x)}}
\rulename{A-Ld}&
\ruledef{\color{blue}{\text{x.f} = \text{y}} \color{black}\hspace*{-1.2ex}\quad o^t_i \in pt\n{(x)}}{ pt\n{(y)} \subseteq pt(o^t_i.\n{f})}
\rulename{A-St}
\\[6ex]
\end{tabular}
\begin{tabular}{l}
\ruledef{\color{blue}{\text{x} = \text{y}.m(\text{arg}_1, ..., \;\text{arg}_n)}\color{black} \quad
o^\phd_i \in pt\n{(y)} \quad m' = dispatch(o^\phd_i, m)}
{
\{o^\phd_i\} \subseteq pt(m'_{this}\n{)}\quad
pt(m'_{ret}\n{)} \subseteq pt\n{(x)}\quad
\forall\ 1 \le k \le n : 
pt(\n{arg}_k) \subseteq pt(m'_{pk}\n{)}
}\rulename{A-Call}\\[4ex]
\end{tabular}
}
\caption{Rules for \emph{Pointer Analysis.}
\label{fig:pta}
}
\end{figure}

In \rulenameT{A-Call} (for non-reflective calls), like the one presented 
in~\cite{Manu13}, the function $dispatch(o_i^\phd,m)$ 
is used to resolve the virtual dispatch of method $m$ on
the receiver object $o_i^\phd$ to be $m'$.
There are two cases. 
If $o_i^t\in \ptr(y)$, we proceed normally
as before. 
For $o_i^u\in \ptr(y)$, it suffices to restrict $m$ to
$\{ \texttt{toString()},\texttt{getClass()}\}$, 
as illustrated in Figure~\ref{fig:lhm} and
explained in Section~\ref{sec:lhm}.
We assume that $m'$ has a 
formal parameter $m_{this}'$ for the receiver object and
$m_{p1}',\dots,m_{pn}'$ for the remaining parameters, 
and a pseudo-variable $m_{ret}'$ is used to hold the return 
value of $m'$. 

\subsubsection{Propagation}
\label{prop}

Figure~\ref{fig:prop} gives the rules for handling
\texttt{forName()},
\texttt{getMethod()} and
\texttt{getField()} calls. Different kinds of 
\texttt{Class}, \texttt{Method} and \texttt{Field} objects
are created depending on 
whether their string arguments are string constants or not. 
For these rules, $\mathbb{SC}$ denotes
a set of string constants and the
function $toClass$ creates a \texttt{Class} object 
$c^t$, where $t$ is the class specified by the string value 
returned by $val$($o_i$) (with $val : 
\mathbb{H} \rightarrow \texttt{java.lang.String}$).

\begin{figure}[h]
\small
{
\ruledef{\color{blue}\hspace{-4ex}{Class \;\; \text{c} = Class.forName(\text{cName})}\quad \color{black}
o^{\tt String}_i \in pt\n{(cName)}}{
pt(\textnormal{c}) \supseteq \left\{
\begin{array}{@{\hspace{1ex}}l@{\hspace{2ex}}l}
\{\texttt{c}^t\} &\text{if} \; o^{\tt String}_i \in\ \mathbb{SC} \\
\{\texttt{c}^u\} & \text{otherwise}
\end{array}
\right.
\begin{tabular}{c}
\hspace{2ex}
$\texttt{c}^t = toClass(val(o^{\tt String}_i))$
\end{tabular}
}
\rulename{P-ForName}
\\[3ex]

\ruledef{\color{blue}{Method \;\; \text{m} = \text{c}'.getMethod(\text{mName}, ...)}\quad \color{black}\hspace{-1ex}
o^{\tt String}_i \in pt\n{(mName)} \;\texttt{c}^{-} \in pt\n{(c$'$)}}{
pt(\textnormal{m}) \supseteq \left\{
\def\arraystretch{1.3}
\begin{array}{@{\hspace{.3ex}}l@{\hspace{1ex}}l}
	\{\texttt{m}^t_s\} &\text{if} \, \texttt{c}^{-} = \texttt{c}^t \, \wedge o^{\tt String}_i \in \mathbb{SC} \\
	\{\texttt{m}^t_u\} &\text{if} \, \texttt{c}^{-} = \texttt{c}^t \, \wedge o^{\tt String}_i \notin \mathbb{SC} \\
	\{\texttt{m}^u_s\} &\text{if} \, \texttt{c}^{-} = \texttt{c}^u \wedge o^{\tt String}_i \in \mathbb{SC} \\
	\{\texttt{m}^u_u\} &\text{if} \, \texttt{c}^{-} = \texttt{c}^u \wedge o^{\tt String}_i \notin \mathbb{SC} 
\end{array}
\right.
\hspace{3ex}
\begin{tabular}{c}
$s.t_r = u$\\
$s.n_m = val(o^{\tt String}_i)$\\
$s.p = u$
\end{tabular}
}
\rulename{P-GetMtd}
\\[3ex]

\ruledef{\color{blue}{Field \;\; \text{f} = \text{c}'.getField(\text{fName})}\quad \color{black}
o^{\tt String}_i \in pt\n{(fName)} \quad \texttt{c}^{-} \in pt\n{(c$'$)}}{
pt(\textnormal{f}) \supseteq \left\{
\def\arraystretch{1.3}
\begin{array}{@{\hspace{.3ex}}l@{\hspace{1ex}}l}
	\{\texttt{f}^t_s\} &\text{if} \, \texttt{c}^{-} = \texttt{c}^t \, \wedge o^{\tt String}_i \in \mathbb{SC} \\
	\{\texttt{f}^t_u\} &\text{if} \, \texttt{c}^{-} = \texttt{c}^t \, \wedge o^{\tt String}_i \notin \mathbb{SC} \\
	\{\texttt{f}^u_s\} &\text{if} \, \texttt{c}^{-} = \texttt{c}^u \wedge o^{\tt String}_i \in \mathbb{SC} \\
	\{\texttt{f}^u_u\} &\text{if} \, \texttt{c}^{-} = \texttt{c}^u \wedge o^{\tt String}_i \notin \mathbb{SC} \\
\end{array}
\right.
\hspace{2ex}
\begin{tabular}{c}
$s.t_f = u$\\
$s.n_f = val(o^{\tt String}_i)$
\end{tabular}
}
\rulename{P-GetFld}
}\\[1ex]
\captionof{figure}{Rules for \emph{Propagation}.}
\label{fig:prop}
\end{figure}

By design, $\texttt{c}^t$, $\texttt{f}^t_s$ and
$\texttt{m}^t_s$ will flow to \emph{Target Search} 
but all the others, 
i.e., $\texttt{c}^u$, $\texttt{f}^u_\phd$, $\texttt{f}^\phd_u$,
$\texttt{m}^u_\phd$ and $\texttt{m}^\phd_u$
will flow to \emph{Inference}, where the missing information
is inferred. During \emph{Propagation}, 
only the name of a method/field
signature $s$ ($s.n_m$ or $s.n_f$) can be discovered but 
its other parts are unknown: $s.t_r= s.p=s.t_f= u$. 


\begin{figure}[htbp!]
\small
{
\begin{tabular}{l}
\color{blue}{x = (\textsf{A}) m.$invoke$(y, args)}\\[2ex]

\ruledef{\texttt{m}^u_{-} \in pt\n{(m)}}{pt\n{(m)}
 \supseteq
\{ \; \texttt{m}^t_{-} \; | \;  o^t_i \in pt\n{(y)}
\}}\rulename{I-InvTp}\\[5ex]

\ruledef{\texttt{m}^{-}_u \in pt\n{(m)}}{pt\n{(m)}
 \supseteq
\{ \; \texttt{m}^{-}_s \; | \; s.p\in Ptp(\n{args}), \; \ifamily{s.t_r}{\textsf{A}}, \; s.n_m = u\}}
\rulename{I-InvSig}\\[5ex]

\ruledef{\texttt{m}^u_s \in pt\n{(m)}\quad o^u_i \in pt \n{(y)}\quad \ifamily{s.t_r}{\textsf{A}}\quad
s.n_m \neq u\quad s.p\in Ptp(\n{args})
}{pt\n{(m)}
 \supseteq
\{  \,\texttt{m}^t_s | \,t \in \calM(s.t_r, s.n_m,s.p)
\}}\rulename{I-InvS2T}\\[5ex]

\color{blue}{x = (\textsf{A}) f.$get$(y)}\\[2ex]

\ruledef{\texttt{f}^u_{-} \in pt\n{(f)}}{pt\n{(f)}
 \supseteq
\{ \; \texttt{f}^t_{-} \; | \;  o^t_i \in pt\n{(y)}
\}}\rulename{I-GetTp}\\[5ex]

\ruledef{\texttt{f}^{-}_u \in pt\n{(f)}}{pt\n{(f)}
 \supseteq
\{ \; \texttt{f}^{-}_s \; | \; \ifamily{s.t_f}{\textsf{A}},\;  s.n_f = {u}\}}
\rulename{I-GetSig}\\[5ex]

\ruledef{\texttt{f}^u_s \in pt\n{(f)} \quad o^u_i \in pt \n{(y)}\quad   s.n_f \neq u \quad \ifamily{s.t_f}{\textsf{A}}}{pt\n{(f)}
 \supseteq
 \{  \texttt{f}^t_s  |  \, t \in \calF(s.n_f,s.t_f)
\}}\rulename{I-GetS2T}\\[5ex]

\color{blue}{f.$set$(y, x)}\\[2ex]

\ruledef{\texttt{f}^u_{-} \in pt\n{(f)}}{pt\n{(f)}
 \supseteq
\{ \; \texttt{f}^t_{-} \; | \;  o^t_i \in pt\n{(y)}
\}}\rulename{I-SetTp}\\[5ex]

\ruledef{\texttt{f}^{-}_u \in pt\n{(f)}}{pt\n{(f)}
 \supseteq
\{ \; \texttt{f}^{-}_s \; | \; o^t_j \in pt\n{(x)}, \; \ibinding{\textsf{t}}{s.t_f}, s.n_f = {u}\}}
\rulename{I-SetSig}\\[5ex]

\ruledef{\texttt{f}^u_s \in pt\n{(f)} \quad o^u_i \in pt \n{(y)}\quad 
s.n_f\neq u\;\; o^{t'}_j \in pt\n{(x)} \;\; \ibinding{t'\hspace{-1pt}}{s.t_f} }{pt\n{(f)}
 \supseteq
\{  \texttt{f}^t_s | \,t \in \calF(s.n_f,s.t_f)
\}}\rulename{I-SetS2T}\\[5ex]

\end{tabular}
}
\caption{Rules for \emph{Collective Inference}.}
\label{fig:infer}
\end{figure}


\subsubsection{Collective Inference}
\label{infer}

Figure~\ref{fig:infer} gives nine rules 
to infer reflective targets at 
\texttt{x = (A) m.invoke(y,args)},
\texttt{x = (A) f.get(y)},
\texttt{f.set(y,x)}, where \texttt{A} indicates a 
post-dominating cast on their results. If 
\texttt{A = Object}, 
then no such cast exists. 
These rules fall into three categories.
In \rulenameT{I-InvTp}, \rulenameT{I-GetTp} and \rulenameT{I-SetTp}, we use the types of the objects pointed to
by \texttt{y} to infer the class type of a method/field.
In \rulenameT{I-InvSig}, \rulenameT{I-GetSig} and
\rulenameT{I-SetSig}, we use the information available
at a call site (excluding \texttt{y}) to infer the descriptor 
of  a method/field
signature. In \rulenameT{I-InvS2T}, \rulenameT{I-GetS2T} and \rulenameT{I-SetS2T}, we use a method/field signature 
to infer the class type of a method/field.

Some notations used are in order.
As is standard, $\ibinding{t}{t'}$ holds when $t$ is $t'$ or
a subtype of $t'$. 
In \rulenameT{I-InvSig}, \rulenameT{I-GetSig}, \rulenameT{I-InvS2T} and \rulenameT{I-GetS2T},
$\ll:$ is used to take advantage
of the post-dominating cast \texttt{(A)} during inference
when \texttt{A} is not \texttt{Object}. By definition,
$\ifamily{u}{\tt Object}$ holds. 
If $t'$ is not \texttt{Object}, then
$\ifamily{t}{t'}$ holds if and only if 
$\ibinding{t}{t'}$ or $\ibinding{t'}{t}$ holds.
In \rulenameT{I-InvSig} and \rulenameT{I-InvS2T}, the information
on \texttt{args} is also exploited, where \texttt{args}
is an array of type \texttt{Object[]}, only when
it can be analyzed exactly element-wise 
by an intra-procedural analysis. In this case, suppose that
\texttt{args} is an array of $n$ elements. Let $A_i$ be the
set of types of the objects pointed to by its $i$-th element,
\texttt{args[$i$]}. Let
$P_i = \{ t'  \mid t \in A_i, 
\ibinding{t}{t'}\}$. Then
$Ptp(\texttt{args})=P_0\times \cdots \times P_{n-1}$.
Otherwise, $Ptp(\texttt{args})=\varnothing$, implying that 
\texttt{args} is ignored as it cannot be exploited
effectively during inference.


To maintain precision
in \rulenameT{I-InvS2T}, \rulenameT{I-GetS2T} and \rulenameT{I-SetS2T},
we use a method (field) signature to infer its 
classes when
both its name and descriptor are known. 
In \rulenameT{I-InvS2T}, the
function $\calM(s_{t_r}, s.{n_m}, s.p)$
returns the set of classes where the method with the
specified signature $s$ is defined if
$s.n_m \neq u$ and $s.p \neq u$, and $\varnothing$
otherwise. The return type of the matching method is ignored if $s.{t_r}=u$. 
In \rulenameT{I-GetS2T} and \rulenameT{I-SetS2T}, 
$\calF(s.n_f,s.t_f)$ 
returns the set of classes where 
the field with the given signature $s$ 
is defined if
$s.n_f\neq u$ and $s.t_f\neq u$,
and $\varnothing$ otherwise.

Let us illustrate our rules by considering
two examples in Figures~\ref{fig:mot} and \ref{fig:handle}.

\begin{example}
\label{ex:infer}
Let us modify the 
reflective allocation site in line 3 (Figure~\ref{fig:mot}) to
\texttt{c1.newInstance()}, where \texttt{c1} represents
a known class, named \texttt{A}, so that
$\texttt{c1}^\texttt{A}\in \ptr(\texttt{c1})$.
By applying \rulenameT{L-KwTp} (introduced later in 
Figure~\ref{fig:dam}) to the modified allocation 
site, \solar will create a new object
$o_3^{\tt A}$, which will flow to line 10, so
that $o_3^{\tt A}\in\ptr(\texttt{v})$. 
Suppose both \texttt{cName2} and 
\texttt{mName2} point to some unknown strings. When
\rulenameT{P-GetMtd} is applied 
to \texttt{c.getMethod(mName,\ldots)} in line 7,
a \texttt{Method} object, say,
$\texttt{m}_u^u$ is created and eventually assigned
to \texttt{m} in line 14. 
By applying 
\rulenameT{I-InvTp}
to \texttt{m.invoke(v, args)} in line 
15, where 
$o_3^{\tt A}\in \ptr(v)$,  \solar deduces that the target
method is a member of class \texttt{A}.
Thus, a
new object $\texttt{m}^{\tt A}_u$ 
is created and assigned to $\ptr(\texttt{m})$. 
Given \texttt{args = new Object[] \{b, c\}},
$Ptp(args)$ is constructed as described earlier.
By applying \rulenameT{I-InvSig} to this 
\texttt{invoke()} call,
\solar will add all new \texttt{Method} objects 
$\texttt{m}^{\tt A}_s$ to 
$\ptr(\texttt{m})$ such that
$s.p\in Ptp(args)$, which represent
the potential target methods called reflectively at
this site.
\qed
\end{example}

\begin{example}
	\label{ex:formal-handle}
In Figure~\ref{fig:handle}, \texttt{hd} is statically
unknown but the string argument of \texttt{getMethod()}
is \texttt{"handle"}, a string constant. By applying \rulenameT{P-ForName}, \rulenameT{P-GetMtd} and \rulenameT{L-UkwTp}
(Figure~\ref{fig:dam}) to the
\texttt{forName()},
\texttt{getMethod()} and
\texttt{newInstance()} calls, respectively, we obtain
$\texttt{c}^u\in \ptr(\texttt{c})$,
$\texttt{m}^u_s\in \ptr(\texttt{m})$ and
$o_i^u \in \ptr(\texttt{h})$, where $s$ 
indicates a signature with a 
known method name (i.e.,
``\texttt{handle}''). Since
the second argument of the \texttt{invoke()} call
can also be exactly analyzed, \solar will be able
to infer the classes $t$ where method \texttt{"handle"} is defined 
by applying \rulenameT{I-InvS2T}.
Finally, \solar will add all inferred
\texttt{Method} objects \texttt{m}$^t_s$ to $pt(\texttt{m})$ at the \texttt{invoke()} call site. 
Since neither the superscript nor the subscript of \texttt{m}$^t_{s}$ is \emph{u}, the inference is finished and the inferred \texttt{m}$^t_{s}$
will be used to find out the reflective targets (represented by it)
in \emph{Target Search} (Section~\ref{search}).
\qed
\end{example}

\begin{figure}[t]
\centering
\includegraphics[width=0.75\textwidth]{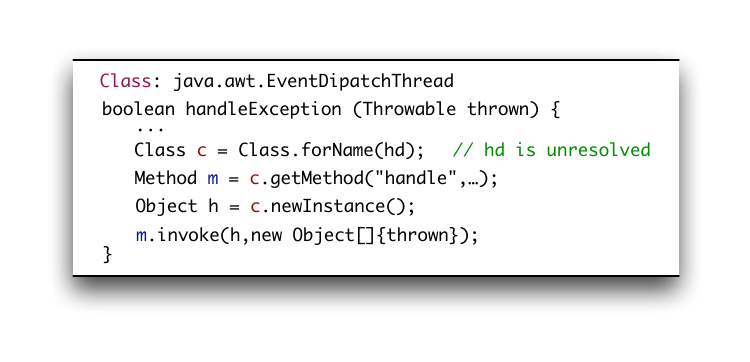}
\vspace{-5ex}
\caption{A simplified real code of Example~\ref{ex:formal-handle} for illustrating inference rule \rulenameT{I-InvS2T}
\label{fig:handle}}
\end{figure}

\subsubsection{Target Search}
\label{search}

For a \texttt{Method} object $\texttt{m}_s^t$ in a known class $t$ (with $s$ being possibly $u$),
we define
$MTD: \mathbb{MO} \rightarrow \mathcal P 
(\mathbb{M})$ to find all the methods matched:
\begin{eqnarray}
\label{formu:mtd}
MTD(\texttt{m}^t_{s}) &= &
\bigcup \limits_{t <: t'} mtdLookUp(t', s.t_r, s.n_m, s.p) 
\end{eqnarray}

\noindent
where $mtdLookUp $ is the standard lookup
function for finding the methods 
according to a
declaring class $t'$ and 
a signature $s$ except that (1)
the return type $s.t_r$ is also considered in the search 
(for better precision) and (2)
any $u$ that appears in $s$ is treated as a wild card
during the search.

Similarly,
we define
$FLD:\mathbb{FO}\! \rightarrow\! \mathcal P 
(\mathbb{F})$
for a \texttt{Field} object $\texttt{f}_s^t$:
\begin{eqnarray}
\label{form:fld}
FLD(\texttt{f}^t_s) &=&
\bigcup \limits_{t <: t'} fldLookUp(t', s.t_f, s.n_f) 
\end {eqnarray}

\noindent
to find all the fields matched,
where $fldLookUp$ plays a similar role as $mtdLookUp$.
Note that both $MTD(\texttt{m}^t_s)$ and $FLD(\texttt{f}^t_s)$ also need to 
consider the super types of $t$ (i.e., the union of the results for all $t'$ where $t <: t'$, as shown in the functions) to be conservative due to the existence of member inheritance
in Java.

\subsubsection{Transformation}
\label{trans}

Figure~\ref{fig:sideeffect} gives the rules used
for transforming a reflective call into a regular statement,
which will be analyzed by the pointer analysis.

\begin{figure}[htbp]
	\small
{
\begin{tabular}{l}
\ruledef{\color{blue}{\n{x} = \n{m}.invoke\n{(y, args)}}\quad \color{black}
\texttt{m}^t_{-} \in pt\n{(m)} \quad m' \in MTD(\texttt{m}^t_{-}) \quad o^\phd_i \in pt\n{(args)} \\
 o^{t'}_j \in pt(o^\phd_i.arr)\quad t''\ \mbox{is declaring type
 of}\ m_{pk}'\quad \ibinding{t'}{t''} \\
}{
\forall\ 1\leq k\leq n:
	\{o^{t'}_j\} \subseteq pt(\n{arg}_k) \quad
\n{x} = \n{y}.m'(\n{arg}_1,...,\n{arg}_n)}
\rulename{T-Inv}
\\[6ex]

\ruledef{ \color{blue}{\n{x} = \n{f.}get\n{(y)}} \color{black}\quad \texttt{f}^t_{-}\in pt\n{(f)}\quad f \in FLD(\texttt{f}^t_{-})}{\n{x} = \n{y.}f}
\rulename{T-Get}\\[5ex]

\ruledef{ \color{blue}{\n{f.}set\n{(y, x)}}\color{black} \quad \texttt{f}^t_{-}\in pt\n{(f)}\quad f \in FLD(\texttt{f}^t_{-})}{\n{y.}f = \n{x}}
\rulename{T-Set}\\[2ex]
\end{tabular}
}
\caption{Rules for \emph{Transformation}.
\label{fig:sideeffect}
}
\end{figure}


Let us examine \rulenameT{T-Inv} in more detail. The second
argument \texttt{args} points to a one-dimensional array of type \texttt{Object[]},
with its elements collapsed to a single field $arr$ during the
pointer analysis, unless \texttt{args} can be analyzed
exactly intra-procedurally in our current implementation. 
Let \texttt{arg$_1$},\ldots, \texttt{arg$_n$} be the $n$
freshly created arguments to be passed to each potential
target method $m'$ found by \emph{Target Search}.
Let $m_{p1}',\dots,m_{pn}'$ be the $n$ parameters
(excluding $this$) of $m'$, such that the declaring type of $m_{pk}'$
is $t''$. We include $o_j^{t'}$ to $\ptr({\tt arg}_k)$ only
when $\ibinding{t'}{t''}$ holds in order to filter out the 
objects that cannot be assigned to $m_{pk}'$. Finally,
the reflective target method found can be analyzed
by \rulenameT{A-Call} in Figure~\ref{fig:pta}.

\subsubsection{Lazy Heap Modeling}
\label{sec:lhm-rules}
In Figure~\ref{fig:dam}, we give the rules for lazily resolving
a \texttt{newInstance()} call, as explained in
Section~\ref{sec:lhm}.  

\begin{figure}[h]
\small
{
\begin{tabular}{l}
\ruledef{\color{blue}{i: \text{o} = \text{c}'.newInstance()} \color{black}\quad
\texttt{c}^t \in pt\n{(c$'$)}}{\{o^t_i\} \subseteq pt\n{(o)}}
\rulename{L-KwTp}\\[5ex]

\ruledef{\color{blue}{i: \text{o} = \text{c}'.newInstance()} \color{black}\quad
\texttt{c}^u \in pt\n{(c$'$)}}{ \{o^u_i\}  \subseteq pt\n{(o)}}
\rulename{L-UkwTp}\\[5ex]

\ruledef{\color{blue}{\textsf{A}\: a = (\textsf{A}) \, \text{x}} \color{black} \quad
o^u_i \in pt\n{(x)} \;\; \ibinding{t}{\textsf{A}}}{\{o^t_i\}\subseteq pt(a)}
\rulename{L-Cast}\\[5ex]

\ruledef{\color{blue}{\text{x} = \text{m}.invoke(\text{y}, ...)} \color{black} \;\; 
o^u_i \in pt\n{(y)}\;\; \texttt{m}^t_{-} \in pt\n{(m)} \;\; \ifamily{t'}{t}}{\{o^{t'}_i\} \subseteq pt\n{(y)}}\rulename{L-Inv}\\[5ex]

\ruledef{\color{blue}{\text{x} = \text{f}.get(\text{y}) \;/ \;\text{f}.set(\text{y, x})} \color{black} \quad 
o^u_i \in pt\n{(y)}\quad \texttt{f}^t_{-} \in pt\n{(f)} \quad \ifamily{t'}{t}}{\{o^{t'}_i\}\subseteq pt\n{(y)}}
\rulename{L-GSet}\\[5ex]
\end{tabular}
}
\caption{Rules for \emph{Lazy Heap Modeling}.
\label{fig:dam}}
\end{figure}

In \rulenameT{L-KwTp}, 
for each \texttt{Class} object
$\texttt{c}^t$ 
pointed to by \texttt{c$'$}, 
an object, $o_i^t$, is created as an instance of
this known type at allocation site~$i$ 
straightaway.
In \rulenameT{L-UkwTp}, as illustrated in
Figure~\ref{fig:lhm}, 
$o_i^u$ is created to enable LHM
if \texttt{c$'$} points to a $\texttt{c}^u$ instead.
Then its lazy object
creation happens at its Case (II)
by applying \rulenameT{L-Cast}
(with $o_i^u$ blocked from flowing from \texttt{x}
to \texttt{a}) and
its Case (III) by applying \rulenameT{L-Inv} and
\rulenameT{L-GSet}.
Note that in \rulenameT{L-Cast}, 
\texttt{A} is assumed not to be \texttt{Object}.


\subsection{Soundness Criteria}
\label{sec:cond}
\label{condition}

\RefJava consists of four side-effect methods as described in Section~\ref{sec:refjava}.  \solar is sound if their calls are
resolved soundly under Assumptions~\ref{ass:world} --
\ref{ass:reach}. Due to Assumption~\ref{ass:reach}
illustrated in Figure~\ref{fig:lhm}, there is no need to consider
\texttt{newInstance()} since it is soundly resolved if
\texttt{invoke()}, \texttt{get()} and \texttt{set()} are.
For convenience, we define:
\begin{eqnarray}
AllKwn({v}) & = & \nexists\, o^u_i \in pt(v)
\end{eqnarray}
which means that the dynamic type of
every object pointed to by $v$ is known. 

Recall the nine rules given for resolving
\texttt{(A) m.invoke(y, args)}, 
\texttt{(A) f.get(y)} and \texttt{f.set(y, x)}
in Figure~\ref{fig:infer}. 
For the \texttt{Method} (\texttt{Field}) objects
$\texttt{m}_s^t$ ($\texttt{f}_s^t$) with known
classes $t$, these targets can 
be soundly resolved by \textit{Target Search},
except that the signatures $s$ can be further
refined by applying 
\rulenameT{I-InvSig}, \rulenameT{I-GetSig} and
\rulenameT{I-SetSig}.

For the \texttt{Method} (\texttt{Field}) objects
$\texttt{m}_s^u$ ($\texttt{f}_s^u$) with unknown
class types $u$, the targets accessed are inferred
by applying the remaining six rules
in Figure~\ref{fig:infer}.
Let us consider a call to \texttt{(A) m.invoke(y, args)}. 
\solar attempts to infer the missing classes of its
\texttt{Method} objects
in two ways, by applying
\rulenameT{I-InvTp} and \rulenameT{I-InvS2T}. 
Such a call is soundly resolved 
if the following condition holds:
\begin{equation}
\label{eq:sound-inv}
SC(\texttt{m.invoke(y,args)}) = 
AllKwn(\texttt{y})
\vee\ \forall\ \texttt{m}_s^u\in  \ptr(\texttt{m}) : s.n_m\neq u
\wedge Ptp(\texttt{args}) \neq \varnothing
\end{equation}

If the first disjunct holds,
applying \rulenameT{I-InvTp} to 
\texttt{invoke()} can over-approximate its
target methods from
the types of all objects pointed to by \texttt{y}.
Thus, every \texttt{Method} object 
$\texttt{m}_\phd^u\in\ptr(m)$ is refined into
a new one $\texttt{m}_\phd^t$ for every $o_i^t\in\ptr(\texttt{y})$. 

If the second disjunct holds, then
\rulenameT{I-InvS2T} comes into play. Its
targets are over-approximated based 
on the known method names $s.n_m$ and 
the types of the objects pointed to by 
\texttt{args}. Thus, every \texttt{Method} object 
$\texttt{m}_s^u\in\ptr(m)$ is refined into
a new one $\texttt{m}_s^t$, where 
$ \ifamily{s.t_r}{A}$ and $s.p\in Ptp(\texttt{args})
\neq \varnothing$. Note that
$s.t_r$ is leveraged only when it is not $u$.
The post-dominating cast (\texttt{A}) is
considered not to exist if \texttt{A = Object}.
In this case, $\ifamily{u}{\tt Object}$ holds (only
for $u$). 

Finally, the soundness criteria for \texttt{get()} and \texttt{set()} are derived similarly: 
\begin{equation}
\label{eq:sound-get}
SC(\texttt{(A)\,f.get(y)}) = 
AllKwn(\texttt{y})\;
\vee\ \forall\ \texttt{f}_s^u\in  \ptr(\texttt{f}) : s.n_f\neq u \wedge 
\texttt{A} \neq \texttt{Object}
\end{equation}
\begin{equation}
\label{eq:sound-set}
SC(\texttt{f.set(y,x)}) =  
AllKwn(\texttt{y}) \;
\vee\ \forall\ \texttt{f}_s^u\in  \ptr(\texttt{f}) : 
s.n_f \neq u \wedge AllKwn(\texttt{x})
\end{equation}

In (\ref{eq:sound-get}), applying
\rulenameT{I-GetTp} (\rulenameT{I-GetS2T}) 
resolves a \texttt{get()} call soundly
if its first (second) disjunct holds.
In (\ref{eq:sound-set}), applying
\rulenameT{I-SetTp} (\rulenameT{I-SetS2T}) 
resolves a \texttt{set()} call soundly
if its first (second) disjunct holds. 
By observing \rulenameT{T-Set}, we see why
$AllKwn(\texttt{x})$ is needed to 
reason about the soundness of 
\rulenameT{I-SetS2T}.

\subsection{Soundness Proof}
\label{properties}
\label{sec:proof}

We prove the soundness of \solar for \refjava
subject to our soundness criteria
(\ref{eq:sound-inv}) --
(\ref{eq:sound-set}) under 
Assumptions~\ref{ass:world} -- \ref{ass:reach}. 
We do so by 
taking advantage of the well-established soundness of
Andersen's pointer analysis 
(Figure~\ref{fig:pta}) stated below.

\begin{lemma}
\label{thm:andersen}
\solar is sound for \refjava with its reflection
API ignored.
\end{lemma}


\newcommand{\str}{\texttt{str}\;}
\newcommand{\alloc}{\texttt{newIn-i}\;}
\newcommand{\strnew}{\texttt{str-i}\;}

If we know the class types of all 
targets accessed at a reflective call but possibly
nothing about their signatures, \solar can 
over-approximate its
target method/fields in \textit{Target
Search}. Hence, the following lemma holds.

\begin{lemma}
\label{thm:const-cname}
\label{thm:consts}
\label{thm:cname}
\solar is sound for \refjavacname, 
the set of all
\refjava programs in which \texttt{cName} 
is a string constant in every
\texttt{Class.forName(cName)} call. 
\end{lemma}
\vspace*{-2ex}
\begin{proof}[Sketch]
By \rulenameT{P-ForName}, 
the \texttt{Class} objects at all
\texttt{Class.forName(cName)} calls 
are created from known class types. By 
Lemma~\ref{thm:andersen}, this has four
implications. (1) LHM is not needed.
For the rules in Figure~\ref{fig:dam}, only
\rulenameT{L-KwTp} is relevant, enabling every
\texttt{c$'$.newInstance()} call to be resolved soundly as
a set of regular \texttt{new t()} calls, for all 
\texttt{Class} objects
$\texttt{c}^t\in\ptr(\texttt{c}')$. 
(2) In
\rulenameT{P-GetMtd}, the
\texttt{Method} objects \texttt{m}$^t_\phd$ of all 
class types $t$ accessed at a \texttt{getMethod()} call
are created, where 
$\texttt{m}^t_\phd$ symbolizes 
over-approximately all target methods in
$MTD(\texttt{m}^t_\phd)$. 
The same sound approximation for \texttt{getField()}
is made by \rulenameT{P-GetFld}. 
(3) For the nine rules given in Figure~\ref{fig:infer},
only \rulenameT{I-InvSig}, \rulenameT{I-GetSig} and
\rulenameT{I-SetSig} are applicable,
since
$\nexists \texttt{m}_\phd^u \in \texttt{m}$,
$\nexists \texttt{f}_\phd^u \in \texttt{f}$ and
$\nexists o_i^u \in \texttt{y}$, in order to
further refine their underlying
method or field signatures. (4) In
Figure~\ref{fig:sideeffect}, 
the set of reflective targets at each call site
is over-approximated. By applying Lemma~\ref{thm:andersen}
again and noting Assumptions~\ref{ass:world} -- \ref{ass:cast}, \solar is sound for \refjavacname.
\end{proof}



\solar is sound subject to 
(\ref{eq:sound-inv}) --  (\ref{eq:sound-set}) under 
Assumptions~\ref{ass:world} -- \ref{ass:reach}. 

\begin{theorem}
\label{t1}
\label{thm:soundness}
\solar is sound for \refjava if $SC(c)$ holds
at every reflective call $c$ of the form 
\texttt{(A) m.invoke(y, args)}, 
\texttt{(A) f.get(y)} or
\texttt{f.set(y, x)}. 
\end{theorem}
\vspace*{-1.ex}
\begin{proof}[Sketch] For reasons of symmetry,
we focus only on a call, $c$, to 
\texttt{(A) m.invoke(y, args)}. $SC(c)$ is given
in (\ref{eq:sound-inv}).
If its first disjunct is true,
then all abstract objects flowing into \texttt{y} are
created either from known classes soundly
by \rulenameT{L-Kwtp} or lazily from unknown classes
as illustrated in
Figure~\ref{fig:lhm} by applying initially
\rulenameT{L-UkwUp} and later
\rulenameT{L-Cast} (for Case (II))
and \rulenameT{L-Inv} (for Case (III)),
but soundly under Assumption~\ref{ass:reach}. 
If its second disjunct is true,
then \solar can always infer the missing class types $t$
in a \texttt{Method} object $\texttt{m}_s^u$ pointed to by
$\ptr(\texttt{m})$ to over-approximate
the set of its target methods as $MTD(\texttt{m}_s^t)$.
This takes us
back to a situation equivalent to the one 
established in the proof of 
Lemma~\ref{thm:cname}. Thus, \solar is sound
for \refjava if $SC(c)$ holds
at every reflective call $c$ of the form 
\texttt{(A) m.invoke(y, args)}.
\end{proof}

\subsection{\probe}
\label{sec:probe}

We instantiate
\probe, as shown in Figure~\ref{fig:solar}, from
\solar as follows. To make \probe resolve reflection precisely, we refrain from performing \solar's LHM
by retaining \rulenameT{L-UkwTp} but
ignoring \rulenameT{L-Cast}, 
\rulenameT{L-GSet} and \rulenameT{L-Inv} and abandon some
of \solar's sophisticated 
inference rules by disabling
\rulenameT{I-InvS2T}, \rulenameT{I-GetS2T} and \rulenameT{I-SetS2T}. 

In \emph{Target Search}, \probe will
restrict itself to only \texttt{Method} objects $\texttt{m}^t_s$ and \texttt{Field} objects 
$\texttt{f}^t_s$, where the signature
$s$ is at least partially known.

\subsection{Static Class Members}
\label{other}

To handle static class members, our rules
can be simply modified. In
Figure~\ref{fig:infer}, \texttt{y = \texttt{null}}. \rulenameT{I-InvTp}, \rulenameT{I-GetTp} and \rulenameT{I-SetTp} are not needed (by assuming
$\ptr(\texttt{null})=\varnothing$). In the
soundness criteria stated in (\ref{eq:sound-inv}) --
(\ref{eq:sound-set}), the first disjunct is
removed in each case.
\rulenameT{I-InvS2T}, \rulenameT{I-GetS2T} and \rulenameT{I-SetS2T} are modified with
$o^u_i \in pt(\texttt{y})$ being replaced by $\texttt{y}
= \texttt{null}$. 
The rules in Figure~\ref{fig:sideeffect} 
are modified to deal with static members. 
In Figure~\ref{fig:dam}, \rulenameT{L-GSet} and \rulenameT{L-Inv}
are no longer relevant. The other rules remain applicable.
The static initializers for the classes in the 
closed world are analyzed. This can happen
at, say loads/stores for static fields as is the standard
but also when some classes are discovered in \rulenameT{P-ForName}, \rulenameT{L-Cast}, \rulenameT{L-GSet} and \rulenameT{L-Inv}.


\section{Implementation}
\label{sec:impl}


\solar, as shown in Figure~\ref{fig:form}, works together 
with a pointer analysis.
We have implemented \solar on top of \doop~\cite{Yannis09}, a 
state-of-the-art pointer analysis framework for Java.
\solar is implemented in the Datalog language. 
Presently, \solar consists
of 303 Datalog rules written in about 2800 lines of code.

We release \solar as an open-source tool at
\url{http://www.cse.unsw.edu.au/~corg/solar}.
Now \solar has been augmented to output its reflection analysis results with the format that is supported by \Soot (a popular framework for analyzing Java and Android applications), which enables \soot's clients to use \solar's results directly.

Below we extend our formalism for \RefJava
to handle the other methods in the Java reflection 
API, divided into entry, member-introspecting and
side-effect methods.
For details, we refer to the open-source software of \solar.

\paragraph{Entry Methods}
We divide the entry methods into six groups, depending how a
\texttt{Class} object is obtained, by
(1) using a special syntax (\texttt{A.class}),
(2) calling \texttt{Object.getClass()}, 
(3) using a full-qualified string name (in
\texttt{Class.forName()} and \texttt{ClassLoader.loadClass()}),
(4) calling proxy API \texttt{Proxy.getProxyClass()}, 
(5) calling, e.g., \texttt{Class.getComponentType()}, on a 
metaobject, and 
(6) calling, e.g.,
\texttt{sun.reflect.Reflection.getCallerClass()} to 
introspect an execution context.
According to Section~\ref{study:entry}, the last two
are infrequently used. Our current implementation handles
the entry methods in (1) -- (4) plus 
\texttt{Class.getComponentType()} in (5) since the latter is
used in array-related side-effect methods.

\paragraph{Member-Introspecting Methods}
\label{sec:mem}
In addition to the two included in \RefJava, there are
ten more member-introspecting methods as given 
in Figure~\ref{study:fig:coreAPI}. 
To show how to handle the other ten methods
based on our existing formalism (Section~\ref{sec:form}), 
we take
\texttt{getDeclaredMethod()} and 
\texttt{getMethods()} as examples as the others can be
dealt with similarly.
For a \texttt{Method} object resolved to be
\texttt{m}$^t_{s}$ at a \texttt{getDeclaredMethod()} call
by \emph{Propagation}, the search for building
$MTD(\texttt{m}^t_s)$ is confined to class $t$ (the only change in the rules). Otherwise,
the search is done as described in Section~\ref{search}.

For \texttt{ms = c$'\!$.getMethods()}, which returns  
an array of \texttt{Method} objects, its analysis 
is similarly done as \texttt{m = c$'\!$.getMethod(mName)}
in Figure~\ref{fig:prop}. We first create a placeholder
array object
$\texttt{ms}_{ph}$ so that $\texttt{ms}_{ph}\in 
\ptr(\texttt{ms})$. 
As \texttt{getMethods()} is parameterless, we
then insert a new \texttt{Method} object
\texttt{m}$^t_u$ (\texttt{m}$^u_u$) into 
$\ptr(\texttt{ms}_{ph}.arr)$ for every 
$\texttt{c}^t$ ($\texttt{c}^u$) 
in $\ptr(\texttt{c}')$.
The underlying
pointer analysis will take care of how eventually
an element of such an array flows into, say \texttt{m},
in \texttt{m.invoke()}. Then the rules 
in Figure~\ref{fig:infer} are applicable.

\paragraph{Side-Effect Methods}

In addition to the four in \RefJava, the following four side-effect methods are 
also analyzed. 
\texttt{Constructor.newInstance()} is handled 
as \texttt{Class.newInstance()} except that its
array argument is handled exactly as in \texttt{invoke()}.
Its
call sites are also modeled lazily.
\texttt{Array.newInstance()}, \texttt{Array.get()} and 
\texttt{Array.set()} 
are handled as in Table~\ref{study:table:sideAPI}. 
Presently, \solar does not handle \texttt{Proxy.newProxyInstance()} in its implementation. However, it can be analyzed according to
its semantics as described in Section~\ref{sec:under:inter:over}.

\section{Evaluation}
\label{eval}
\label{sec:eval}

We evaluate \solar by comparing it (primarily) with
two state-of-the-art
reflection analyses for Java, \doopr~\cite{Yannis15} and 
\elf~\cite{Yue14}.
In particular,
our evaluation addresses the following research 
questions (RQs): 

\begin{itemize}
\item RQ1. 
How well does \solar achieve full
automation without any annotations?

\item RQ2. 
How well does \solar reason about the soundness and identify the unsoundly or imprecisely resolved (or ``problematic'') reflective calls?

\item RQ3.
How well does \solar make the trade-offs among soundness,
precision and scalability in terms of reflection analysis?
In other words, can \solar resolve reflection more soundly than \Doopr and \elf with acceptable precision and efficiency?

\end{itemize} 

Below we describe our experimental setup, revisit our
assumptions, and answer these RQs in the given order.

\subsection{Experimental Setup}
\label{setup}

As all the three reflection
analyses, \Doopr, \elf 
and \solar, are implemented in \Doop~\cite{Yannis09} (a state-of-the-art pointer analysis framework for Java) and run with a pointer analysis, we compare them by running each together with
the same context-sensitive pointer
analysis option provided by \doop: \texttt{selective-2-type-sensitive+heap}~\cite{Yannis13pldi}.

\Doopr denotes the reflection analysis proposed in~\cite{Yannis15} and all its provided reflection resolution options are turned on including 
\texttt{-enable-reflection}, 
\texttt{-enable-reflection-substring-analysis},
\texttt{-enable-reflection-use-based-analysis} and
\texttt{-enable-reflection-invent-unknown-objects}, in the experiment.
\elf (version 0.3) is the reflection analysis proposed in ~\cite{Yue14} with its default setting.

We use the LogicBlox Datalog engine (v3.9.0) 
on a Xeon E5-2650 2GHz
machine with 64GB of RAM. We consider
7 large DaCapo benchmarks (2006-10-MR2) and 4 real-world applications, \texttt{avrora}-1.7.115 (a simulator), \texttt{checkstyle}-4.4 (a checker), \texttt{freecs}-1.3.20111225 (a server) and \texttt{findbugs}-1.2.1 (a bug detector), under a
large reflection-rich Java library,
\texttt{JDK 1.6.0\_45}.
The other 4 small DaCapo benchmarks are excluded
since reflection is much less used. 

\subsection{Assumptions}
\label{sec:eval-ass}
When analyzing the real code in the experiment,
we accommodate
Assumptions~\ref{ass:world} -- \ref{ass:reach} as follows.
For Assumption~\ref{ass:world},
we rely on \doop's pointer analysis to simulate the
behaviors of Java native methods. 
Dynamic class loading is assumed to be resolved
separately. 
To simulate its effect,
we create a closed world for a program, by 
locating the classes referenced with \doop's
fact generator and adding additional ones
found through program runs under \TamiFlex~\cite{Bodden11}.
For the DaCapo benchmarks, their
three standard inputs are used. 
For \texttt{avrora} and \texttt{checkstyle}, 
their default test cases are used.
For \texttt{findbugs}, one Java program is 
developed as its input
since no default ones are available. 
For \texttt{freecs}, a server requiring user
interactions, we only initialize it as the input
in order to ensure repeatability.
For a class in the closed-world of a program, \solar
analyzes its static initializer at the outset
if it is dynamically loaded and proceeds 
as discussed in 
Section~\ref{other} otherwise.

Assumptions~\ref{ass:loader} and~\ref{ass:cast} are taken for
granted.

As for Assumption~\ref{ass:reach}, 
we validate it
for all reflective allocation sites where $o_i^u$ is created
in the application code of the 10 programs that can be analyzed scalably. 
This assumption is
found to hold at 75\% of these sites automatically
by performing a simple intra-procedural
analysis. 
We have inspected the remaining 25\% interprocedurally and found only two violating 
sites (in \texttt{eclipse} and \texttt{checkstyle}), 
where $o_i^u$ is never used.
In the other sites inspected,
$o_i^u$ flows through only local variables with all the
call-chain lengths being at most 2. 
This shows that Assumption~\ref{ass:reach} generally holds in practice.


\subsection{RQ1: Automation and Annotations}
\label{sec:rq1}

Figure~\ref{fig:ann} compares \solar and 
other reflection analyses~\cite{Livshits05,Yue14,Yannis15,Ernst15,Lili16,Zhang17} denoted 
by ``Others'' by the degree of automation
achieved. For an analysis, this 
is measured by the number of annotations required
in order to improve the soundness of the
reflective calls identified to be potentially
unsoundly resolved. 
\begin{figure}[thp]
\centering
\includegraphics[width=0.8\textwidth]{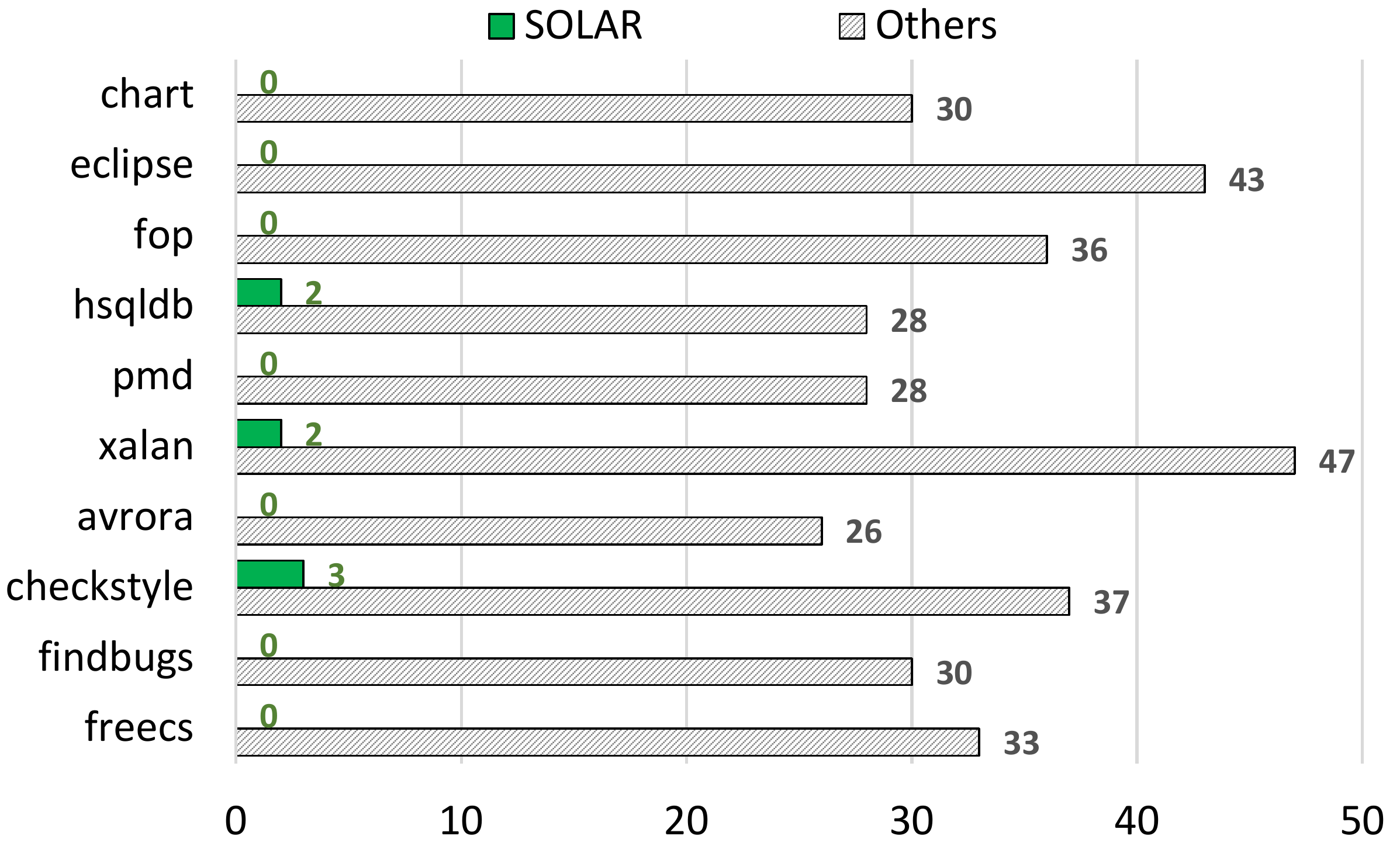}
\caption{Number of annotations required
for improving the soundness of unsoundly resolved 
reflective calls. 
\label{fig:ann}}
\end{figure}

\solar analyzes 7 out of the 11 programs
scalably with full automation. For 
\texttt{hsqldb},
\texttt{xalan} and
\texttt{checkstyle},
\solar is unscalable (under 3 
hours). With \probe, a total of 
14 reflective calls are
flagged as being potentially unsoundly/imprecisely resolved. 
After 7 annotations, 
2 in \texttt{hsqldb},
2 in \texttt{xalan} and 
3 in \texttt{checkstyle}, 
\solar is scalable, as discussed
in Section~\ref{sec:rq2}.
However, \solar is unscalable 
(under 3 hours) for \texttt{jython},
an interpreter (from \texttt{DaCapo})
for Python in which
the Java libraries and application code are invoked 
reflectively from the Python code. Neither are \doopr and \elf. 

``Others'' cannot
identify which reflective calls may be
unsoundly resolved; however, they can improve
their soundness by adopting an annotation
approach suggested in~\cite{Livshits05}.
Generally, this approach requires users 
to annotate the string arguments
of calls to entry and member-introspecting methods
if they are not string constant.
To reduce the annotation effort, e.g., for
a \texttt{clz = Class.forName(cName)} call with 
\texttt{cName}
being an input string, this approach leverages the intra-procedural post-dominating
cast on the result of a
\texttt{clz.newInstance()} call
to infer
the types of \texttt{clz}. 

As shown in Figure~\ref{fig:ann},
``Others'' will require a total of 
338 annotations initially and possibly more in the
subsequent iterations (when more code is 
discovered). As discussed in Section~\ref{sec:mech},
\solar's annotation approach is also iterative. However,
for these  programs, \solar requires 
a total of 7 annotations in only one iteration.

\solar outperforms ``Others'' due to its
powerful inference engine for performing reflection
resolution and its
effective mechanism in accurately identifying unsoundly
resolved reflective calls
as explained in Section~\ref{sec:solar:gci}.

\subsection{RQ2: Automatically Identifying 
``Problematic'' Reflective Calls}
\label{sec:rq2}

If \solar is scalable in analyzing a program
with no unsoundly resolved calls reported, then 
\solar is sound for this program 
under Assumptions~\ref{ass:world} -- \ref{ass:reach} 
(Theorem~\ref{thm:soundness}).
Thus, as discussed in Section~\ref{sec:rq1}, \solar 
is sound (in this sense) 
for the seven Java programs scalably
analyzed (under 3 hours) with full automation in our experiment.

\solar is unscalable for
\texttt{hsqldb}, \texttt{xalan} and 
\texttt{checkstyle} (under 3 hours). \probe is
then run to identify their ``problematic'' 
reflective calls, reporting 13 potentially
unsoundly resolved calls: 1 in \texttt{hsqldb}, 12 in 
\texttt{xalan} and 0 in \texttt{checkstyle}. 
Their
handling is \emph{all unsound} by code inspection, 
demonstrating the effectiveness of \solar in \emph{accurately}
pinpointing a small number of right parts of
the program
to improve unsoundness.

In addition, we currently adopt
a simple approach to alerting users for
potentially imprecisely resolved reflective calls.
\probe sorts all the
\texttt{newInstance()} call sites
according to the number of objects lazily created
at the cast operations operating on the 
result of a \texttt{newInstance()} call 
(by \rulenameT{L-Cast}) 
in non-increasing order. In addition, \probe ranks the
remaining reflective call sites (of other side-effect methods) according to 
the number of reflective
targets resolved, also in non-increasing order.

By focusing on 
unsoundly and imprecisely resolved reflective 
calls as opposed to the unknown input strings (see Section~\ref{sec:rq1}), only lightweight
annotations are needed as shown in
Figure~\ref{fig:ann}, with
2 \texttt{hsqldb}, 2 \texttt{xalan} and 3 in
\texttt{checkstyle} as explained below.

\paragraph{Probing \texttt{\emph{hsqldb}}}

Earlier, Figure~\ref{fig:output} shows the output automatically generated for \texttt{hsqldb} by
\probe (highlighting which reflective calls are resolved unsoundly or imprecisely), together with the suggested annotation sites (as introduced in Section~\ref{sec:anno}). In Figure~\ref{fig:output}, all the call sites
to (invoke) the same method are numbered from 0.

The unsound list contains one \texttt{invoke()}, with its
relevant code appearing in class 
\texttt{org.hsqldb.Function} as shown. After \probe
has finished, \texttt{mMethod} in line~352 points to a 
\texttt{Method} object \texttt{m}$^u_u$ that is
initially created in line 179 
and later flows into line 182, 
indicating that the class type of $\texttt{m}_u^u$ is unknown since
\texttt{cn} in line 169 is unknown.
By inspecting the code, we find that \texttt{cn} can only be
\texttt{java.lang.Math} or \texttt{org.hsqldb.Library}, read
from some hash maps or obtained by string manipulations
and is annotated as:


\vspace{2ex}
\begin{BVerbatim}[baselinestretch=1.5]
org.hsqldb.Function.<init>/java.lang.Class.forName/0 java.lang.Math
org.hsqldb.Function.<init>/java.lang.Class.forName/0 org.hsqldb.Library
\end{BVerbatim}
\vspace{2ex}

The imprecise list for \texttt{hsqldb} is divided into 
two sections. 
In ``\emph{newInstance (Type Casting)}'', there are 10 listed cast operations
$(T)$ reached by an $o_i^u$ object such that the number
of types inferred from $T$ is larger than~10. The top cast \texttt{java.io.Serializable} has
$1391$ subtypes and is marked to be reached by 
a \texttt{newInstance()} call site in
\texttt{java.io.ObjectStreamClass}. However, this is a false
positive for the harness used
due to imprecision in pointer analysis. Thus,
we have annotated its
corresponding \texttt{Class.forName()} call site
in method \texttt{resolveClass} of 
class \texttt{java.io.ObjectInputStream} to return nothing. 
With the two annotations, \solar terminates in 45
minutes with its unsound list being empty.

\paragraph{Probing \texttt{\emph{xalan}}}
\probe reports 12 unsoundly resolved
\texttt{invoke()} call sites. 
All \texttt{Method}
objects flowing into these call sites are created 
at two \texttt{getMethods()} call sites in class 
\texttt{extensions.MethodResolver}. 
By inspecting the code, we find that the
string arguments for 
the two \texttt{getMethods()} calls and their
corresponding entry methods
are all read from a file with its name hard-wired as 
\texttt{xmlspec.xsl}. For this 
input file provided by DaCapo, these 
two calls are never executed and thus annotated to be 
disregarded.  With these two annotations,
\solar terminates in 28
minutes with its unsound list being empty.

\paragraph{Probing \texttt{\emph{checkstyle}}}
\label{sec:eval-checkstyle}

\probe reports no unsoundly resolved call.
To see why~\solar is unscalable, 
we examine one
\texttt{invoke()} call in line 1773 of Figure~\ref{fig:check} found automatically by
\probe 
that stands out as being
possibly imprecisely resolved. 

\begin{figure}[thp]
\centering
\includegraphics[width=0.9\textwidth]{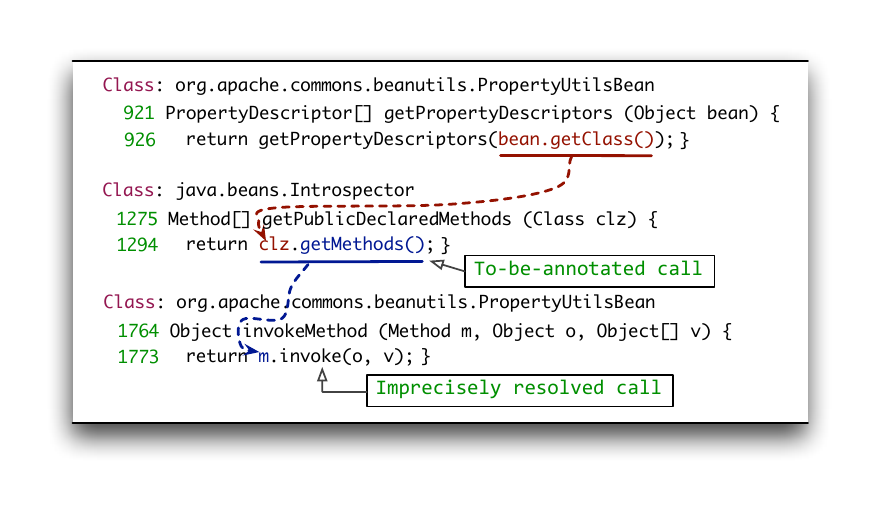}
\vspace{-5ex}
\caption{Probing  \texttt{checkstyle}.
\label{fig:check}}
\end{figure}

There are 962 target methods inferred 
at this call site. \probe highlights
its corresponding member-introspecting method 
\texttt{clz.getMethods()}
(in line 1294) and its entry methods (with
one of these being shown in line 926). Based on this, we find easily
by code inspection that the target 
methods called reflectively at the \texttt{invoke()}
call are the setters whose names share 
the prefix ``\texttt{set}''. As a result, the
\texttt{clz.getMethods()} call is 
annotated to return 158 ``\texttt{setX}'' methods 
in all the subclasses of \texttt{AutomaticBean}.

In addition, the
\texttt{Method} objects created at
one \texttt{getMethods()} call and one 
\texttt{getDeclaredMethods()} call in class 
\texttt{*.beanutils.MappedPropertyDescriptor\$1} 
flow into the \texttt{invoke()} call in line 1773
as false positives due to imprecision in the
pointer analysis. These \texttt{Method} 
objects have been annotated away.

After the three annotations,
\solar is scalable, terminating in 38 minutes.

\subsection{RQ3: Soundness, Precision and Efficiency}
\label{sec:rq3}

In this section, we first compare the soundness of \doopr, \elf and \solar and 
then compare their analysis precision and efficiency 
to see if \solar's precision and efficiency are 
still acceptable while being significantly more sound than the prior work.

\subsubsection{Soundness}
\label{dynamic}
To compare the soundness of 
\doopr, \elf and \solar, it is the
most relevant to compare their \emph{recall}, measured
by the number of true reflective targets discovered at reflective calls 
to side-effect methods that are dynamically executed
under certain inputs. 
Unlike \doopr and \elf,
\solar can automatically identify 
``problematic'' reflective calls for
lightweight annotations. To ensure a fair comparison,
the three annotated programs shown in 
Figure~\ref{fig:ann} are used by all the three
reflection analyses.

\paragraph{Recall}

Table~\ref{table:recall} compares the recall
of \doopr, \elf and \solar.
We use \TamiFlex~\cite{Bodden11}, a practical 
dynamic reflection analysis tool
to find the targets accessed at reflective
calls in our programs
under the inputs described in Section~\ref{sec:eval-ass}.
Of the three analyses compared,
\solar is the only one to achieve
total recall, for all reflective targets, including
both methods and fields, accessed.

For each DaCapo benchmark, its main code body is 
run under a reflection-based harness. As a result,
static analysis tools, 
including \doopr, \elf and \solar,  use its
\texttt{xdeps} version driven by 
a reflection-free harness. 
However, \TamiFlex needs to run each 
DaCapo benchmark with the reflection-based harness 
dynamically. For each DaCapo
benchmark, two harnesses lead to 
different versions used for a 
few classes (e.g., \texttt{org.apache.xerces.parsers.SAXParser}) in \texttt{eclipse}, \texttt{fop}, \texttt{pmd} and \texttt{xalan}. 
In Table~\ref{table:recall}, we thus ignore the totally
21 extra reflective targets accessed by \TamiFlex. 

\begin{table}[h]
\tbl{Recall comparison. For each program,
	\TamiFlex denotes the number of targets found by \TamiFlex and
	Recall denotes the number of such (true) targets 
also	discovered by each reflection analysis. 
There are two types of \texttt{newInstance()}: 
\texttt{c.newInstance} (in \texttt{Class}) and
\texttt{ct.newInstance} (in \texttt{Constructor}).}
{
\renewcommand\arraystretch{1.3}
\setlength{\tabcolsep}{2.5pt}

\scalebox{.95}{
\begin{tabular}{c|l|ccc|ccc|ccc|ccc}

\Xhline{1pt}
\multicolumn{2}{c|}{\multirow{2}[4]{*}{}} & \multicolumn{3}{c|}{\texttt{c.newInstance}} & \multicolumn{3}{c|}{\texttt{ct.newInstance}} & \multicolumn{3}{c|}{\texttt{invoke}} & \multicolumn{3}{c}{\texttt{get/set}}\\
\cline{3-14}\multicolumn{2}{c|}{} & \doopr  & \elf   & \solar & \doopr  & \elf   & \solar & \doopr  & \elf   & \solar & \doopr  & \elf   & \solar\\
\Xhline{1pt}

\multirow{2}[4]{*}{chart} & Recall & 13    & 21    & 22    & 2     & 2     & 2     & 2     & 2     & 2     & 8     & 8     & 8 \\
\cline{2-14}      & 	\TamiFlex &       & 22    &       &       & 2     &       &       & 2     &       &       & 8     &  \\
\Xhline{1pt}
\multirow{2}[4]{*}{eclipse} & Recall & 9     & 24    &  37    & 1     & 3     & 3     & 2     & 7     & 7     & 8     & 1039  & 1039 \\
\cline{2-14}      & 	\TamiFlex &   & 37   &       &       & 3     &       &       & 7   &      &       & 1039  &  \\
\Xhline{1pt}
\multirow{2}[4]{*}{fop} & Recall & 9     & 12    &13    & 0     & 0     & 0     & 1     & 1     & 1     & 8     & 8     & 8 \\
\cline{2-14}      & 	\TamiFlex & & 13    &    &       & 0     &       &       & 1     &       &       & 8     &  \\
\Xhline{1pt}
\multirow{2}[4]{*}{hsqldb} & Recall & 5     & 9     & 10    & 1     & 1     & 1     & 0     & 0     & 0     & 8     & 8     & 8 \\
\cline{2-14}      & 	\TamiFlex &       & 10    &       &       & 1     &       &       & 0     &       &       & 8     &  \\
\Xhline{1pt}
\multirow{2}[4]{*}{pmd} & Recall & 4     & 8     & 13    & 0     & 0     & 0     & 0     & 0     & 7     & 8     & 8     & 8 \\
\cline{2-14}      & 	\TamiFlex &     & 13  &    &       & 0     &       &       & 7     &       &       & 8     &  \\
\Xhline{1pt}
\multirow{2}[4]{*}{xalan} & Recall & 36    & 42    &43    & 0     & 0     & 0     & 32    & 5     & 32    & 8     & 8     & 8 \\
\cline{2-14}      & 	\TamiFlex &   & 43   &     &       & 0     &       &       & 32    &       &       & 8     &  \\
\Xhline{1pt}
\multirow{2}[4]{*}{avrora} & Recall & 50    & 46    & 53    & 0     & 0     & 0     & 0     & 0     & 0     & 0     & 0     & 0 \\
\cline{2-14}      & 	\TamiFlex &       & 53    &       &       & 0     &       &       & 0     &       &       & 0     &  \\
\Xhline{1pt}
\multirow{2}[4]{*}{checkstyle} & Recall & 7     & 9     & 72    & 1     & 1     & 24    & 1     & 5     & 28    & 8     & 8     & 8 \\
\cline{2-14}      & 	\TamiFlex &       & 72    &       &       & 24    &       &       & 28    &       &       & 8     &  \\
\Xhline{1pt}
\multirow{2}[4]{*}{findbugs} & Recall & 6     & 10    & 15    & 8     & 8     & 115   & 1     & 1     & 1     & 8     & 8     & 8 \\
\cline{2-14}      & 	\TamiFlex &       & 15    &       &       & 115   &       &       & 1     &       &       & 8     &  \\
\Xhline{1pt}
\multirow{2}[4]{*}{freecs} & Recall & 6     & 10    & 12    & 2     & 2     & 2     & 1     & 1     & 55    & 8     & 8     & 8 \\
\cline{2-14}      & 	\TamiFlex &       & 12    &       &       & 2     &       &       & 55    &       &       & 8     &  \\
\Xhline{1pt}
\multirow{2}[4]{*}{Total} & Recall & 141   & 191   &290   & 15    & 17    & 147   & 40    & 22    & 133   & 72    & 1103  & 1103 \\
\cline{2-14}      & 	\TamiFlex &    & 290   &    &       & 147   &       &   &133  &     &       & 1103  &  \\
\Xhline{1pt}
\end{tabular}%
} 
}
\label{table:recall}
\end{table}

In practice, a reflection analysis must
handle \texttt{newInstance()} and \texttt{invoke()} 
well in order to build the  call graph for a program. 
Let us see how resolving more (true) reflective targets 
in a program by \solar can affect the call graph of the
program.

\begin{figure}[t]
\centering
\includegraphics[width=0.86\textwidth]{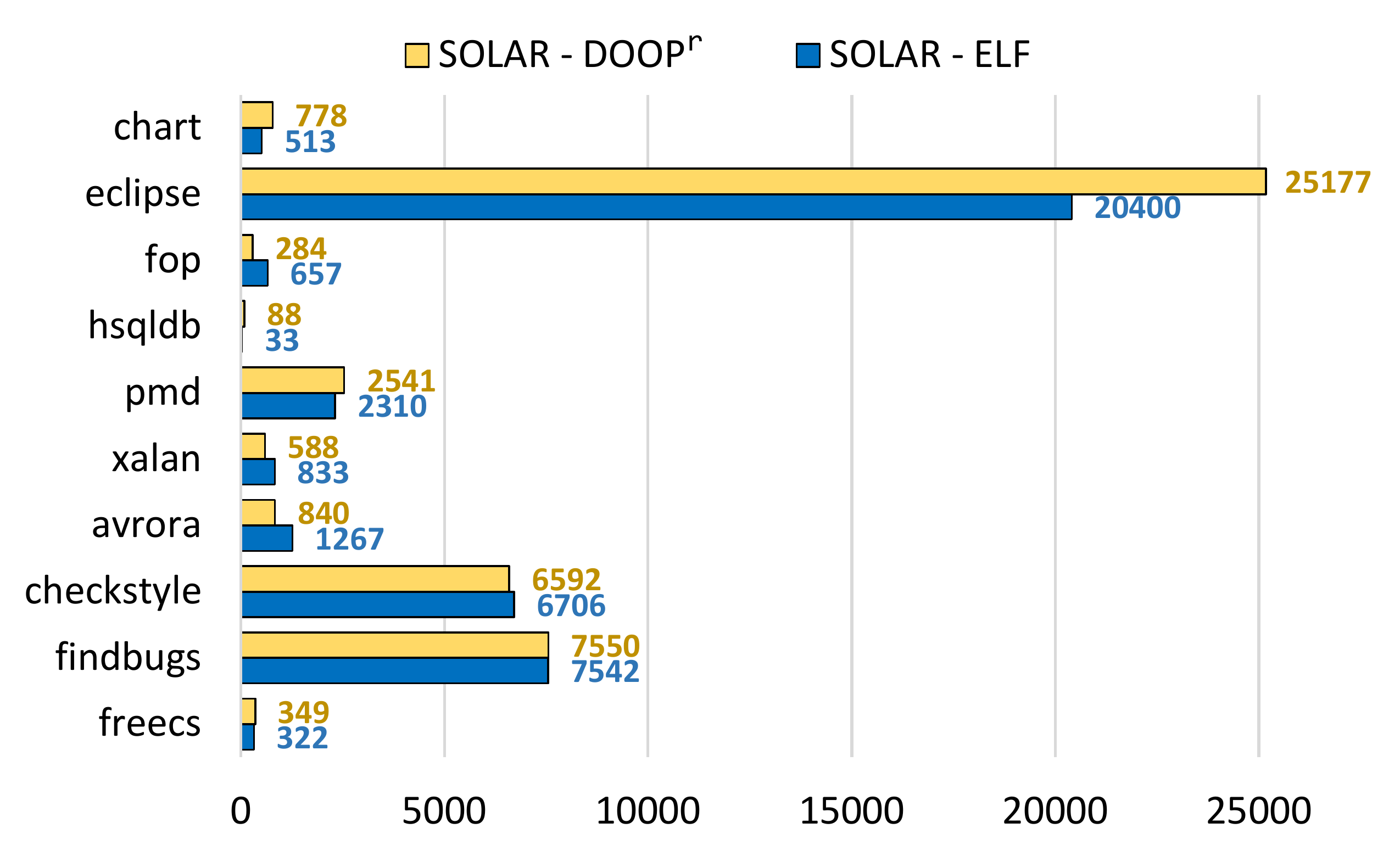}
\caption{More true 
caller-callee relations discovered in recall
by \solar
than \doopr, denoted
$\emph{\solar}\!-\!\emph{\doopr}$, and 
by \solar than \elf, denoted
$\emph{\solar}\!-\!\emph{\elf}$. 
\label{fig:mtdcall}}
\end{figure}

\paragraph{Call Graph}
Figure~\ref{fig:mtdcall} compares \doopr, \elf and \solar
in terms of true caller-callee relations discovered.
These numbers are statically calculated and obtained by
an instrumental tool written on top of
{\sc Javassist} \cite{Chiba00}.  
According to 
Table~\ref{table:recall},
\solar recalls a total of 191\% (148\%) more targets 
than \doopr (\elf) at the calls to 
\texttt{newInstance()} and \texttt{invoke()}, translating
into
44787 (40583) more \emph{true} caller-callee relations found
for the 10 programs under
the inputs described in Section~\ref{sec:eval-ass}.
These numbers are expected to improve 
when more inputs are used. 
Note that all method targets recalled by
\doopr and \elf are recalled by \solar so that we can 
use the ``subtraction'' (i.e., \solar$-\,$\doopr and \solar$-\,$\elf) 
in Figure~\ref{fig:mtdcall}.

Let us examine \texttt{eclipse} and
\texttt{checkstyle}. We start with \texttt{eclipse}. 
According to Table~\ref{table:recall}, \solar 
finds 35 more target methods than \doopr, causing
25177 more true caller-callee relations to be discovered.
Due to LHM, \solar finds 13 more target methods at
\texttt{newInstance()} calls than
\elf, resulting in 20400 additional true caller-callee
relations to be introduced.
Let us consider now \texttt{checkstyle}.
Due to
collective inference, \solar has resolved 23 more
target methods at \texttt{invoke()} calls than 
\elf, resulting
in 2437 more true caller-callee relations (i.e., over
a third of the total number of such relations, 6706) to 
be discovered.




\subsubsection{Precision}
\label{sec:precision}

Table~\ref{table:precision} compares the
analysis precision of \doopr, \elf and \solar
with an important client, devirtualization, which
is only applicable to virtual
calls with one single target each.
This client is critical in call graph construction 
and thus often used for
measuring the precision of a pointer analysis.
Despite achieving
better recall (Table~\ref{table:recall}), which results in  
more true caller-callee relations to be discovered 
(Figure~\ref{fig:mtdcall}), \solar maintains a similar
precision as \doopr and \elf for this client.

\begin{table}[t]
\tbl{Precision comparison. For each program,
	the percentage of 
virtual calls that can be devirtualized is given
by each reflection analysis. Devirtualization
is an important client for call graph construction.
}
{
\scalebox{0.93}{
\small
\setlength{\tabcolsep}{0.8ex}
\begin{tabular}{cIcIcIcIcIcIcIcIcIcIcIc}
\Xhline{1pt}
\rowcolor[rgb]{ .851,  .851,  .851} & chart & eclipse & fop   & hsqldb & pmd   & xalan & avrora & checkstyle & findbugs & freecs & Average \bigstrut\\
\Xhline{1pt}
  \doopr  & 93.40 & 94.69 & 90.73 & 95.34 & 92.53 & 92.63 & 91.98 & 93.09 & 92.41 & 95.27 & 93.20 \bigstrut[t]\\
   \elf   & 93.53 & 88.07 & 92.34 & 94.80 & 92.87 & 92.70 & 94.50 & 93.19 & 92.53 & 94.94 & 92.93 \\
 \solar & 93.51 & 87.69 & 92.26 & 94.51 & 92.39 & 92.65 & 92.43 & 93.39 & 92.37 & 95.26 & 92.63\\
\Xhline{1pt}
\end{tabular}%
}
}
\label{table:precision}
\end{table}

In \solar, its lazy heap modeling (LHM) relies on cast 
types (and their subtypes) to infer
reflective targets. Will this inference strategy
introduce too much imprecision into the analysis
when a cast type has many subtypes?
To address this concern, we conduct an experiment about the percentage distribution for the number of types inferred at cast-related LHM points.
As shown in Figure~\ref{fig:cast},
the number of inferred types in a program is 1
($\leqslant\!10$) at
30.8\% (85.4\%) of its cast-related \dam points.
Some types
(e.g., \texttt{java.io.Serializable}) 
are quite wide, giving rise to more than 50 inferred types each,
but rare, appearing
at an average of about 1.9\% cast-related \dam points
in a program.

\begin{figure}[htp]
\centering
\includegraphics[width=0.8\textwidth]{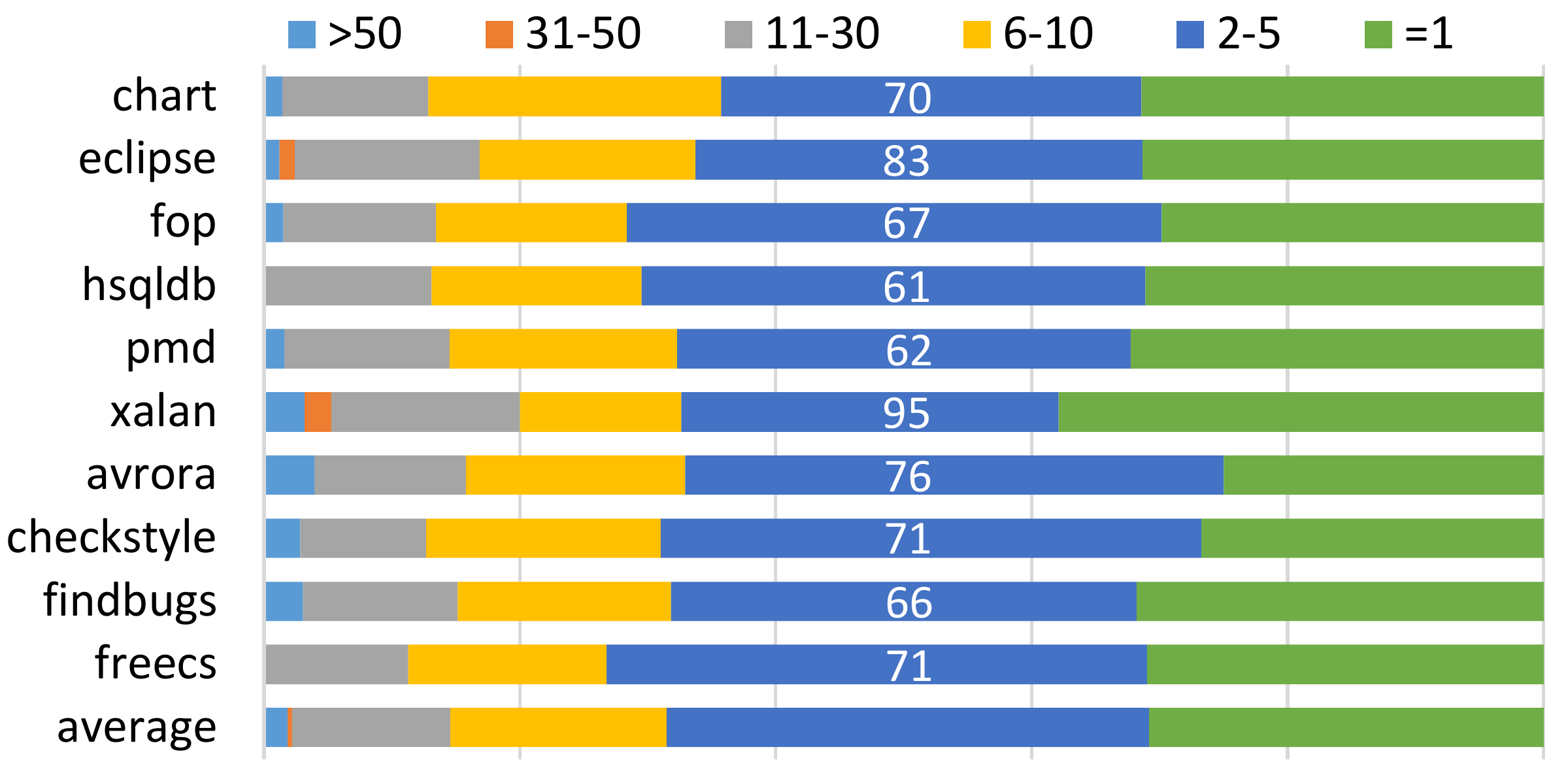}
\caption{
Percentage distribution for the number of types
inferred at cast-related LHM points. For each program,
the total number of the cast-related LHM points is shown 
in the middle of its own bar.
\label{fig:cast}}
\end{figure}

\begin{table}[h]
\tbl{Comparing \solar, \elf and \doopr in terms of analysis times  (secs).}
{
\small
\setlength{\tabcolsep}{0.5ex}
\begin{tabular}{cIcIcIcIcIcIcIcIcIcIcIc}
\Xhline{1pt}
\rowcolor[rgb]{ .851,  .851,  .851}       & chart & eclipse & fop   & hsqldb & pmd   & xalan & avrora & checkstyle & findbugs & freecs & Average \bigstrut\\
\Xhline{1pt}
\doopr & 1864  & 340   & 2631  & 1173  & 602   & 1092  & 2934  & 531   & 2753  & 371   & 1074 \bigstrut[t]\\
\elf & 3434  & 5496  & 2821  & 1765  & 1363  & 1432  & 932   & 1463  & 2281  & 1259  & 1930 \\
\solar & 4543  & 10743 & 4303  & 2695  & 2156  & 1701  & 3551  & 2256  & 8489  & 2880  & 3638 \\
\Xhline{1pt}
\end{tabular}%
}
\label{table:perform}
\end{table}

\subsubsection{Efficiency}
\label{sec:efficiency}

Table~\ref{table:perform} compares the
analysis times of \doopr, \elf and \solar.
Despite producing significantly better
soundness than \doopr and \elf, \solar
is only several-fold slower. 
When analyzing \texttt{hsqldb}, \texttt{xalan} and 
\texttt{checkstyle}, \solar requires some 
lightweight annotations. Their analysis times
are the ones consumed by \solar on
analyzing the annotated programs.
Note that these annotated programs are also
used by \doopr and \elf (to ensure a fair
comparison).

%
%

In summary,
the experimental results described
in Section~\ref{sec:rq3} demonstrate that \solar's soundness-guided
design is effective in balancing
soundness, precision and scalability,
in practice: \solar is able to achieve significantly better soundness than the state-of-the-art while being still reasonably precise and efficient.

\section{Related Work}
\label{sec:related}

\subsection{Static Reflection Analysis}
\label{rel:static}
\citeN{Livshits05} introduce the first static reflection analysis for Java. By interleaving with a pointer
analysis, this reflection analysis discovers constant
string values by performing regular string inference
and infers the types of reflectively created objects
at \texttt{clz.newInstance()} calls, where \texttt{clz}
is unknown, by exploiting their intra-procedurally
post-dominating cast operations, if any.
Many modern pointer analysis frameworks such as \doop~\cite{Yannis09}, \wala~\cite{wala} and Chord~\cite{Mayur06}, adopt a similar approach to analyze Java reflection, and many subsequent reflection analyses~\cite{Yue14,Yue15,Yannis15,Zhang17}, including \solar, are also
inspired by the same work.


Our \elf tool~\cite{Yue14} represents the first reflection analysis that 
takes advantage of the self-inferencing property (Definition~\ref{def:self}) to handle Java reflection. \elf can be considered as performing the collective inference in \solar's inference engine except that \elf's 
inference rules are more strict (for better analysis precision). In \elf, a reflective target will not be resolved unless both a red circle and a blue circle in Figure~\ref{fig:collect} (i.e., a class name and part of a 
member signature) are available.
As a result, \elf is usually more precise and efficient, but (much) less sound than \solar.

In \doopr~\cite{Yannis15} (the one 
compared with \solar in Section~\ref{sec:eval}), the authors propose to leverage partial string information to resolve reflective targets. As explained in Section~\ref{study:sec:string}, many string arguments are 
generated through complex string manipulations, causing their values 
to be hard to resolve statically. However, in some cases, its \emph{substring analysis} makes it possible to leverage
some partially available string values, in the form of
self-inferencing hints, to help infer reflective targets.  
Combining this with some sophisticated string analyses~\cite{moller03} is expected to generate more effective results, which is worth pursuing in future work. 

To improve soundness, \doopr attempts to
infer the class types at \texttt{Class.forName()} (rather than \texttt{newInstance()}) call 
sites by leveraging
both (1) the intra- and inter-procedural cast operations
that are not necessarily post-dominating at their
corresponding \texttt{newInstance()} call sites
and (2) the member names at their corresponding
member-introspecting call sites (i.e., the
blue circles formed by \emph{Target Propagation} in Figure~\ref{fig:collect}).
However, these two strategies may make the analysis 
imprecise. For example,
the second strategy may suffer from a precision loss 
when different unrelated classes contain identically-named 
members~\cite{Yue14}. 
In addition, 
in both strategies (1) and (2), the class types (at a \texttt{Class.forName()} call site), which are back-propagated (inferred) from some member-introspecting call sites (cast sites), may further pollute the analysis precision at the other member-introspecting and side-effect calls. 
To reduce such imprecision, an approach called \emph{inventing objects} is proposed in \doopr. It creates objects at the \texttt{newInstance()} call sites (according to the types at the related cast sites), rather than the corresponding \texttt{Class.forName()} call sites. This method is similar to Case (II) in \solar's lazy heap modeling (Figure~\ref{fig:lhm}).

\citeN{Ernst15} analyze Java reflection and intents in Android apps, by reasoning about soundness from the
perspective of a given client.
Consider a client for detecting sensitive data leaks.
Any \texttt{invoke()} call that cannot be resolved 
soundly is assumed to return conservatively 
some sensitive data. While sound for the client,
the reflection analysis itself is still unsound. 
In terms of reflection inference, \citeN{Ernst15} exploit only a subset
of self-inferencing hints used in \solar.
For example, how the arguments in an \texttt{invoke()} 
call are handled.
\citeN{Ernst15} uses only their arity but
ignore their types. In contrast, \solar takes both into
account. There are precision implications.
Consider the code fragment for
\texttt{Eclipse} in Figure~\ref{study:fig:invoke}. If the types of the elements of its one-dimensional
array, \texttt{parameters}, are ignored,
many spurious target methods in line 175
would be inferred, as single-argument methods are common.

Despite recent advances on reflection analysis~\cite{Yue14,Yue15,Yannis15}, a sophisticated
reflection analysis does not co-exist well with a sophisticated pointer analysis, since the latter
is usually unscalable for large programs~\cite{Livshits05,Yue14,Yue15,Yannis15}. 
If a scalable but imprecise pointer analysis
is used instead, the reflection analysis may introduce many spurious call graph edges,
making its underlying client applications to be too imprecise to be practically useful.
This problem can be alleviated by using a recently proposed program slicing approach, called program tailoring~\cite{Yue16}. 

Briefly, program tailoring accepts a sequential criterion (e.g., a sequence of reflective call sites: \texttt{forName()}$\rightarrow$\texttt{getMethod()}$\rightarrow$\texttt{invoke()}) and generates a soundly
tailored program that contains the statements 
only relevant to the given sequential criterion. In other words, the tailored program comprises the statements in all possible execution paths passing through the 
sequence in the given order. As a result, a more precise (but less scalable) pointer analysis may be scalable when the tailored program (with smaller size) is analyzed, resulting a more precise result resolved at the given reflective call site (e.g., the above \texttt{invoke()} call site). These reflective call sites can be the problematic ones generated by \solar as demonstrated in~\cite{Yue16}.

\textsc{DroidRa}~\cite{Lili16} analyzes reflection in
Android apps by propagating
string constants. This is similar to the target propagation in \solar's collective inference 
except that \textsc{DroidRa} uses the 
solver \textsc{Coal}~\cite{Octeau15} to resolve the string values
in a context- and flow-sensitive manner. Currently,
\solar's target propagation is context-sensitive only, 
achieved by the underlying pointer analysis used.

\textsc{Ripple}~\cite{Zhang17} resolves reflection in
Android apps by tackling their incomplete information environments (e.g., undetermined intents or unmodeled services).
Technically, even if some data flows are missing (i.e, \texttt{null}), it is still able to resolve reflective targets by performing type inference, which is similar 
to \solar's collective inference except that some sophisticated inference rules are not supported.

Unlike
the prior work~\cite{Livshits05,Yue14,Yannis15,Ernst15,Lili16,Zhang17},
as highlighted in Figure~\ref{fig:position}, 
\solar is capable of reasoning about its
soundness and accurately identifying its unsoundness. 

\subsection{Dynamic Reflection Analysis}
\label{rel:dynamic}

\citeN{Hirzel07} propose an online pointer
analysis for handling some dynamic 
features in Java at run time. To tackle reflection, 
their analysis instruments a program
so that constraints are generated dynamically when 
the injected code is triggered at run time. 
Thus, the
points-to information is incrementally updated 
when new constraints are gradually introduced by reflection. 
This technique on reflection handling can be used in 
JIT optimizations but may not be suitable for
whole-program static analysis.

To facilitate static analysis, \citeN{Bodden11} 
leverage the runtime information gathered
at reflective calls. 
Their tool, \TamiFlex, records usage information 
of reflective calls 
in the program at run-time, interprets the logging 
information, and finally, transforms these reflective calls into regular Java method calls. 
In addition, \TamiFlex
inserts runtime checks to warn the user in cases that
the program encounters reflective calls that
diverge from the recorded information of previous runs.

\textsc{Harvester}~\cite{Bodden16} is designed to automatically extract runtime values from Android
applications. It uses a variation of traditional program slicing and dynamic execution to extract
values from obfuscated malware samples that obfuscate method calls using reflection or hide sensitive values in native code.

\subsection{Others}
\label{rel:others}

\citeN{Braux00} provide offline
partial evaluation support for reflection in order to 
perform aggressive compiler optimizations for Java 
programs. It transforms a program by compiling away 
the reflection code into regular operations on objects 
according to their concrete types that
are constrained manually. The inference engine of \solar can be viewed as a tool 
for inferring such constraints automatically. 

To increase code coverage, some static analysis tools 
\cite{Yannis09,bddbddb} allow users to provide 
ad hoc manual specifications about reflection usage 
in a program. However, due to the diversity and 
complexity of applications, it is not yet clear how 
to do so in a systematic manner. For
framework-based web applications,
Sridharan et al.~\cite{Manu11} introduce a
framework that exploits 
domain knowledge to automatically generate a 
specification of framework-related behaviours
(e.g., reflection usage) by processing both application 
code and configuration files. \solar may 
also utilize domain knowledge to analyze 
some configuration files, but 
only for those reflective call sites 
that cannot be resolved effectively.

\citeN{Li17} propose an object-oriented dynamic
symbolic execution framework for testing 
object-oriented libraries in Java.
They support polymorphism by introducing
constraints for method invocation targets and 
field manipulations via symbolic types.
They have also generalized the notion of symbolic types
to symbolic methods and symbolic fields to handle Java reflection symbolically.

Finally, the dynamic analyses \cite{Bodden11,Hirzel07,Bodden16} 
work in the presence of both dynamic
class loading and reflection. 
Nguyen and Xue~\cite{Xue05,Xue05cc} introduce
an inter-procedural side-effect analysis for open-world
Java programs by allowing dynamic class
loading but disallowing reflection.
Like other static reflection analyses 
\cite{Livshits05,Yue14,Yannis15,Ernst15,Lili16,Zhang17}, \solar can presently
analyze closed-world Java programs only.

\section[Conclusions]{Conclusions}
\label{sec:conclude}

Reflection analysis is a long-standing hard open problem. In the past years, almost all the research papers consider Java reflection as a separate assumption and most static analysis tools either handle it partially or totally ignore it. In the program analysis community, people hold to the
common opinion that ``\emph{reflection is a dynamic feature, so how could it be handled effectively in static analysis?}'' 
This paper aims to change such informed opinion by showing that effective static analysis for handling Java reflection is feasible. 
We present a comprehensive understanding of Java reflection by illustrating its concept (origin), interface design and real-world usage. Many useful findings are
given to guide the design of more effective reflection analysis approaches and tools.
We introduce \solar, a soundness-guided reflection analysis, which is able to achieve significantly better soundness than previous work and can also accurately and automatically identify which reflective calls are resolved unsoundly or imprecisely. 
Such capabilities are helpful, in practice,
as users can be aware of how sound an analysis is. In addition, for some clients where high-quality 
analysis results are needed (e.g., bug detectors
with good precision or verification tools
with good soundness) for a program, \solar provides an opportunity
to help users obtain the analysis results as desired,
by guiding them to add some lightweight 
annotations, if needed, to the parts of the program,
where unsoundly or imprecisely resolved reflective
calls are identified.


\bibliographystyle{ACM-Reference-Format-Journals}
\bibliography{citation}

\end{document}